\documentclass[prx,floatfix,amsmath,twocolumn,notitlepage,superscriptaddress,nofootinbib]{revtex4-2}

\usepackage{amssymb}
\usepackage{graphicx}
\usepackage{graphics}
\usepackage{amsmath}
\usepackage{amsthm}
\usepackage{color}
\usepackage{dsfont}
\usepackage{float}
\usepackage{mathtools}
\usepackage[colorlinks=true, linkcolor=blue, citecolor=red]{hyperref}
\usepackage{verbatim}
\usepackage{xcolor}
\usepackage{array}
\usepackage{hhline}
\usepackage{setspace}

\newcolumntype{L}[1]{>{\raggedright\let\newline\\\arraybackslash\hspace{0pt}}p{#1}}

\newcommand{\mbf}{\mathbf}
\newcommand{\mbb}{\mathbb}
\newcommand{\mc}{\mathcal}
\renewcommand{\ol}{\overline}

\newcommand{\tr}{\textrm{Tr}}
\newcommand{\id}{\textrm{id}}
\newcommand{\ket}[1]{|#1\rangle}
\newcommand{\bra}[1]{\langle #1|}
\newcommand{\ip}[2]{\langle#1| #2\rangle}
\newcommand{\op}[2]{|#1\rangle\langle #2|}

\newcommand{\1}{\mathds{1}}
\newcommand{\msf}{\mathsf}
\newcommand{\ML}{\text{ML}}
\newcommand{\wt}{\widetilde}

\DeclareMathOperator{\rank}{rank}
\DeclareMathOperator{\conv}{conv}

\definecolor{cool_green}{rgb}{0.0, 0.5, 0.0}

\theoremstyle{definition}

\newtheorem{definition}{Definition}

\newtheorem{lemma}{Lemma}
\newtheorem{proposition}{Proposition}
\newtheorem{theorem}{Theorem}

\newtheorem*{remark}{Remark}

\begin{document}

\title{Certifying the Classical Simulation Cost of a Quantum Channel}

\author{Brian Doolittle and Eric Chitambar}

\date{\today}

\begin{abstract}
    A fundamental objective in quantum information science is to determine the cost in classical resources of simulating a particular quantum system.
    The classical simulation cost is quantified by the signaling dimension which specifies the minimum amount of classical communication needed to perfectly simulate a channel's input-output correlations when unlimited shared randomness is held between encoder and decoder.
    This paper provides a collection of device-independent tests that place lower and upper bounds on the signaling dimension of a channel.  Among them, a single family of tests is shown to determine when a noisy classical channel can be simulated using an amount of communication strictly less than either its input or its output alphabet size.
    In addition, a family of eight Bell inequalities is presented that completely characterize when any four-outcome measurement channel, such as a Bell measurement, can be simulated using one communication bit and shared randomness.
    Finally, we bound the signaling dimension for all partial replacer channels in $d$ dimensions.
    The bounds are found to be tight for the special case of the erasure channel.
\end{abstract}

\maketitle

\section{Introduction}

The transmission of quantum states between devices is crucial for many quantum network protocols.
In the near-term, quantum memory limitations will restrict quantum networks to ``prepare and measure'' functionality \cite{wehner2018quantum}, which allows for quantum communication between separated parties but requires measurement immediately upon reception.
Prepare and measure scenarios exhibit quantum advantages for tasks that involve distributed information processing \cite{BuhrmanComplexity2010} or establishing nonlocal correlations which cannot be reproduced by bounded classical communication and shared randomness \cite{Vicente2017shared}. These nonlocal correlations lead to quantum advantages in random access codes \cite{ambainis2008quantum,tavakoli2015}, randomness expansion \cite{li2011semi}, device self-testing \cite{tavakoli2018}, semi-device-independent key distribution \cite{pawlowski2011semi}, and dimensionality witnessing \cite{brunner2008testing,hendrych2012experimental}.

The general communication process is depicted in Fig. \ref{Fig:simulation}(a) with Alice (the sender) and Bob (the receiver) connected by some quantum channel $\mc{N}^{A\to B}$.  Alice encodes a classical input $x\in\mc{X}$ into a quantum state $\rho_x$ and sends it through the channel to Bob, who then measures the output using a positive-operator valued measure (POVM) $\{\Pi_y\}_{y\in\mc{Y}}$ to obtain a classical message $y\in\mc{Y}$.
The induced classical channel, denoted by $\mbf{P}_\mc{N}$, has transition probabilities
\begin{equation}
\label{Eq:channel-induce}
    P_{\mc{N}}(y|x)=\tr\Big[\Pi_y\mc{N}\big(\rho_x\big)\Big].
\end{equation}
\begin{figure}[b] 
    \centering
    \includegraphics[width=.48\textwidth]{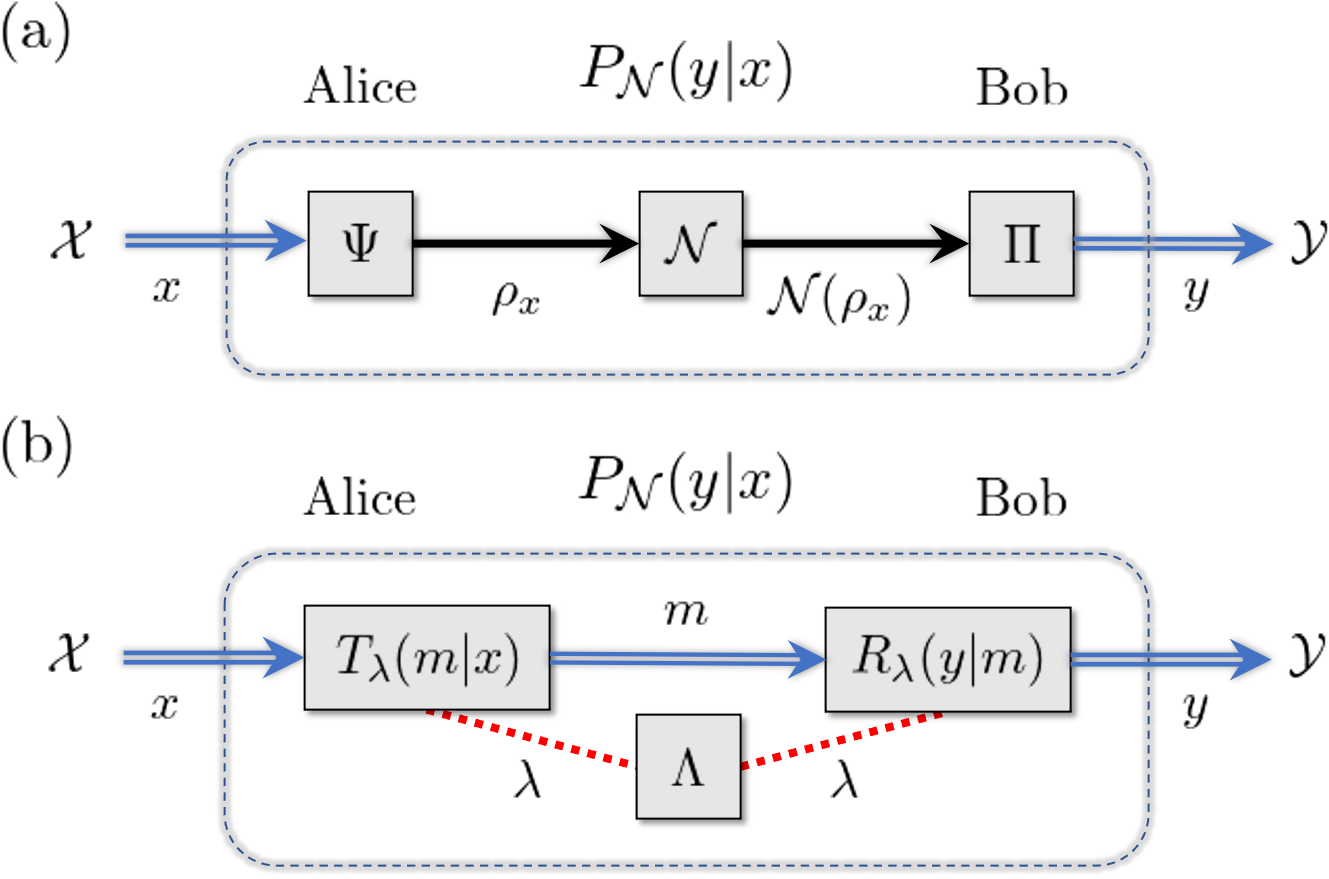}
 \caption{\linespread{1}\selectfont{\small
    A general classical communication process. We represent classical information as blue double lines, quantum information as black solid lines, and shared randomness as dotted red lines.  (a)
    A classical channel $\mbf{P}_{\mc{N}}$ is generated from a quantum channel $\mc{N}$ via Eq. \eqref{Eq:channel-induce}.
    A classical-quantum encoder $\Psi$ maps the classical input $x\in\mc{X}$ into a quantum state $\rho_x$.
    A quantum-classical decoder $\Pi$ implements POVM $\{\Pi_y\}_{y\in\mc{Y}}$
.    (b) Channel $\mbf{P}_{\mc{N}}$ is simulated using shared randomness and a noiseless classical channel via Eq. \eqref{Eq:d-simulate}.
    Alice encodes input $x$ into classical message $m$ with probability $T_\lambda(m|x)$ while Bob decodes message $m$ into output $y$ with probability $R_\lambda(y|m)$.
    The protocol is coordinated using a shared random value $\lambda$ drawn from sample space $\Lambda$ with probability $q(\lambda)$.
}}
 \label{Fig:simulation}
\end{figure}

\noindent A famous result by Holevo implies that the communication capacity of
$\mbf{P}_\mc{N}$ is limited by $\log_2d$, where $d$ is the input Hilbert space dimension of $\mc{N}$ \cite{holevo1973bounds}; hence a noiseless classical channel transmitting $d$ messages has a capacity no less than $\mbf{P}_\mc{N}$.

However, channel capacity is just one figure of merit, and there may be other features of a $\mbf{P}_{\mc{N}}$ that do not readily admit a classical simulation.
The strongest form of simulation is an exact replication of the transition probabilities $P_{\mc{N}}(y|x)$ for any set of states $\{\rho_x\}_{x\in\mc{X}}$ and POVM $\{\Pi_y\}_{y\in\mc{Y}}$.
This problem falls in the domain of zero-error quantum information theory \cite{Korner-1998a, Duan-2009a, Cubitt-2011b, Cubitt-2011a, Duan-2016a}, which considers the classical and quantum resources needed to perfectly simulate a given channel.
Unlike the capacity, a zero-error simulation of $\mbf{P}_\mc{N}$ typically requires additional communication beyond the input dimension of $\mc{N}$.
For example, a noiseless qubit channel $\id_2$ can generate channels $\mbf{P}_{\id_2}$ that cannot be faithfully simulated using a one bit of classical communication \cite{Vicente2017shared}.  

The simulation question becomes more interesting if ``static'' resources are used for the channel simulation \cite{Devetak-2004a, Devetak-2008a}, in addition to the ``dynamic'' resource of noiseless classical communication.  For example, shared randomness is a relatively inexpensive classical resource that Alice and Bob can use to coordinate their encoding and decoding maps used in the simulation protocol shown in Fig. \ref{Fig:simulation}(b).  Using shared randomness, a channel can be exactly simulated with a forward noiseless communication rate that asymptotically approaches the channel capacity; a fact known as the Classical Reverse Shannon Theorem \cite{Bennett2002_rev_shannon}.  More powerful static resources such as shared entanglement or non-signaling correlations could also be considered \cite{Cubitt-2011a, Wang-2020a, Fang-2020a}. 

While the Classical Reverse Shannon Theorem describes many-copy channel simulation, this work focuses on zero-error channel simulation in the single-copy case.  The minimum amount of classical communication (with unlimited shared randomness) needed to perfectly simulate every classical channel $\mbf{P}_{\mc{N}}$ having the form of Eq. \eqref{Eq:channel-induce} is known as the signaling dimension of $\mc{N}$ \cite{Dall'Arno_no_hypersignaling_2017}.
Significant progress in understanding the signaling dimension was made by Frenkel and Weiner who showed that every $d$-dimensional quantum channel requires no more than $d$ classical messages to perfectly simulate \cite{Frenkel2015}.
This result is a ``fine-grained'' version of Holevo's Theorem for channel capacity mentioned above.
However, the Frenkel-Weiner bound is not tight in general.
For example, consider the completely depolarizing channel on $d$ dimensions, $\mc{D}(\rho)=\mbb{I}/d$.
For any choice of inputs $\{\rho_x\}_x$ and POVM $\{\Pi_y\}_y$, the Frenkel-Weiner protocol yields a simulation of $\mbf{P}_\mc{D}$ that uses a forward transmission of $d$ messages.  However, this is clearly not optimal since $\mbf{P}_{\mc{D}}$ can be reproduced with no forward communication whatsoever; Bob just samples from the distribution $P(y)=\tr[\Pi_y]/d$.
A fundamental problem is then to understand when a noisy classical channel sending $d$ messages from Alice to Bob actually requires $d$ noiseless classical messages for zero-error simulation.  As a main result of this paper, we provide a family of simple tests that determine when this amount of communication is needed.  In other words, we characterize the conditions in which the simulation protocol of Frenkel and Weiner is optimal for the purposes of sending $d$ messages over a $d$-dimensional quantum channel.

This work pursues a device-independent certification of signaling dimension similar to previous approaches used for the device-independent dimensionality testing of classical and quantum devices \cite{Gallego-2010a, Dall'Arno-2012a, Ahrens-2012a, brunner2013dim_test, Dall'Arno-2017b}.
Specifically, we obtain Bell inequalities that stipulate necessary conditions on the signaling dimension of $\mc{N}$ in terms of the probabilities $P_\mc{N}(y|x)$, with no assumptions made about the quantum states $\{\rho_x\}_x$, POVM $\{\Pi_y\}_y$, or channel $\mc{N}$ \cite{Dall'Arno-2017a}. 
Complementary results have been obtained by Dall'Arno \textit{et al.} who approached the simulation problem from the quantum side and characterized the set of channels $\mbf{P}_\mc{N}$ that can be obtained using binary encodings for special types of quantum channels $\mc{N}$ \cite{Dall'Arno-2017a}.  
In this paper, we compute a wide range of Bell inequalities using the adjacency decomposition technique \cite{christof2001decomposition}, recovering prior results of Frenkel and Weiner \cite{Frenkel2015} and generalizing work by Heinosaari and Kerppo \cite{Heinosaari2019}.  For certain cases we prove that these inequalities are complete, i.e. providing both necessary and sufficient conditions for signaling dimension.
As a further application, we compute bounds for the signaling dimension of partial replacer channels.
Proofs for our main results are found in the Appendix while our supporting software is found on Github \cite{SignalingDimension.jl}.
 
\section{Signaling Polytopes}

We begin our investigation by reviewing the structure of channels that use noiseless classical communication and shared randomness.  Let $\mc{P}^{n\to n'}$ denote the family of channels having input set $\mc{X}=[n]:=\{1,\cdots,n\}$ and output set $\mc{Y}=[n']$.
A channel $\mbf{P} \in \mc{P}^{n \to n'}$ is represented by an $n'\times n$ column stochastic matrix, and we thus identify $\mc{P}^{n\to n'}$ as a subset of $\mbb{R}^{n'\times n}$, the set of all $n'\times n$ real matrices.  In general we refer to a column (or row) of a matrix as being stochastic if its elements are non-negative and sum to unity, and a column (resp. row) stochastic matrix has only stochastic columns (resp. rows).  The elements of a real matrix $\mbf{G}\in \mbb{R}^{n'\times n}$ are denoted by $G_{y,x}$, while those of a column stochastic matrix $\mbf{P}\in\mc{P}^{n\to n'}$ are denoted by $P(y|x)$ to reflect their status as conditional probabilities.  The Euclidean inner product between $\mbf{G},\mbf{P}\in\mbb{R}^{n'\times n}$ is expressed as $\langle\mbf{G},\mbf{P}\rangle:=\sum_{x,y}G_{y,x}P(y|x)$, and for any $\mbf{G}\in\mbb{R}^{n'\times n}$ and $\gamma\in\mbb{R}$, we let the tuple $(\mbf{G},\gamma)$ denote the linear inequality $\langle\mbf{G},\mbf{P}\rangle\leq \gamma$.

Consider now a scenario in which Alice and Bob have access to a noiseless channel capable of sending $d$ messages. They can use this channel to simulate a noisy channel by applying pre- and post-processing maps.
If they coordinate these maps using a shared random variable $\lambda$ with probability mass function $q(\lambda)$, then they can simulate any channel $\mbf{P}$ that decomposes as

\begin{equation}
\label{Eq:d-simulate}
     P(y|x)=\sum_\lambda q(\lambda)\sum_{m\in[d]}R_\lambda(y|m)T_\lambda(m|x),
\end{equation}

\noindent where $m\in[d]$ is the message sent from Alice to Bob and $T(m|x)$ (resp. $R(y|m)$) is an element of Alice's encoder $\mbf{T}\in\mc{P}^{n\to d}$ (resp. Bob's decoder $\mbf{R}\in\mc{P}^{d\to n'}$).

\begin{definition} \label{Def-signaling-polytope}
For given positive integers $n$, $n'$, and $d$, the set of all channels satisfying Eq. \eqref{Eq:d-simulate} constitute the \textit{signaling polytope}, denoted by $\mc{C}_d^{n\to n'}$.
\end{definition}

\noindent The signaling polytope $\mc{C}_d^{n\to n'}$ is a convex polytope of dimension $n(n'-1)$ whose vertices $\mbf{V}\in\mc{P}^{n \to n'}$ have 0/1 matrix elements and $\rank(\mbf{V}) \leq d$.  We define $\langle\mbf{G},\mbf{P}\rangle\leq\gamma$ as a Bell inequality for $\mc{C}_d^{n\to n'}$ if $\mc{C}_d^{n\to n'}\subset\{\mbf{P}\in\mc{P}^{n \to n'}\;|\;\langle \mbf{G},\mbf{P}\rangle\leq \gamma\}$, and it is a ``tight'' Bell inequality if the equation $\langle\mbf{G},\mbf{P}\rangle=\gamma$ is also solved by $n(n'-1)$ affinely independent vertices.  When the latter holds, the solution space to $\langle\mbf{G},\mbf{P}\rangle=\gamma$ is called a facet of $\mc{C}_d^{n\to n'}$.
The Weyl-Minkowski Theorem ensures that a complete set of tight Bell inequalities $ \{\langle\mbf{G}_k,\mbf{P}\rangle\leq\gamma_k\}_{k=1}^r$ exists such that $\mbf{P}\in \mc{C}_d^{n\to n'}$ \textit{iff} it satisfies all inequalities in this set \cite{Ziegler-2012a}.
Additional details about signaling polytopes are found in Appendix \ref{Appendix-polytope-structure}. 

Having introduced signaling polytopes, we can now define the signaling dimension of a channel.
This terminology is adopted from recent work by Dall'Arno \textit{et al.} \cite{Dall'Arno_no_hypersignaling_2017} who defined the signaling dimension of a system in generalized probability theories; 
an analogous quantity without shared randomness has also been studied by Heinosaari \textit{et al.} \cite{Heinosaari-2020a}.  In what follows, we assume that $\mc{N}:\mc{S}(A)\to\mc{S}(B)$ is a completely positive trace-preserving (CPTP) map, with $\mc{S}(A)$ denoting the set of density operators (i.e. trace-one positive operators) on system $A$, and similarly for $\mc{S}(B)$.

\begin{definition}\label{def-signaling-dimension}
Let $\mc{P}_{\mc{N}}^{ n \to n'}$ be the set of all classical channels $\mbf{P}_{\mc{N}}\in\mc{P}^{n \to n'}$ generated from 
$\mc{N}$ via Eq. \eqref{Eq:channel-induce}.  The $n\to n'$ signaling dimension of $\mc{N}$, denoted by $\kappa^{n\to n'}(\mc{N})$, is the smallest $d$ such that $\mc{P}_{\mc{N}}^{ n \to n'}\subset\mc{C}_d^{n\to n'}$.  The signaling dimension of a channel $\mc{N}$, denoted by $\kappa(\mc{N})$, is the smallest $d$ such that $\mc{P}_{\mc{N}}^{n\to n'}\!\!\subset\mc{C}_d^{n\to n'}$ for all $ n,n'$.
\end{definition}

For any channel $\mc{N}$, a trivial upper bound on the $n \to n'$ signaling dimension is given by

\begin{equation}
\label{Eq:classical-cardinality-bound}
\kappa^{n\to n'}(\mc{N})\leq\min\{n,n'\}.
\end{equation}

\noindent Indeed, when this bound is attained, Alice and Bob can simulate any $\mbf{P}\in\mc{P}^{n \to n'}$: either Alice applies channel $\mbf{P}$ on her input and sends the output to Bob, or she sends the input to Bob and he  applies $\mbf{P}$ on his end.
In Theorem \ref{Thm:main1} we provide necessary and sufficient conditions for when this trivial upper bound is attained.
For a quantum channel $\mc{N}$, the trivial upper bound is

\begin{equation}
\label{Eq:bound-dimension}
    \kappa(\mc{N})\leq\min\{d_A,d_B\},
\end{equation}

\noindent where $d_A$ and $d_B$ are the Hilbert space dimensions of Alice and Bob's systems.  This bound is a direct consequence of Frenkel and Weiner's result \cite{Frenkel2015}, which can be restated in our terminology as $\kappa({\text{id}_d}) =d$, where $\text{id}_d$ is the noiseless channel on a $d$-dimensional quantum system.
To prove Eq. \eqref{Eq:bound-dimension}, Alice can either send the states $\{\rho_x\}_x$ to Bob who then performs the POVM $\{\mc{N}^\dagger(\Pi_y)\}_y$, or she can send the states $\{\mc{N}(\rho_x)\}_x$ to Bob who then perfoms the POVM $\{\Pi_y\}_y$.  Here $\mc{N}^\dagger$ denotes the adjoint map of $\mc{N}$.  
Another relationship we observe is
\begin{equation}
    \label{Eq:bound-carethedory}
    \kappa^{n\to n'}(\mc{N})=\kappa^{n\to d_B^2}(\mc{N})
\qquad \forall \; n'\geq d_B^2.
\end{equation}
This follows from Carath\'{e}odory's Theorem \cite{Barvinok-2002a}, which implies that every POVM on a $d_B$-dimensional system can be expressed as a convex combination of POVMs with no more than $d_B^2$ outcomes \cite{Davies-1978a}.  Since shared randomness is free, Alice and Bob can always restrict their attention to POVMs with no more than $d_B^2$ outcomes for the purposes of simulating any channel in $\mc{P}_\mc{N}^{n\to n'}$ when $n'\geq d_B^2$.

The notion of signaling dimension also applies to noisy classical channels.  A classical channel from set $\mc{X}$ to $\mc{Y}$ can be represented by a CPTP map $\mc{N}:\mc{S}(\mbb{C}^{|\mc{X}|})\to\mc{S}(\mbb{C}^{|\mc{Y}|})$ that completely dephases its input and output in fixed orthonormal bases $\{\ket{x}\}_{x\in\mc{X}}$ and $\{\ket{y}\}_{y\in\mc{Y}}$, respectively.  The transition probabilities of $\mc{N}$ are then given by
Eq. \eqref{Eq:channel-induce} as $P_{\mc{N}}(y|x) = \tr\big[\op{y}{y}\mc{N}\big(\op{x}{x}\big)\big]$.
The channel $\mc{N}$ can be used to generate another channel $\ol{\mc{N}}$ with input and output alphabets $\ol{\mc{X}}$ and $\ol{\mc{Y}}$ by performing a pre-processing map $\mbf{T}:\ol{\mc{X}}\to\mc{X}$ and post-processing map $\mbf{R}:\mc{Y}\to\ol{\mc{Y}}$, thereby yielding the channel  $\mbf{P}_{\ol{\mc{N}}}=\mbf{R}\mbf{P}_\mc{N}\mbf{T}$.  When this relationship holds, $\mbf{P}_{\ol{\mc{N}}}$ is said to be \textit{ultraweakly majorized} by $\mbf{P}_\mc{N}$ \cite{Heinosaari2019, Heinosaari-2020a},  and the signaling dimension of $\mbf{P}_{\ol{\mc{N}}}$ is no greater than that of $\mbf{P}_\mc{N}$ \cite{Cubitt-2011a}.

In practice, the channel connecting Alice and Bob may be unknown or not fully characterized.  This is the case in most experimental settings where unpredictable noise affects the encoded quantum states.  In such scenarios it is desirable to ascertain certain properties of the channel without having to perform full channel tomography, a procedure that requires trust in the state preparation device on Alice's end and the measurement device on Bob's side.  A device-independent approach infers properties of the channel by analyzing the observed input-output classical correlations $P(y|x)$ obtained as sample averages over many uses of the memoryless channel \cite{Dall'Arno-2017a}.  The Bell inequalities introduced in the next section can be used to certify the signaling dimension of the channel: if the correlations $P(y|x)$ are shown to violate a Bell inequality of $\mc{C}_d^{n\to n'}$, then the signaling dimension $\kappa^{n\to n'}(\mc{N})> d$.
If these correlations arise from some untrusted quantum channel $\mc{N}^{A\to B}$, by Eq. \eqref{Eq:bound-dimension} it then follows that $\min\{d_A,d_B\}>d$.  Hence a device-independent certification of signaling dimension leads to a device-independent certification of the physical input/output Hilbert spaces of the channel connecting Alice and Bob.

\section{Bell Inequalities for Signaling Polytopes}

In this section we discuss Bell inequalities for signaling polytopes.  Since signaling polytopes are invariant under the relabelling of inputs and outputs, all discussed inequalities describe a family of inequalities where each element is obtained by a permutation of the inputs and/or outputs.
Additionally, a Bell inequality for one signaling polytope can be lifted to a polytope having more inputs and/or outputs  \cite{Pironio2005,Rosset2014} (see Fig. \ref{Fig:input-output-lifting}).  Formally, a Bell inequality $\langle\mbf{G},\mbf{P}\rangle\leq\gamma$ is said to be input lifted to $\langle\mbf{G}'',\mbf{P}\rangle\leq\gamma$ if $\mbf{G}''\in\mbb{R}^{n'\times m}$ is obtained from $\mbf{G}\in\mbb{R}^{n'\times n}$ by padding it with $(m-n)$ all-zero columns.  On the other hand, a Bell inequality $\langle\mbf{G},\mbf{P}\rangle\leq\gamma$ is said to be output lifted to $\langle\mbf{G}',\mbf{P}\rangle\leq\gamma$ if $\mbf{G}'\in \mbb{R}^{m'\times n}$ is obtained from $\mbf{G}\in\mbb{R}^{n'\times n}$ by copying rows; \textit{i.e.}, there exists a surjective function $f:[m']\to [n']$ such that $G'_{y,x}=G_{f(y),x}$ for all $y\in[m']$ and $x\in[n]$.
Note that $m'>n'$ and $m>n$ in these examples.

\begin{figure}[t] 
    \centering
    \begin{tabular}{c c c c}
    (a) & $\quad\mbf{G} = \begin{bmatrix}
        1 & 0 & 0 \\ 0 & 1 & 0 \\ 0 & 0 & 1 \\
    \end{bmatrix}$ & $\to$ & $\mbf{G}'' = \begin{bmatrix}
        1 & 0 & 0 & 0 \\ 0 & 1 & 0 & 0 \\ 0 & 0 & 1 & 0 \\
    \end{bmatrix}$ \\
    \hfill\\
    (b) & $\quad\mbf{G} = \begin{bmatrix}
        1 & 0 & 0 \\ 0 & 1 & 0 \\ 0 & 0 & 1 \\
    \end{bmatrix}$ & $\to$ & $\mbf{G}' = \begin{bmatrix}
        1 & 0 & 0 \\ 1 & 0 & 0 \\ 0 & 1 & 0 \\ 0 & 0 & 1 \\
    \end{bmatrix}$
    \end{tabular}
 \caption{\linespread{1}\selectfont{\small
    (a) Input and (b) output liftings of $\mbf{G}=\mbb{I}_3$.
}}
\label{Fig:input-output-lifting}
\end{figure}

To obtain polytope facets, it is typical to first enumerate the vertices, then use a transformation technique such as Fourier-Motzkin elimination to derive the facets \cite{Ziegler-2012a}.
Software such as PORTA \cite{PORTA, XPORTA.jl} assists in this computation, but the large number of vertices leads to impractical run times.
To improve efficiency, we utilize the adjacency decomposition technique which heavily exploits the permutation symmetry of signaling polytopes \cite{ christof2001decomposition} (see Appendix \ref{Appendix-Adjacency}).
Our software and computed facets are publicly available on Github \cite{SignalingDimension.jl} while a catalog of general tight Bell inequalities is provided in Appendix \ref{Appendix-facet-classes}.
We now turn to a specific family of Bell inequalities motivated by our computational results.

\subsection{Ambiguous Guessing Games}

For $k\in[0,n']$ and $d\leq \min\{n,n'\}$, let $\mbf{G}^{n, n'}_{k,d}$ be any $n'\times n$ matrix such that (i) $k$ rows are stochastic with 0/1 elements, and (ii) the remaining $(n'-k)$ rows have $1/(n-d+1)$ in each column.  As explained below, it will be helpful to refer to rows of type (i) as  ``guessing rows'' and rows of type (ii) as ``ambiguous rows.''
For example, if $n=n'=6$, $k=5$, and $d=2$, then up to a permutation of rows and columns we have
\begin{equation}
    \mbf{G}^{6,6}_{5,2} = \begin{bmatrix}
        1 & 0 & 0 & 0 & 0 & 0 \\
        1 & 0 & 0 & 0 & 0 & 0 \\
        0 & 1 & 0 & 0 & 0 & 0 \\
        0 & 1 & 0 & 0 & 0 & 0 \\
        0 & 0 & 1 & 0 & 0 & 0 \\
        \tfrac{1}{5} & \tfrac{1}{5} & \tfrac{1}{5} & \tfrac{1}{5} & \tfrac{1}{5} & \tfrac{1}{5} \\
    \end{bmatrix}.
\end{equation}
\noindent  
For any channel $\mbf{P}\in\mc{C}_d^{n\to n'}$, the Bell inequality 

\begin{align}
\label{Eq:Ambiguous}
    \langle\mbf{G}^{n, n'}_{k,d},\mbf{P}\rangle\leq d
\end{align}

\noindent is satisfied. To prove this bound, suppose without loss of generality that the first $k$ rows of $\mbf{G}^{n, n'}_{k,d}$ are guessing rows.
Let $\mbf{V}$ be any vertex of $\mc{C}_d^{n\to n'}$ where $t$ of its first $k$ rows are nonzero.  If $t=d$, then clearly Eq. \eqref{Eq:Ambiguous} holds.  Otherwise, if $t<d$, then  $\langle\mbf{G}^{n, n'}_{k,d},\mbf{V}\rangle\leq t+(n-t)/(n-d+1)\leq d$, where the last inequality follows after some algebraic manipulation.

Equation \eqref{Eq:Ambiguous} can be interpreted as the score of a guessing game that Bob plays with Alice.  Suppose that Alice chooses a channel input $x\in[n]$ with uniform probability and sends it through a channel $\mbf{P}$.  Based on the channel output $y$, Bob guesses the value of $x$.  Formally, Bob computes $\hat{x}=f(y)$ for some guessing function $f$, and if $\hat{x}=x$ then he receives one point.
In this game, Bob may also declare Alice's input as being ambiguous or indistinguishable, meaning that $f:[d]\to [n]\cup\{?\}$ with $``?"$ denoting Bob's declaration of the ambiguous input.  However, whenever Bob declares $``?"$ he only receives $1/(n-d+1)$ points.
Then, Eq. \eqref{Eq:Ambiguous} says that whenever $\mbf{P}\in\mc{C}_d^{n\to n'}$ Bob's average score is bounded by $\frac{d}{n}$.
Note, there is a one-to-one correspondence between each $\mbf{G}^{n,n'}_{k,d}$ and the particular guessing function $f$ that Bob performs.
If $y$ labels a guessing row of $\mbf{G}^{n,n'}_{k,d}$, then $f(y)=\hat{x}$, where $\hat{x}$ labels the only nonzero column of row $y$.  On the other hand, if $y$ labels an ambiguous row, then $f(y)=``?"$.

We define the $(k,d)$-\textit{ambiguous polytope} $\mc{A}_{k,d}^{n\to n'}$ as the collection of all channels $\mbf{P}\in\mc{P}^{n\to n'}$ satisfying Eq. \eqref{Eq:Ambiguous} for every $\mbf{G}^{n,n'}_{k,d}$.
Naturally, $\mc{C}_d^{n \to n'}\subset \mc{A}_{k,d}^{n \to n'}$ for all $k\in[0,n']$, therefore, if $\mbf{P}\notin\mc{A}_{k,d}^{n \to n'}$, then $\mbf{P}\notin \mc{C}_d^{n \to n'}$.
Based on the discussion of the previous paragraph, it is easy to decide membership of $\mc{A}_{k,d}^{n\to n'}$.

\begin{proposition}
\label{Prop:ambiguous-condition}
    A channel $\mbf{P}\in\mc{P}^{n\to n'}$ belongs to $\mc{A}_{k,d}^{n\to n'}$ \textit{iff}
\begin{equation}
\label{Eq:Prop-ambiguous-form}
    \max_{\pi\in\msf{S}_{n'}}\sum_{i=1}^k\Vert \mbf{r}_{\pi(i)}\Vert_\infty+\frac{1}{n-d+1}\sum_{i=k+1}^{n'}\Vert \mbf{r}_{\pi(i)}\Vert_1\leq d,
\end{equation}

\noindent where the maximization is taken over all permutations on $[n']$, $\mbf{r}_i$ denotes the $i^{th}$ row of $\mbf{P}$, $\Vert\mbf{r}_i\Vert_\infty$ is the largest element in $\mbf{r}_i$, and $\Vert\mbf{r}_i\Vert_1$ is the row sum of $\mbf{r}_i$.
\end{proposition}
The maximization on the LHS of Eq. \eqref{Eq:Prop-ambiguous-form} can be performed efficiently using the following procedure.  For each row $\mbf{r}_i$ we assign a pair $(a_i, b_i)$ where $a_i=\Vert\mbf{r}_i\Vert_\infty$ and $b_i=\frac{1}{n-d+1}\Vert\mbf{r}_i\Vert_1$.  Define $\delta i=a_i - b_i$, and relabel the rows of $\mbf{P}$ in non-increasing order of the $\delta_i$.  
Then according to this sorting, we have an ambiguous guessing game score of $\sum_{i=1}^k a_i +\sum_{i=k+1}^{n'} b_i$,
which we claim attains the maximum on the LHS of Eq. \eqref{Eq:Prop-ambiguous-form}.  Indeed, for any other row permutation $\pi$, the guessing game score is given by
\begin{align}
&\sum_{\substack{i\in\{1,\cdots,k\}\\\pi(i)\in\{1,\cdots,k\}}} a_i+\sum_{\substack{i\in\{1,\cdots,k\}\\\pi(i)\in\{k+1,\cdots,n'\}}} b_i\notag\\
+&\sum_{\substack{i\in\{k+1,\cdots,n'\}\\\pi(i)\in\{1,\cdots,k\}}} a_i+\sum_{\substack{i\in\{k+1,\cdots,n'\}\\\pi(i)\in\{k+1,\cdots,n'\}}}b_i.\end{align}
Hence the difference in these two scores is
\begin{align}
    \sum_{\substack{i\in\{1,\cdots,k\}\\\pi(i)\in\{k+1,\cdots,n'\}}} (a_i-b_i)-\sum_{\substack{i\in\{k+1,\cdots,n'\}\\\pi(i)\in\{1,\cdots,k\}}} (a_i-b_i)\geq 0,
\end{align}
where the inequality follows from the fact that we have ordered the indices in non-increasing order of $(a_i-b_i)$, and the number of terms in each summation is the same since $\pi$ is a bijection.

A special case of the ambiguous guessing games arises when $k=n'$.  Then up to a normalization factor $\frac{1}{n}$, we interpret the LHS of Eq. \eqref{Eq:Prop-ambiguous-form} as the success probability when Bob performs maximum likelihood estimation of Alice's input value $x$ given his outcome $y$ (i.e. he chooses the value $x$ that maximizes $P(y|x)$).  We hence define $\mc{M}_d^{n\to n'}:=\mc{A}_{n',d}^{n\to n'}$ as the \textit{maximum likelihood (ML) estimation polytope}.  Using Proposition \ref{Prop:ambiguous-condition} we see that 
\begin{equation}
    \mbf{P}\in \mc{M}_d^{n\to n'}\quad\Leftrightarrow\quad  \sum_{y=1}^{n'}\max_{x\in[n]}P(y|x)\leq d.
\end{equation}

An important question is whether the ambiguous guessing Bell inequalities of Eq. \eqref{Eq:Ambiguous} are tight for a signaling polytope $\mc{C}_d^{n\to n'}$.  In general this will not be case.  For instance, $\langle\mbf{G}^{n, n'}_{k,d},\mbf{P}\rangle\leq d$ is trivially satisfied whenever $k=0$.  Nevertheless, in many cases we can establish tightness of these inequalities.  A demonstration of the following facts is carried out in Appendix \ref{proof-of-prop-ambiguous-guessing-facets}.

\begin{proposition}
\label{prop-ambiguous-guessing-facets}
\begin{enumerate}
    \item[]
    \item[(i)] For $\min\{n,n'\} > d > 1$ and $k=n'$, Eq. \eqref{Eq:Ambiguous} is a tight Bell inequality of $\mc{C}_d^{n\to n'}$ \textit{iff} $\mbf{G}^{n,n'}_{k,d}$ can be obtained by performing input/output liftings and row/column permutations on an $m\times m$ identity matrix $\mbb{I}_m$, with $\min\{n,n'\} \geq m > d$.
    \item[(ii)] For $n'>k\geq n>d>1$, Eq. \eqref{Eq:Ambiguous} is a tight Bell inequality of $\mc{C}_d^{n\to n'}$ \textit{iff} $\mbf{G}^{n,n'}_{k,d}$ can be obtained from the $(n+1)\times n$ matrix
    \begin{equation}
    \label{Eq:ambiguous-canonical}
        \begin{bmatrix}&&\mbb{I}_n&&\\\tfrac{1}{n-d+1}&\cdots&\tfrac{1}{n-d+1}&\cdots \tfrac{1}{n-d+1}\end{bmatrix}
        \end{equation}
        by performing output liftings and row/column permutations.  
    \end{enumerate}
\end{proposition}

\noindent Note that the input/output liftings are used to manipulate the identity matrix $\mbb{I}_m$ and the matrix of Eq. \eqref{Eq:ambiguous-canonical} into an $n'\times n$ matrix $\mbf{G}_{k,d}^{n,n'}$.  The tight Bell inequalities described in Proposition \ref{prop-ambiguous-guessing-facets}(i) completely characterize the ML polytope $\mc{M}_d^{n\to n'}$.  For this reason, we refer to any $\mbf{G}^{n,n'}_{k,d}$ satisfying the conditions of Proposition \ref{prop-ambiguous-guessing-facets}(i) as a maximum likelihood (ML) facet (see Appendix \ref{Appendix-ml-facets}). Likewise, we refer to any $\mbf{G}^{n,n'}_{k,d}$ satisfying the conditions of Proposition \ref{prop-ambiguous-guessing-facets}(ii) as an ambiguous guessing facet (see Appendix \ref{Appendix-ambiguous-guessing-facets}).

\subsection{Complete Sets of Bell Inequalities}

In general, we are unable to identify the complete set of tight Bell inequalities that bound each signaling polytope $\mc{C}_d^{n \to n'}$.
However, we analytically solve the problem in special cases.
\begin{theorem}
\label{Thm:main1}
Let $n$ and $n'$ be arbitrary integers.
\begin{enumerate}
    \item[(i)] If $d=n'-1$, then $\mc{C}^{n\to n'}_d=\mc{M}_d^{n\to n'}$.
    \item[(ii)] If $d=n-1$, then $\mc{C}_d^{n \to n'} = \displaystyle\bigcap_{k=n}^{n'} \mc{A}_{k,d}^{n\to n'}$.
\end{enumerate}
\end{theorem}

\noindent  In other words, to decide whether a channel can be simulated by an amount of classical messages strictly less than the input/output alphabets, it suffices to consider the ambiguous guessing games.  Moreover, by Eq. \eqref{Eq:Prop-ambiguous-form} it is simple to check if these conditions are satisfied for a given channel $\mbf{P}$.  A proof of Theorem \ref{Thm:main1} is found in Appendix \ref{appendix-proof-thm-main1}.

\begin{figure}[b]
    \centering
    \begin{tabular}{c c}
         (a) $2 \geq \begin{bmatrix}
            1 & 0 & 0 & 0 & 0 & 0 \\
            1 & 0 & 0 & 0 & 0 & 0 \\
            0 & 1 & 0 & 0 & 0 & 0 \\
            0 & 0 & 1 & 0 & 0 & 0 \\
        \end{bmatrix}$ & (b) $2 \geq \begin{bmatrix}
            1 & 0 & 0 & 0 & 0 & 0 \\
            0 & 1 & 0 & 0 & 0 & 0 \\
            0 & 0 & 1 & 0 & 0 & 0 \\
            0 & 0 & 0 & 1 & 0 & 0 \\
        \end{bmatrix}$ \\
        \hfill \\
        (c) $3 \geq \begin{bmatrix}
            1 & 1 & 0 & 0 & 0 & 0 \\
            1 & 0 & 1 & 0 & 0 & 0 \\
            0 & 1 & 1 & 0 & 0 & 0 \\
            0 & 0 & 0 & 1 & 0 & 0 \\
        \end{bmatrix}$ & (d) $5 \geq \begin{bmatrix}
            1 & 1 & 1 & 0 & 0 & 0 \\
            1 & 0 & 0 & 1 & 1 & 0 \\
            0 & 1 & 0 & 1 & 0 & 1 \\
            0 & 0 & 1 & 0 & 1 & 1 \\
        \end{bmatrix}$ \\
        \hfill \\
        (e) $4 \geq \begin{bmatrix}
            2 & 0 & 0 & 0 & 0 & 0 \\
            0 & 2 & 0 & 0 & 0 & 0 \\
            0 & 0 & 2 & 0 & 0 & 0 \\
            1 & 1 & 1 & 0 & 0 & 0 \\
        \end{bmatrix}$ & (f) $4 \geq \begin{bmatrix}
            2 & 0 & 0 & 0 & 0 & 0 \\
            0 & 2 & 0 & 0 & 0 & 0 \\
            0 & 0 & 1 & 1 & 0 & 0 \\
            1 & 1 & 1 & 0 & 0 & 0 \\
        \end{bmatrix}$\\
        \hfill \\
        (g) $4 \geq \begin{bmatrix}
            2 & 0 & 0 & 0 & 0 & 0 \\
            0 & 1 & 0 & 1 & 0 & 0 \\
            0 & 0 & 1 & 0 & 1 & 0 \\
            1 & 1 & 1 & 0 & 0 & 0 \\
        \end{bmatrix}$ & (h) $4 \geq \begin{bmatrix}
            1 & 0 & 0 & 1 & 0 & 0 \\
            0 & 1 & 0 & 0 & 1 & 0 \\
            0 & 0 & 1 & 0 & 0 & 1 \\
            1 & 1 & 1 & 0 & 0 & 0 \\
        \end{bmatrix}$
    \end{tabular}
    \caption{\linespread{1}\selectfont{\small Generator facets for the $\mc{C}_2^{6 \to 4}$ signaling polytope. 
    Each inequality is expressed as $\gamma \geq \mbf{G}$ where the inner product $\langle \mbf{G} , \mathbf{P}\rangle$ is implied.
    (a) ML facet input/output lifted from $\mc{C}_2^{3,3}$.
    (b) ML facet output lifted from $\mc{C}_2^{4 \to 4}$.
    (c) Anti-guessing facet output lifted from $\mc{C}_2^{4 \to 4}$.
    (d) $k$-guessing facet of $\mc{C}_2^{6 \to 4}$. (e) Ambiguous guessing facet output lifted from $\mc{C}_2^{3 \to 4}$.
    (f-h) Rescalings of the $\mc{C}_2^{3 \to 4}$ ambiguous guessing facet output lifted to $\mc{C}_2^{6 \to 4}$.
    General forms of these tight Bell inequalities are derived in Appendix \ref{Appendix-facet-classes}.
    }}
    \label{Fig:6-2-4_facets}
\end{figure}

We also characterize the  $\mc{C}^{n\to 4}_2$ signaling polytope.  As an application, this case can be used to understand the classical simulation cost of performing Bell measurements on a two-qubit system, since this process induces a classical channel with four outputs.
\begin{theorem}
\label{Thm:main2}
For any integer $n$, a channel $\mbf{P}\in \mc{P}^{n\to 4}$ belongs to $\mc{C}^{n\to 4}_2$ \textit{iff} it satisfies the eight Bell inequalities depicted in Fig. \ref{Fig:6-2-4_facets} and all their input/output permutations. 
\end{theorem}

\noindent Remarkably, this result shows that no new facet classes for $\mc{C}^{n\to 4}_2$ are found when $n > 6$.
Consequently, to demonstrate that a channel $\mbf{P}\in \mc{P}^{n\to 4}$ requires more than one bit for simulation, it suffices to consider input sets of size no greater than six.  For $n<6$, the facet classes of $\mc{C}^{n\to 4}_2$ are given by the facets in Fig. \ref{Fig:6-2-4_facets} having $(6-n)$ all-zero columns.
We conjecture that in general, no more than $\binom{n'}{d}$ inputs are needed to certify that a channel $\mbf{P}\in\mc{P}^{n\to n'}$ has a signaling dimension larger than $d$.  A proof of Theorem \ref{Thm:main2} is found in Appendix \ref{appendix-proof-of-thm-2}.

\subsection{The Signaling Dimension of Replacer Channels}\label{section-replacer-channel}

In the device-independent scenario, Alice and Bob make minimal assumptions about the channel $\mc{N}^{A\to B}$ connecting them; they simply try to lower bound the dimensions of $\mc{N}$ using input-output correlations $P_{\mc{N}}(y|x)$.  Applying the results of the previous section, if $\langle\mbf{G},\mbf{P}\rangle\leq \gamma$ is a Bell inequality for $\mc{C}_d^{n\to n'}$ and 
\begin{equation}\label{eq:quantum-optimization}
    \max_{\{\rho_x\}_x,\{\Pi_y\}_y} \sum_{y,x} G_{y,x} \tr\Big[\Pi_y \mc{N}\big(\rho_x\big)\Big]>\gamma,
\end{equation}

\noindent then $\min\{d_A,d_B\}\geq \kappa(\mc{N})>d$.  Eq. \eqref{eq:quantum-optimization} describes a conic optimization problem that can be analytically solved only in special cases \cite{cv-2021a}.  Hence deciding whether a given quantum channel can violate a particular Bell inequality is typically quite challenging.

Despite this general difficulty, we nevertheless establish bounds for the signaling dimension of partial replacer  channels.
A $d$-dimensional partial replacer channel has the form
\begin{equation} \label{Eq:replacer-channel}
    \mc{R}_{\mu}(X)= \mu X+(1-\mu)\tr[X]\sigma,
\end{equation}
\noindent where $1 \geq \mu \geq 0$ and $\sigma$ is some fixed density matrix.  The partial depolarizing channel $\mc{D}_{\mu}$ corresponds to $\sigma$ being the maximally mixed state whereas the partial erasure channel $\mc{E}_\mu$ corresponds to $\sigma$ being an erasure flag $\op{E}{E}$ with $\ket{E}$ being orthogonal to $\{\ket{1},\cdots,\ket{d}\}$.

\begin{theorem} \label{thm-replacer-bounds}
The signaling dimension of a partial replacer channel is bounded by
\begin{equation}  \label{Eq:replacer-bounds}
\lceil \mu d+1-\mu\rceil\leq\kappa(\mc{R}_\mu)\leq\min\{d,\lceil\mu d+1\rceil\}.
\end{equation}
 Moreover, for the partial erasure channel, the upper bound is tight for all $\mu\in[0,1]$.
\end{theorem}   

\begin{proof}
We first prove the upper bound in Eq. \eqref{Eq:replacer-bounds}.  The trivial bound $\kappa(\mc{R}_\mu)\leq d$ was already observed in Eq. \eqref{Eq:bound-dimension}.  To show that $\kappa(\mc{R}_\mu)\leq \lceil\mu d+1\rceil$, let $\{\rho_x\}_x$ be any collection of inputs and $\{\Pi_y\}_y$ a POVM.  Then 
 \begin{equation}
     \label{Eq:Depolarizing-classical-channel}
 P_{\mc{R}_\mu}(y|x)=\mu P(y|x)+(1-\mu)S(y),
\end{equation}
 where $P(y|x)=\tr[\Pi_y\rho_x]$ and $S(y)=\tr[\Pi_y\sigma]$.  From Ref. \cite{Frenkel2015}, we know that $P(b|x)$ can be decomposed like Eq. \eqref{Eq:d-simulate}.  Substituting this into Eq. \eqref{Eq:Depolarizing-classical-channel} yields
 \begin{align}
 P_{\mc{R}_\mu}(u|x)&=\sum_\lambda q(\lambda)\sum_{m=1}^dR_\lambda(m|x)\notag\\
     &\qquad\times[\mu T_\lambda(y|m)+(1-\mu)S(y)].
 \end{align}
 For $r=\lceil\mu d +1\rceil$, let $\nu$ be a random variable uniformly distributed over $\{\binom{d}{r-1}\}$, which is the collection of all subsets of $[d]$ having size $r-1$.  For a given $\lambda$, $\nu$, and input $x$, Alice performs the channel $\mbf{T}_\lambda$.  If $m\in\nu$, Alice sends message $m'=m$; otherwise, Alice sends message $m'=0$.  Upon receiving $m'$, Bob does the following: if $m'\not=0$ he performs channel $\mbf{R}_\lambda$ with probability $\tfrac{\mu d}{r-1}$ and samples from distribution $S(y)$ with probability $1-\tfrac{\mu d}{r-1}$; if $m'\not=0$ he samples from $S(y)$ with probability one. Since $Pr\{m\in\nu\}=\frac{r-1}{d}$, this protocol faithfully simulates $P_{\mc{R}_\mu}$.  To establish the lower bound in Eq. \eqref{Eq:Depolarizing-classical-channel}, suppose that Alice sends orthogonal states $\{\ket{1},\cdots,\ket{d}\}$
  and Bob measures in the same basis.  Then 
  \begin{equation}
      \sum_{i=1}^d\bra{i}\mc{R}_\mu(\op{i}{i})\ket{i}=d\mu+(1-\mu),
  \end{equation}
  which will violate Eq. \eqref{Eq:Ambiguous} for the ML polytope $\mc{M}^{d\to d}_{r}$ whenever $r< \mu d+(1-\mu)$.  Hence any zero-error simulation will require at least $\lceil \mu d+1-\mu\rceil$ classical messages.  For the erasure channel, this lower bound can be tightened by considering the score for other ambiguous games, as detailed in Appendix \ref{Appendix-replacer}.
\end{proof}

\section*{Discussion}

In this work, we have presented the signaling dimension of a channel as its classical simulation cost.
In doing so, we have advanced a device-independent framework for certifying the signaling dimension of a quantum channel as well as its input/output dimensions.  While this work focuses on communication systems, our framework also applies to computation and memory tasks.  

The family of ambiguous guessing games includes the maximum likelihood facets, which say that $\sum_{y=1}^{n'}\max_{x\in[n]}P(y|x)\leq d$ for all $\mbf{P}\in\mc{C}_d^{n\to n'}$.  Since the results of Frenkel and Weiner imply that $\mc{P}_{\mc{N}}^{n\to n'}\subset \mc{C}_d^{n\to n'}$ whenever $d \geq \min\{d_A,d_B\}$ for channel $\mc{N}^{A\to B}$ \cite{Frenkel2015}, it follows that
\begin{equation}
     \max_{\substack{\{\rho_x\}_{x\in[n]}\\\{\Pi_y\}_{y\in[n]}}} \sum_{x=1}^n \tr\Big[\Pi_x \mc{N}\big(\rho_x\big)\Big]\leq d,
\end{equation}
an observation also made in Ref. \cite{brunner2013dim_test}.    Despite the simplicity of this bound, in general it is too loose to certify the input/output Hilbert space dimensions of a channel.  For example, consider the $50:50$ erasure channel $\mc{E}_{1/2}$ acting on a $d_A=3$ system.
It can be verified  that $\mc{P}_{\mc{E}_{1/2}}^{n\to n'}\subset \mc{M}_2^{n\to n'}$, \textit{i.e.} $\sum_x\tr\big[\Pi_x \mc{E}_{1/2}(\rho_x)\big]\leq 2$ for all $\{\rho_x\}_x$ and $\{\Pi_y\}_y$.
Hence maximum likelihood estimation yields the lower bound $\kappa(\mc{E}_{1/2}) \geq 2$.
On the other hand, the classical channel 

\begin{equation}
    \mbf{P_{\mc{E}_{1/2}}} = \begin{bmatrix}
        0.5 & 0 & 0 \\
        0 & 0.5 & 0 \\
        0 & 0 & 0.5 \\
        0.5 & 0.5 & 0.5 \\
    \end{bmatrix}
\end{equation}

\noindent generated by orthonormal input states $\{\ket{1},\ket{2},\ket{3}\}$ and a measurement in the orthonormal basis $\{\ket{1},\ket{2},\ket{3},\ket{E}\}$ violates Eq. \eqref{Eq:Prop-ambiguous-form} for the $\mc{A}_{3,2}^{3\to 4}$ ambiguous polytope.
Hence $\mbf{P}_{\mc{E}_{1/2}} \notin \mc{A}_{3,2}^{3\to 4}$ and it follows that $\kappa^{3 \to 4}(\mc{E}_{1/2}) \geq 3$.
Therefore, the ambiguous guessing game certifies the qutrit nature of the input space whereas maximum likelihood estimation does not.

Our results can be extended in two key directions.
First, our characterization of the signaling polytope is incomplete.
Novel Bell inequalities, lifting rules, and complete sets of facets can be derived beyond those discussed in this work.
Such results would help improve the signaling dimension bounds and the efficiency of computing Bell inequalities.
Second, the signaling dimension specifies the classical cost of simulating a quantum channel, but not the protocol that achieves the classical simulation.
Such a simulation protocol would apply broadly across the field of quantum information science and technology.

\subsection*{Supporting Software}

This work is supported by SignalingDimension.jl \cite{SignalingDimension.jl}.
This software package includes our signaling polytope computations, numerical facet verification, and signaling dimension certification examples.
SignalingDimension.jl is publicly available on Github and written in the Julia programming language \cite{Julia-2017}.
The software is documented, tested, and reproducible on a laptop computer.
The interested reader should review the software documentation as it elucidates many details of our work.
\vspace{12pt} 

\subsection*{Acknowledgements}

We thank Marius Junge for enlightening discussions during the preparation of this paper.  We acknowledge NSF Award \# 2016136  for supporting this work. 

\bibliography{references}

\onecolumngrid

\newpage

\appendix

\section{Notation Glossary}

\begin{table}[h]
    \centering
    \def\arraystretch{1.6}
    \begin{tabular}{|L{0.13\textwidth}|L{0.28\textwidth}|L{0.55\textwidth}|}
        \hline 
        \textbf{\large Notation} & \textbf{\large Terminology} & \textbf{\large Definition} \\
        \hhline{|=|=|=|}
        $\mc{P}^{n \to n'}$ & Set of Classical Channels & The subset of $\mbb{R}^{n'\times n}$ containing column stochastic matrices. \\
        \hline
        $\mbf{P}$ & Classical Channel & An element of $\mc{P}^{n \to n'}$ that represents a classical channel with $n$ inputs and $n'$ outputs. \\
        \hline
        $\mc{N}$ & Quantum Channel & A completely positive trace-preserving map. \\
        \hline 
        $\mc{P}_{\mc{N}}^{n \to n'}$ & Set of Classical Channels Generated from $\mc{N}$ &  The subset of $\mc{P}^{n \to n'}$ which decomposes as Eq. \eqref{Eq:channel-induce} for some quantum channel
        $\mc{N}$. \\
        \hline 
        $\mc{C}_d^{n \to n'}$ & Signaling Polytope & The subset of $\mc{P}^{n \to n'}$ containing channels that decomposes as Eq. \eqref{Eq:d-simulate} (see Def. \ref{Def-signaling-polytope}).  \\
        \hline
        $(\mbf{G},\gamma)$ & Linear Bell Inequality & A tuple describing the linear inequality $\langle \mbf{G},\mbf{P}\rangle \leq \gamma$ where $\mbf{G}\in\mbb{R}^{n'\times n}$, $\gamma\in\mbb{R}$, and $\mbf{P}\in\mc{P}^{n \to n'}$. \\
        \hline
        $\kappa^{n \to n'}(\mc{N})$ & The $n\to n'$ Signaling Dimension of $\mc{N}$ & The smallest integer $d$ such that $\mc{P}_{\mc{N}}^{n \to n'}\subset\mc{C}_d^{n \to n'}$ (see Def. \ref{def-signaling-dimension}).\\
        \hline
        $\kappa(\mc{N})$ & The Signaling Dimension of $\mc{N}$ & The smallest integer $d$ such that $\mc{P}_{\mc{N}}^{n \to n'}\subset\mc{C}_d^{n \to n'}$ for all positive integers $n$ and $n'$ (see Def. \ref{def-signaling-dimension}). \\
        \hline
        $(\mbf{G}^{n,n'}_{k,d},d)$ & Ambiguous Guessing Game & A signaling polytope Bell inequality where $\mbf{G}^{n,n'}_{k,d}\in\mbb{R}^{n'\times n}$ has $k$ rows that are row stochastic with 0/1 elements and $(n'-k)$ rows with each column containing $1/(n-d+1)$.\\
        \hline
        $\mc{A}_{k,d}^{n \to n'}$ & Ambiguous Polytope & The subset of $\mc{P}^{n \to n'}$ which is tightly bound by inequalities of the form $(\mbf{G}^{n,n'}_{k,d}, d)$. \\
        \hline 
        $\mc{M}_d^{n \to n'}$ & Maximum Likelihood Estimation Polytope & The subset of $\mc{P}^{n \to n'}$ defined as the ambiguous polytope $\mc{A}^{n \to n'}_{k,d}$ where $k=n'$. \\
        \hline
        $\mc{R}_\mu$ & Partial Replacer Channel & A quantum channel that replaces the input state $\rho_x$ with quantum state $\sigma$ with probability $(1-\mu)$. \\
        \hline
        $\mc{E}_\mu$ & Partial Erasure Channel & A partial replacer channel that replaces the  input with $\sigma = \op{E}{E}$ where $\ket{E}$ is orthogonal to the input Hilbert space. \\
        \hline
        $\mc{V}_d^{n \to n'}$ & Signaling Polytope Vertices & The subset of $\mc{C}_d^{n \to n'}$ containing classical channels with 0/1 elements.\\
        \hline 
        $\mc{F}_d^{n \to n'}$ & Signaling Polytope Facets & The complete set of tight Bell inequalities for  $\mc{C}_d^{n \to n'}$.\\
        \hline 
        $\mc{G}_d^{n \to n'}$ & Signaling Polytope Generator Facets & The subset of $\mc{F}_d^{n \to n'}$ containing a representative of each facet class in $\mc{F}_d^{n \to n'}$ (see Appendix \ref{Appendix-generator-facets}). \\
        \hline
        $(\mbf{G}_{\text{K}}^{n',k},\gamma_{\text{K}}^{n',k,d})$ & $k$-Guessing Facet & Tight Bell inequality for signaling polytopes (see Appendix \ref{appendix-k-guessing-facets}).   \\
        \hline
        $(\mbf{G}_{\text{ML}}^{n'},d)$ & Maximum Likelihood Facet & Tight Bell inequality for signaling polytopes (see Appendix \ref{Appendix-ml-facets}). \\
        \hline
        $(\mbf{G}_{?}^{n',d}, \gamma^{n',d}_{?})$ & Ambiguous Guessing Facet & Tight Bell inequality for signaling polytopes (see Appendix \ref{Appendix-ambiguous-guessing-facets}). \\
        \hline
        $(\mbf{G}_{\text{A}}^{\varepsilon,m'}, \gamma_{\text{A}}^{\varepsilon,d})$ & Anti-Guessing Facet & Tight Bell inequality for signaling polytopes (see Appendix \ref{appendix-anti-guessing-facets}). \\
        \hline
    \end{tabular}
    \caption{\linespread{1}\selectfont{\small Notation used throughout this work.}}
\end{table}

\section{Signaling Polytope Structure} \label{Appendix-polytope-structure} 

In this section we provide details about the structure of signaling polytopes (see Definition \ref{Def-signaling-polytope}).
The signaling polytope, denoted by $\mc{C}_d^{n \to n'}$, is a subset of $\mc{P}^{n \to n'}$.
Therefore, a channel $\mbf{P}\in\mc{C}_d^{n \to n'}$ has matrix elements $P(y|x)$ subject to the constraints of non-negativity $P(y|x) \geq 0$ and normalization $\sum_{y\in[n']}P(y|x) = 1$ for all $y\in[n']$ and $x\in[n]$.
Furthermore, since channels $\mbf{P}\in\mc{C}_d^{n \to n'}$ are permitted the use of shared randomness, the set $\mc{C}_d^{n \to n'}$ is convex.

In the two extremes of communication, the signaling polytope admits a simple structure.
For maximum communication, $d=\min\{n,n'\}$, any  channel $\mbf{P}\in\mc{P}^{n \to n'}$ can be realized, hence $\mc{C}_{\min\{n,n'\}}^{n \to n'} = \mc{P}^{n \to n'}$.
For no communication, $d=1$, Bob's output $y$ is independent from Alice's input $x$  meaning that $P(y|x) = P(y|x')$ for any choice of $x,x'\in[n]$ and $y\in[n']$.
This added constraint simplifies the signaling polytope $\mc{C}_1^{n \to n'}$ to $\mc{P}^{1 \to n'}$ which is formally an $n'$-simplex \cite{Ziegler-2012a}.
For all other cases, $\min\{n,n'\} > d > 1$, the signaling polytope $\mc{C}_d^{n \to n'}$ takes on a more complicated structure.

\subsection{Vertices}

The vertices of the signaling polytope are denoted by $\mc{V}_d^{n \to n'}$.
Signaling polytopes are convex and therefore described as the convex hull of their vertices, $\mc{C}_d^{n \to n'} = \conv(\mc{V}_d^{n \to n'})$.
As noted in the main text, a vertex $\mbf{V}\in\mc{V}_d^{n \to n'}$ is an $n'\times n$ column stochastic matrices with 0/1 elements and rank $\rank(\mbf{V})\leq d$.
For instance,

\begin{equation}
    \mbf{V}=\begin{bmatrix}1&0&0&1\\0&0&0&0\\0&1&1&0\end{bmatrix}
\end{equation}

\noindent is a vertex $\mbf{V}\in\mc{V}_2^{4\to 3}$.
Naturally, each vertex $\mbf{V}\in\mc{V}_d^{n \to n'}$ has no more than $d$ nonzero rows.
A straightforward counting argument shows that $\mc{V}^{n\to n'}_d$ contains $\sum_{c=1}^d \genfrac{\{}{\}}{0pt}{1}{n}{c}\binom{n'}{c} c!$ vertices (see Supplemental Material of Ref. \cite{Dall'Arno_no_hypersignaling_2017}), where $\genfrac{\{}{\}}{0pt}{1}{n}{c}$ denotes Stirling's number of the second kind and $\binom{n'}{c}$ a binomial coefficient.
An important observation is that number of vertices in $\mc{V}_d^{n \to n'}$ grows exponentially in  the number of inputs, $n$, and factorially in the number of outputs, $n'$.
The large number of  vertices represents a key challenge in characterizing the signaling polytope.

\subsection{Polytope Dimension}

The dimension of the signaling polytope $\dim(\mc{C}_d^{n \to n'})\leq \dim(\mc{P}^{n \to n'}) = n(n'-1)$.
This upper bound follows from the facts that $\mc{C}_d^{n \to n'} \subseteq \mc{P}^{n \to n'}$ and any $\mbf{P} \in\mc{P}^{n \to n'}$ must satisfy $n$ normalization constraints, one for each column of $\mbf{P}$.
Naively, $\mc{P}^{n \to n'}\subset\mbb{R}^{n'\times n}$ where $\dim(\mbb{R}^{n'\times n}) = n(n')$, however, the $n$ normalization constraints restrict $\mc{P}^{n \to n'}$ to $\dim(\mc{P}^{n \to n'}) = n(n'-1)$. 
To evaluate the dimension of $\mc{C}_d^{n \to n'}$ with greater precision, the number of affinely independent vertices in $\mc{V}_d^{n \to n'}$ can be counted where $\dim(\mc{C}_d^{n \to n'})$ is one less than the number of affinely independent vertices.
When $d\geq 2$, one can count $n(n'-1)+1$ affinely independent vertices in $\mc{V}_d^{n \to n'}$, therefore, $\dim(\mc{C}^{n\to n'}_d)=n(n'-1)$.
In the remaining case of $d=1$, each of the $n'$ vertices are affinely independent and $\dim(\mc{C}_1^{n \to n'}) = n'-1$.
This result is not surprising because, as noted before, $\mc{C}_1^{n \to n'} = \mc{P}^{1 \to n'}$ and $\dim(\mc{P}^{1 \to n'}) = n'-1$.

\subsection{Facets}

A linear Bell inequality is represented as a tuple $(\mbf{G},\gamma)$ with $\mbf{G}\in \mbb{R}^{n'\times n}$ and $\gamma\in\mbb{R}$ where the inequality $\langle\mbf{G},\mbf{P}\rangle = \sum_{x,y}G_{y,x}P(y|x) \leq \gamma$ is formed by the Euclidean inner product with a channel $\mbf{P}\in\mc{P}^{n \to n'}$.
For convenience, we identify two  polyhedra of channels 

\begin{align}
    \mc{C}(\mbf{G},\gamma)&:=\{\mbf{P}\in\mc{P}^{n\to n'}\;|\;\langle\mbf{G},\mbf{P}\rangle\leq \gamma\},
    \\
    \mc{F}(\mbf{G},\gamma)&:=\{\mbf{P}\in\mc{P}^{n\to n'}\;|\;\langle\mbf{G},\mbf{P}\rangle= \gamma\}.
\end{align}

\lemma \label{lemma-facet-proof}

An inequality $(\mbf{G},\gamma)$ is a tight Bell inequality of the $\mc{C}_d^{n \to n'}$ signaling polytope \textit{iff}

    \begin{enumerate}
        \item \label{facet-proof-bound-condition}  $\mc{C}_d^{n \to n'}\subset \mc{C}(\mbf{G},\gamma) $;
        \item \label{facet-proof-aff-ind-condition}$\dim\left(\mc{C}^{n\to n'}_d\cap \mc{F}(\mbf{G},\gamma)\right)=\dim\left(\mc{C}_d^{n \to n'}\right)-1$.
    \end{enumerate}

\noindent Condition \ref{facet-proof-bound-condition} requires that Bell inequality $(\mbf{G},\gamma)$ contains all channels $\mbf{P}\in\mc{C}_d^{n \to n'}$ while Condition \ref{facet-proof-aff-ind-condition} requires that inequality $(\mbf{G},\gamma)$ is both a proper half-space and a \textit{facet} of $\mc{C}_d^{n \to n'}$.
Tight Bell inequalities and facets are closely related and described by the same inequality $(\mbf{G},\gamma)$.
The key difference is that a tight Bell inequality is a half-space inequality $\langle\mbf{G},\mbf{P}\rangle \leq \gamma$ whereas a facet is the polytope $\mc{C}_d^{n \to n'}\cap\mc{F}(\mbf{G},\gamma)$.
The complete set of signaling polytope facets is denoted by $\mc{F}_d^{n \to n'}$ and the signaling polytope is simply the intersection of all tight Bell inequalities $(\mbf{G}_m,\gamma_m)\in\mc{F}_d^{n \to n'}$,

\begin{equation} 
    \mc{C}_d^{n \to n'}  = \bigcap_{m=1}^r \mc{C}(\mbf{G}_m,\gamma_m),
\end{equation}

\noindent The number of facet inequalities $r$ is typically larger than the set of vertices $\mc{V}_d^{n \to n'}$ presenting another challenge in the characterization of signaling polytopes.

\begin{remark}
    A given Bell inequality $(\mbf{G},\gamma)\in\mc{F}_d^{n \to n'}$ does not have a unique form.
    Therefore, it is convenient to establish a \textit{normal form} for a given facet inequality \cite{christof2001decomposition}.
    First, observe that multiplying an inequality $(\mbf{G},\gamma)$ by a scalar $a\in\mbb{R}$ does not change the inequality, that is, $\mc{C}(\mbf{G}, \gamma) = \mc{C}(a\mbf{G},a(\gamma))$.
    Second, observe that the vertices in $\mc{V}_d^{n \to n'}$ have 0/1 elements and the rational arithmetic in Fourier-Motzkin elimination \cite{Ziegler-2012a,PORTA} results in the matrix coefficients of $\mbf{G}$ being rational.
    Therefore, there exists a rational scalar $a$ such that $a G_{y,x}$ and $a\gamma$ are integers for all $x\in[n]$ and $y\in[n']$.
    Third, observe that the normalization and non-negativity constraints for channels $\mbf{P}\in\mc{P}^{n \to n'}$ allows the equivalence between the following two inequalities
    
    \begin{equation}\label{Eq:facet-normal-form-normalization}
        \gamma \geq \langle \mbf{G},\mbf{P} \rangle \Longleftrightarrow
         \gamma + 1 \geq \langle \mbf{G},\mbf{P} \rangle + \sum_{y\in[n']}G_{y,x'}P(y|x')
    \end{equation}
    
    \noindent for any $x'\in[n]$.
    Therefore, it is always possible to find a form of inequality $(\mbf{G},\gamma)$ where $G_{y,x}\geq 0$ for all $y\in[n']$ and $x\in[n]$.
    Hence we define a normal form for any tight Bell inequality $(\mbf{G},\gamma)\in\mc{F}_d^{n \to n'}$:
    
    \begin{itemize}
        \item Inequality $(\mbf{G},\gamma)$ is scaled such that $\gamma$ and all $G_{y,x}$ are integers with a greatest common factor of 1.
        \item Normalization constraints are added or subtracted from all columns using Eq. \eqref{Eq:facet-normal-form-normalization} such that $G_{y,x}\geq 0$ and the smallest element in each column of $\mbf{G}$ is zero.
    \end{itemize}
\end{remark}

\subsection{Permutation Symmetry} \label{Appendix-permutation-symmetry}

The input and output values $x$ and $y$ are merely labels for a channel $\mbf{P}\in\mc{P}^{n \to n'}$, therefore, swapping labels $x\leftrightarrow x'$ and $y \leftrightarrow y'$ where $x,x'\in[n]$ and $y,y'\in[n']$ does not affect $\mc{P}^{n \to n'}$ \cite{Rosset2014}.
The relabeling operation is implemented using elements from the set of doubly stochastic $k\times k$ permutation matrices $S_k$.
For example,

\begin{equation} \label{eq:channel-permutation}
    \mbf{P}' = \pi_{\mc{Y}}\mbf{P}\pi_{\mc{X}}, \quad \text{where} \quad \mbf{P},\mbf{P}'\in\mc{P}^{n \to n'},
\end{equation}

\noindent $\pi_\mc{X}\in S_{n}$, and $\pi_\mc{Y}\in S_{n'}$.
Note that permuting the rows or columns of a matrix cannot change the rank of a matrix, therefore, if $\mbf{V}\in\mc{V}_d^{n \to n'}$ and $\mbf{V}' = \pi_\mc{Y}\mbf{V}\pi_\mc{X}$, then $\mbf{V}'\in\mc{V}_d^{n \to n'}$.
It follows that this permutation symmetry holds for any channel in the signaling polytope, $\mbf{P},\mbf{P}'\in\mc{C}_d^{n \to n'}$ where $\mbf{P}'$ is a permutation of $\mbf{P}$.
Likewise, a facet inequality $(\mbf{G},\gamma)\in\mc{F}_d^{n \to n'}$ can be permuted into a new facet inequality $(\mbf{G}',\gamma)\in\mc{F}_d^{n \to n'}$ where $\mbf{G}'=\pi_\mc{Y}\mbf{G}\pi_\mc{X}$.

\subsection{Generator Facets} \label{Appendix-generator-facets}

Permutation symmetry motivates the notion of a \textit{facet class} defined as a collection of facet inequalities formed by taking all permutations of a canonical facet $(\mbf{G}^{\star},\gamma)\in\mc{F}_d^{n \to n'}$ which we refer to as a \textit{generator facet}.
The canonical facet is arbitrary thus we define the generator facet as the \textit{lexicographic normal form} \cite{christof2001decomposition,Rosset2014} of the facet class.
The set of generator facets, denoted by $\mc{G}_d^{n \to n'}:=\{(\mbf{G}_i^{\star},\gamma_i)\}_{i=1}^{r'}$, is the subset of $\mc{F}_d^{n \to n'}$ containing the generator facet of each facet class bounding $\mc{C}_d^{n \to n'}$.
Since the number of input and output permutations scale as factorials of $n$ and $n'$ respectively, the set of generator facets is considerably smaller than $\mc{F}_d^{n \to n'}$ and therefore, provides a convenient simplification to $\mc{F}_d^{n \to n'}$.
To recover the complete set of facets from $\mc{G}_d^{n \to n'}$, we take all row and column permutations of each generator facet $(\mbf{G}^{\star},\gamma)\in\mc{G}_d^{n \to n'}$.
As a final remark, we note that $\mc{V}_d^{n \to n'}$ can also be reduced to a set of generator vertices, however, this set is not required for our current discussion of signaling polytopes.







\section{Adjacency Decomposition} \label{Appendix-Adjacency}

This section provides an overview of the adjaceny decomposition technique \cite{christof2001decomposition}.
In our work, we use an adjacency decomposition algorithm to compute the generator facets of the signaling polytope.
Our implementation can be found in our supporting software \cite{SignalingDimension.jl}.
The adjacency decomposition provides a few key advantages in the computation of Bell inequalities:

\begin{enumerate}
    \item The algorithm stores only the generator facets $\mc{G}_d^{n  \to n'}$ instead of the complete set of facets $\mc{F}_d^{n \to n'}$. This considerably reduces the required memory.
    \item New generator facets are derived in each iteration of the computation, hence, the algorithm does not need to run to completion to provide value.
    \item The algorithm can be widely parallelized \cite{christof2001decomposition}.
\end{enumerate}

\subsection{Adjacency Decomposition Algorithm}

The adjacency decomposition is an iterative algorithm which requires as input the signaling polytope vertices $\mc{V}_d^{n \to n'}$ and a seed generator facet $(\mbf{G}^{\star}_{\text{seed}},\gamma_{\text{seed}})\in\mc{G}_d^{n \to n'}$.
The algorithm maintains a list of generator facets $\mc{G}_{\text{list}}$ where each facet $(\mbf{G}^{\star},\gamma)\in\mc{G}_{\text{list}}$ is marked either as \textit{considered} or \textit{unconsidered}.
The generator facet is defined as the lexicographic normal form of the facet class \cite{christof2001decomposition,Rosset2014}.
Before the algorithm begins, $(\mbf{G}^{\star}_{\text{seed}},\gamma_{\text{seed}})$ is added to $\mc{G}_{\text{list}}$ and marked as \textit{unconsidered}.
In each iteration, the algorithm proceeds as follows \cite{christof2001decomposition}:

\begin{enumerate}
    \item An \textit{unconsidered} generator facet $(\mbf{G}^{\star},\gamma)\in\mc{G}_{\text{list}}$ is selected.
    \item All facets adjacent to $(\mbf{G}^{\star},\gamma)$ are computed.
    \item Each adjacent facet is converted into its lexicographic normal form.
    \item Any new generator facets identified are marked as \textit{unconsidered} and added to $\mc{G}_{\text{list}}$.
    \item Facet $(\mbf{G}^{\star},\gamma)$ is marked as \textit{considered}.
\end{enumerate}

\noindent The procedure repeats until all facets in $\mc{G}_{\text{list}}$ are marked as \textit{considered}.
If run to completion, then $\mc{G}_{\text{list}}=\mc{G}_d^{n \to n'}$ and all generator facets of the signaling polytope $\mc{C}_d^{n \to n'}$ are identified.
The algorithm is guaranteed to find all generator facets due to the permutation symmetry of the signaling polytope.
By this symmetry, any representative of a given facet class has the same fixed set of facet classes adjacent to it.
For the permutation symmetry to hold for all facets in the signaling polytope, there cannot be two disjoint sets of generator facets where the members of one set do not lie adjacent to the members of the other.

The inputs of the adjacency decomposition are easy to produce computationally.
A seed facet can always be constructed using the lifting rules for signaling polytopes (see Fig. \ref{Fig:input-output-lifting}) and
the signaling polytope vertices $\mc{V}_d^{n \to n'}$ can be easily computed (see supporting software \cite{SignalingDimension.jl}).
Note, however, that the exponential growth of $\mc{V}_d^{n \to n'}$ eventually hinders the performance of the adjacency decomposition algorithm.

\subsection{Facet Adjacency}

A key step in the adjacency decomposition algorithm is to compute the set of facets adjacent to a given facet $(\mbf{G},\gamma)$.
In this section, we define facet adjacency and outline the method used to compute the adjacent facets.

\begin{lemma} Two facets $(\mbf{G}_1,\gamma_1),(\mbf{G}_2,\gamma_2)\in\mc{F}_d^{n \to n'}$ are adjacent \textit{iff} they share a \textit{ridge} $\mc{H}$ defined as:

\begin{enumerate}
    \item $\mc{H} := \mc{F}(\mbf{G}_1,\gamma_1)\cap\mc{F}(\mbf{G}_2,\gamma_2)\cap\mc{C}_d^{n \to n'}$,
    \item where $\dim(\mc{H}) = \dim(\mc{C}_d^{n \to n'}) - 2$.
\end{enumerate}

\end{lemma}

\noindent A ridge can be understood as a facet of the facet polytope $\mc{C}_d^{n\to n'}\cap\mc{F}(\mbf{G},\gamma)$.
Therefore, to compute the ridges of a given facet $(\mbf{G},\gamma)\in\mc{F}_d^{n \to n'}$ we take the typical approach for computing facets.
Namely, the set of vertices $\{\mbf{V}\in\mc{V}_d^{n \to n'}| \; \langle\mbf{G},\mbf{V}\rangle = \gamma \}$ is constructed and PORTA \cite{PORTA,XPORTA.jl} is used to compute the ridges of $(\mbf{G},\gamma)$.
A facet adjacent to $(\mbf{G},\gamma)$ is computed from each ridge using a \textit{rotation algorithm} described by Christof and Reinelt \cite{christof2001decomposition}.
Given the signaling polytope vertices $\mc{V}_d^{n\to n'}$, this procedure computes the complete set of facets adjacent to $(\mbf{G},\gamma)$.

\section{Tight Bell Inequalities} \label{Appendix-facet-classes}

In this section we discuss the general forms for each of the signaling polytope facets in Fig. \ref{Fig:6-2-4_facets}.
Each facet class is described by a generator facet (see Appendix \ref{Appendix-permutation-symmetry}) where
all permutations and input/output liftings of these generator facets are also tight Bell inequalities.
To prove that an inequality $(\mbf{G},\gamma)$ is a facet of $\mc{C}_d^{n \to n'}$, both conditions of Lemma \ref{lemma-facet-proof} must hold.
The proofs contained by this section verify Condition \ref{facet-proof-aff-ind-condition} of Lemma \ref{lemma-facet-proof} by constructing a set of $\dim(\mc{C}_d^{n \to n'})= n(n'-1)$ affinely independent $\{\mbf{V}\in\mc{V}_d^{n \to n'}|\; \langle\mbf{G},\mbf{V}\rangle = \gamma\}$.
These enumerations are verified numerically in our supporting software \cite{SignalingDimension.jl}.
To assist with the enumeration of affinely independent vertices, we introduce a simple construction for affinely independent vectors with 0/1 elements.

\begin{lemma} \label{lemma-aff-ind-enumeration}
    Consider an $n$-element binary vector $\vec{b}_k\in\{0,1\}^n$ with $n_0$  null elements and $n_1$ unit elements where $n_0+n_1 = n$.
    A set of $n$ affinely independent vectors $\{\vec{b}_k\}_{k=1}^n$ is constructed as follows:

    \begin{itemize}
        \item Let $\vec{b}_1$ be the binary vector where the first $n_0$ elements are null and the next $n_1$ elements are unit values.
        \item For $k\in[2,n_0+1]$, $\vec{b}_k$ is derived from $\vec{b}_1$ by swapping the unit element at index $(n_0+1)$ with the null element at index $(k-1)$.
        \item For $k\in[n_0+2,n]$, $\vec{b}_k$ is derived from $\vec{b}_1$ by swapping the null element at index $n_0$ with the unit element at index $k$.
    \end{itemize}

    \noindent For example, when $n=5$, $n_0=2$, and $n_1=3$ the enumeration yields
    
    \begin{align}
        \left\{ \vec{b}_1 = [0,0,1,1,1], \quad \vec{b}_2 = [1,0,0,1,1], \quad \vec{b}_3 = [0,1,0,1,1], \quad \vec{b}_4 = [0,1,1,0,1], \quad \vec{b}_5 = [0,1,1,1,0] \right\}.
    \end{align}
    
    \begin{proof}
        To verify the affine independence of $\{\vec{b}\}_{k=1}^n$ it is sufficient to show the linear independence of $\{\vec{b}_1 - \vec{b_k} \}_{k=2}^n$.
        Note that each $(\vec{b}_1-\vec{b}_k)$ has two nonzero elements, one of which occurs at an index that is zero for all $(\vec{b}_1 - \vec{b}_{k'})$ where $k\neq k'$.
        Therefore, the vectors in $\{\vec{b}_1 - \vec{b_k} \}_{k=2}^n$ are linearly independent and  $\{\vec{b}_k\}_{k=1}^n$ is affinely independent.
    \end{proof}

        
        
\end{lemma}

\subsection{\textit{k}-Guessing Facets} \label{appendix-k-guessing-facets}

Consider a guessing game with $k$ correct answers out of $n'$ possible answers.
In this game, Alice has $n=\binom{n'}{k}$ inputs where each value $x$ corresponds to a unique set of $k$ correct answers.
Given an input $x\in[n]$, Alice signals to Bob using a message $m\in[d]$ and Bob makes a guess $y\in[n']$.
A correct guess scores 1 point while an incorrect guess scores 0 points.
This type of guessing game is described by Heinosaari \textit{et al.} \cite{Heinosaari2019,Heinosaari-2020a} and used to test the communication performance of a particular theory.
In this work, we treat this \textit{$k$-guessing game} as a Bell inequality $(\mbf{G}_{\text{K}}^{n',k},\gamma_{\text{K}}^{n',k,d})$ of the signaling polytope $\mc{C}_d^{n \to n'}$ where 

\begin{equation}
    \gamma_{\text{K}}^{n',k,d} = \binom{n'}{k} - \binom{n'-d}{k}
\end{equation}

\noindent and $\mbf{G}_{\text{K}}^{n',k}\in\mbb{R}^{n' \times \binom{n'}{k}}$ is a matrix with each column containing a unique distribution of $k$ unit elements and $(n'-k)$ null elements.
For example,

\begin{equation}
    \mbf{G}_{\text{K}}^{6,2} = \begin{bmatrix}
            1 & 1 & 1 & 1 & 1 & 0 & 0 & 0 & 0 & 0 & 0 & 0 & 0 & 0 & 0 \\
            1 & 0 & 0 & 0 & 0 & 1 & 1 & 1 & 1 & 0 & 0 & 0 & 0 & 0 & 0 \\
            0 & 1 & 0 & 0 & 0 & 1 & 0 & 0 & 0 & 1 & 1 & 1 & 0 & 0 & 0 \\
            0 & 0 & 1 & 0 & 0 & 0 & 1 & 0 & 0 & 1 & 0 & 0 & 1 & 1 & 0 \\
            0 & 0 & 0 & 1 & 0 & 0 & 0 & 1 & 0 & 0 & 1 & 0 & 1 & 0 & 1 \\
            0 & 0 & 0 & 0 & 1 & 0 & 0 & 0 & 1 & 0 & 0 & 1 & 0 & 1 & 1 \\
        \end{bmatrix}.
\end{equation}

\noindent This general Bell inequality for signaling polytopes was identified by Frenkel and Weiner \cite{Frenkel2015}, who showed that given a channel  $\mbf{P}\in\mc{C}_d^{n \to n'}$, the bounds of this inequality are
    
\begin{equation} \label{Eq:generalized-guessing-game}
    \binom{n'}{k} - \binom{n'-d}{k}\geq \langle \mbf{G}_{\text{K}}^{n',k}, \mbf{P}\rangle \geq  \binom{n'-d}{n'-k}.
\end{equation}

\noindent However, we only focus on the upper  bound $\gamma_{\text{K}}^{n',k,d}$.
We now show conditions for which $(\mbf{G}_{\text{K}}^{n',k},\gamma_{\text{K}}^{n',k,d})\in\mc{F}_d^{n \to n'}$.

\begin{proposition} \label{prop-k-guessing-game}
    The inequality $(\mbf{G}_{\text{K}}^{n',k},\gamma_{\text{K}}^{n',k,d})$ is a facet of $\mc{C}_d^{n \to n'}$ with $n = \binom{n'}{k}$, $n'-2 \geq k \geq 1$, and $d = n'-k$. 

    \begin{proof}
        
        To prove that $(\mbf{G}_{\text{K}}^{n',k}, \gamma_{\text{K}}^{n',k,d})$ is a facet  of $\mc{C}_d^{n \to n'}$ we construct a set of $\dim(\mc{C}_d^{n \to n'})=(n'-1)\binom{n'}{k}$ affinely independent vertices $\{\mbf{V}\in\mc{V}_d^{n \to n'} |\;\gamma_{\text{K}}^{n',k,d} = \langle  \mbf{G}_{\text{K}}^{n',k}, \mbf{V}\rangle\}$.
        Observe that separating the first row from the rest of $\mbf{G}_{\text{K}}^{n',k}$ results in a block matrix of form,
        
        \begin{equation}
            \mbf{G}_{\text{K}}^{n',k} = \left[\begin{array}{c|c}
                \vec{1} & \vec{0}  \\
                \hline
                \mbf{G}_{\text{K}}^{(n'-1),(k-1)} & \mbf{G}_{\text{K}}^{(n'-1),k}\\
            \end{array}\right],\quad e.g. \quad \mbf{G}_{\text{K}}^{5,2} = \left[\begin{array}{cccc|cccccc}
                1 & 1 & 1 & 1 & 0 & 0 & 0 & 0 & 0 & 0 \\
                \hline 
                1 & 0 & 0 & 0 & 1 & 1 & 1 & 0 & 0 & 0 \\
                0 & 1 & 0 & 0 & 1 & 0 & 0 & 1 & 1 & 0 \\
                0 & 0 & 1 & 0 & 0 & 1 & 0 & 1 & 0 & 1 \\
                0 & 0 & 0 & 1 & 0 & 0 & 1 & 0 & 1 & 1 \\
            \end{array}\right] = \left[\begin{array}{c|c}
                \vec{1} & \vec{0}  \\
                \hline
                \mbf{G}_{\text{K}}^{4,1} & \mbf{G}_{\text{K}}^{4,2}\\
            \end{array}\right],
        \end{equation}
        
        \noindent where $\vec{0}$ and $\vec{1}$ are row vectors containing zeros and ones, and we refer to $\mbf{G}_{\text{K}}^{(n'-1),(k-1)}$ and $\mbf{G}_{\text{K}}^{(n'-1),k}$ as left and right $k$-guessing blocks respectively.
        The left and right $k$-guessing blocks suggest a recursive approach to our construction of affinely independent vertices.
        Namely, we construct $\binom{n'}{k}$ vertices by targeting the first row of $\mbf{G}_{\text{K}}^{n',k}$ while Proposition \ref{prop-k-guessing-game} is recursively applied to enumerate the remaining vertices using the left and right $k$-guessing blocks.
        The recursion requires two base cases to be addressed:
        
        \begin{enumerate}
            \item When $d=2$ and  $n'=k + d$, the construction of affinely independent vertices is described in Proposition \ref{prop-k-guessing-facet-d=2}.
            \item When $k=1$, the construction of affinely independent  vertices is described in Proposition \ref{prop-ml-facet}.
        \end{enumerate}
        
        \noindent An iteration of this recursive construction proceeds as follows.

        First, we construct an affinely independent vertex for each of the $\binom{n'}{k}$ elements in the first row of $\mbf{G}_{\text{K}}^{n',k}$.
        For each index $x'_1$ in the $\vec{1}$ block, a vertex $\mbf{V}_1$ is constructed by setting all $V_1(1|x)=1$ where $x\neq x'_1$ and $V_1(y|x'_1)=1$ where $y>1$ is the smallest row index such that $G_{y,x'_1}=1$.
        The remaining rows of $\mbf{V}_1$ are filled to maximize the right $k$-guessing block.
        Then, for each index $x'_0$ in the $\vec{0}$ block, a vertex $\mbf{V}_0$ is constructed by setting $V_0(1|x'_0)=1$ and all $V_0(1|x)=1$ where $G_{1,x} = 1$.
        The remaining $(d-1)$ rows of $\mbf{V}_0$ are filled to maximize the right $k$-guessing block.
        This procedure enumerates $\binom{n'}{k}$ affinely independent vertices.
        
        Then, the remaining $(n'-2)\binom{n'}{k}$ vertices are found by individually targeting the left and right $k$-guessing blocks.
        To construct a vertex $\mbf{V}_L$ using the left block $\mbf{G}_{\text{K}}^{(n'-1)),(k-1)}$, the first row of $\mbf{V}_L$ is not used.
        The left block is then a $(k-1)$-guessing  game with $(n'-1)$ outputs where $d= (n'-1)-(k-1) = n'-k$, hence, Proposition \ref{prop-k-guessing-game} holds and  $(n'-2)\binom{n'-1}{k-1}$ affinely independent vertices are enumerated using the described recursive process.
        Note that for each vertex of form $\mbf{V}_L$, the remaining elements are filled to maximize the right $k$-guessing block $\mbf{G}_{\text{K}}^{(n'-1),k}$.
        Similarly, to construct a vertex $\mbf{V}_R$ using the right block $\mbf{G}_{\text{K}}^{(n'-1),k}$, we set all elements $V_R(1|x)=1$ where $G^{n',k}_{1,x} =1$.
        The remaining $(d-1)$ rows of $\mbf{V}_R$ are filled by  optimizing the $\mbf{G}_{\text{K}}^{(n'-1),k}$ block.
        Since $d=n'-k$ and $(d-1) = (n'-1)-k$, Proposition \ref{prop-k-guessing-game} holds, and recursively applying this procedure constructs $(n'-2)\binom{n'-1}{k}$ vertices of form $\mbf{V}_R$ using the right $k$-guessing block.
        
        Finally, vertices of forms $\mbf{V}_0,$ $\mbf{V}_1$, $\mbf{V}_L$ and $\mbf{V}_R$ are easily verified to be affinely independent.
        Summing these vertices yields $(n'-2)\binom{n'-1}{k-1} + (n'-2)\binom{n'-1}{k} +\binom{n'}{k} = (n'-1)\binom{n'}{k}$ affinely independent vertices, therefore, the  $k$-guessing Bell inequality is proven to be tight when $n'=k+d$.
    \end{proof}
\end{proposition}

\begin{proposition} \label{prop-k-guessing-facet-d=2}
    The \textit{$k$-guessing game} Bell inequality $(\mbf{G}_{\text{K}}^{n',k},\gamma_{\text{K}}^{n',k,d})$ is a tight Bell inequality of all signaling polytopes $\mc{C}_d^{n \to n'}$ with $n = \binom{n'}{k}$, $d=2$, and $k=n'-2$.    
    
    \begin{proof}
        To prove the tightness we construct a set containing $(n'-1)\binom{n'}{k}$ affinely independent vertices $\{\mbf{V}\in\mc{V}_2^{n \to n'}\; | \; \langle\mbf{G}_{\text{K}}^{n',(n'-2)},\mbf{V}\rangle = \binom{n'}{n'-2} -1 \}.$ 
        To help illustrate this proof, we use the example of $(\mbf{G}_{\text{K}}^{5,3}, \gamma_{\text{K}}^{5,3,2})$ where
            
            \begin{equation}
                \mbf{G}_{\text{K}}^{5,3} = \begin{bmatrix}
                    1 & 1 & 1 & 1 & 1 & 1 & 0 & 0 & 0 & 0 \\
                    1 & 1 & 1 & 0 & 0 & 0 & 1 & 1 & 1 & 0 \\
                    1 & 0 & 0 & 1 & 1 & 0 & 1 & 1 & 0 & 1 \\
                    0 & 1 & 0 & 1 & 0 & 1 & 1 & 0 & 1 & 1 \\
                    0 & 0 & 1 & 0 & 1 & 1 & 0 & 1 & 1 & 1 \\
                \end{bmatrix}.
            \end{equation}
    
            \noindent and $\gamma_{\text{K}}^{5,3,2} = \binom{5}{3}-1$.
            Since $d=2$, we consider vertices $\mbf{V}\in\mc{V}_d^{n \to n'}$ with $\rank(\mbf{V})=2$ where each vertex uses two rows $y$ and $y'$ where $y< y'$.
            In general, each of the $\binom{n'}{2}$ two-row selections from $\mbf{G}_{\text{K}}^{n',(n'-2)}$ have a unique column $x_0$ containing null elements both rows $y$ and $y'$.
            Therefore, for each unique pair $y$ and $y'$, two affinely independent vertices $\mbf{V}_1$ and $\mbf{V}_2$ are constructed by setting $V_1(y|x_0)=1$ and $V_2(y'|x_0)=1$ while the remaining terms are arranged such that all unit elements in row $y$ and the remaining elements in row $y'$ are selected to achieve the optimal score.
            Performing this procedure for the first two rows of $\mbf{G}_{\text{K}}^{5,3}$ ($y=1$ and $y'=2$) constructs the vertices
            
            \begin{equation}
                \mbf{V}_1 = \begin{bmatrix}
                        1 & 1 & 1 & 1 & 1 & 1 & 0 & 0 & 0 & 1 \\
                        0 & 0 & 0 & 0 & 0 & 0 & 1 & 1 & 1 & 0 \\
                        0 & 0 & 0 & 0 & 0 & 0 & 0 & 0 & 0 & 0 \\
                        0 & 0 & 0 & 0 & 0 & 0 & 0 & 0 & 0 & 0 \\
                        0 & 0 & 0 & 0 & 0 & 0 & 0 & 0 & 0 & 0 \\
                \end{bmatrix}, \quad \mbf{V}_2 = \begin{bmatrix}
                        1 & 1 & 1 & 1 & 1 & 1 & 0 & 0 & 0 & 0 \\
                        0 & 0 & 0 & 0 & 0 & 0 & 1 & 1 & 1 & 1 \\
                        0 & 0 & 0 & 0 & 0 & 0 & 0 & 0 & 0 & 0 \\
                        0 & 0 & 0 & 0 & 0 & 0 & 0 & 0 & 0 & 0 \\
                        0 & 0 & 0 & 0 & 0 & 0 & 0 & 0 & 0 & 0 \\
                \end{bmatrix}
            \end{equation}
            
            \noindent where $x_0=10$ in this example.
            Repeating this procedure for each of the $\binom{n'}{2}$ row selections produces, $2\binom{n'}{2} = 2\binom{n'}{k}$ affinely independent vertices,
            one for each null element in $\mbf{G}_{\text{K}}^{n',(n'-2)}$.
           
            The remaining vertices are constructed by selecting a target row $y\in[n'-1]$.
            In the target row, for each $x'$ where $G_{y,x'}=1$ a vertex $\mbf{V}_3$ is constructed by setting $V_3(y|x) = 1$ for all $x\neq x'$ that satisfy $G_{y,x}=1$.
            A secondary row $y'>y$ of $\mbf{V}_3$ is chosen where $y'$ is the smallest index satisfying $G_{y',x'}=1$.
            We then set $V(y'|x')=1$ while the remaining elements of $\mbf{V}_3$ are set to achieve the optimal score.
            For selected rows $y$ and $y'$, the null column at index $x_0$ is set in the target row  as $V_3(y|x_0) = 1$.
            For example, consider $\mbf{G}_{\text{K}}^{5,}$ with the target row as $y=1$ and $x'=4$ we construct the vertex,
            
            \begin{equation}
                \mbf{V}_3 = \begin{bmatrix}
                    1 & 1 & 1 & 0 & 1 & 1 & 0 & 0 & 1 & 0 \\
                    0 & 0 & 0 & 0 & 0 & 0 & 0 & 0 & 0 & 0 \\
                    0 & 0 & 0 & 1 & 0 & 0 & 1 & 1 & 0 & 1 \\
                    0 & 0 & 0 & 0 & 0 & 0 & 0 & 0 & 0 & 0 \\
                    0 & 0 & 0 & 0 & 0 & 0 & 0 & 0 & 0 & 0 \\
                \end{bmatrix}
            \end{equation}
            
            \noindent Note that all secondary row indices $y'\leq y+3$ are required to construct a vertex $\mbf{V}_3$ for each unit element in the target row $y$.
            Let $\Delta y = y'-y$, then $\sum_{\Delta y=1}^3\binom{n'-1-\Delta y}{d+1-\Delta y}$ vertices are constructed for target row $y$.
            For $y=n'-2$ and $y=n'-1$, the sum terminates at $\Delta y = 2$ and $\Delta y=1$ respectively because the vertices are only affinely independent if the secondary row has index $y' > y$.
            Thus, this process produces 
            
            \begin{align}
                \sum_{\Delta y=1}^3(n'-\Delta y)\binom{n'-1-\Delta y}{d+1-\Delta y} = (n'-3)\binom{n'}{d}
            \end{align}
            
            \noindent affinely independent vertices where the identities $\frac{l}{m}\binom{l}{m} = \binom{l-1}{m-1}$ and $\frac{l+1-m}{m}\binom{l}{m} = \binom{l}{m-1}$ are used to convert the binomial coefficients to the form $\binom{n'}{d}=\binom{n'}{k}$.
            Combining the vertices of form $\mbf{V}_1$, $\mbf{V}_2$, and $\mbf{V}_3$ yields a set of $2\binom{n'}{k} + (n'-3)\binom{n'}{k} = (n'-1)\binom{n'}{k}$ affinely independent vertices.
            Therefore, when $d=2$ and  $k=n'-2$, $(\mbf{G}_{\text{K}}^{n',(n'-2)},\binom{n'}{n'-2}-1)$ is a tight Bell inequality of the $\mc{C}_2^{\binom{n'}{k}\to n'}$ signaling polytope.
    \end{proof}
\end{proposition}

\subsection{Maximum Likelihood Facets} \label{Appendix-ml-facets}

In this section, we discuss the conditions for which maximum likelihood games (see main text) are tight Bell inequalities.
The maximum likelihood Bell inequality $(\mbf{G}_{\text{ML}}^{n'}, d)$ is a $(k=1)$-guessing game where $\mbf{G}_{\text{ML}}^{n'} = \mbf{G}_{\text{K}}^{n',1}$.
For simplicity, this section considers unlifted forms of $\mbf{G}_{\text{ML}}^{n'}$ is a $n'\times n'$ doubly stochastic matrix with 0/1 elements such as the $n'\times n'$ identity matrix.
For any vertex $\mbf{V}\in\mc{V}_d^{n \to n'}$,

\begin{align} \label{Eq:ML-appendix}
    \langle\mbf{G}_{\text{ML}}^{n'},\mbf{V}\rangle\leq d,
\end{align}

\noindent is satisfied because $\rank(\mbf{V}) \leq d$ and $\mbf{G}_{\text{ML}}^{n'}$ is doubly stochastic.
By the convexity of $\mc{C}_d^{n \to  n'}$, inequality \eqref{Eq:ML-appendix} must hold for all $\mbf{P}\in\mc{C}_d^{n\to n'}$.
We now discuss the conditions for  which $(\mbf{G}_{\text{ML}}^{n'},d)$ is a tight Bell inequality.

\begin{proposition} \label{prop-ml-facet}

The \textit{maximum likelihood} (ML) Bell inequality $(\mbf{G}_{\text{ML}}^{n'},d)$ is a facet of all signaling polytopes $\mc{C}_d^{n \to n'}$ with $n = n'$ and $n'>d>1$.

\begin{proof}
    To prove that $(\mbf{G}_{\text{ML}}^{n'},d)$ is a tight Bell inequality of $\mc{C}_d^{n \to n'}$ we construct a set of $\dim(\mc{C}_d^{n' \to n'})=n'(n'-1)$ affinely independent vertices $\{\mbf{V}\in\mc{V}_d^{n \to n'}|\; \langle\mbf{G}_{\text{ML}}^{n'},\mbf{P}\rangle = d\}$. 
    Taking $\mbf{G}_{\text{ML}}^{n'}$ to be the $n'\times n'$ identity matrix, a vertex $\mbf{V}$ satisfies $d = \langle \mbf{G}_{\text{ML}}^{n'}, \mbf{V}\rangle$ when $d$ unit elements of $\mbf{V}$ lie along the diagonal.
    In this case, $(n'-d)$ unit elements of $\mbf{V}$ can be freely distributed in the remaining columns of the $d$ selected rows.
    For simplicity, we place all free elements in a single row with index $y\in[n']$ which we refer to as the target row.
    In the target row, we set $V(y|y) =1$  while the off-diagonals, $V(y|x\neq y)$ with $x\in[n']$ contain $(n'-d)$ unit elements and $(d-1)$ null elements.
    Lemma \ref{lemma-aff-ind-enumeration} describes a construction of $(n'-1)$ affinely independent vectors $\{\vec{b}_k\}_{k\in[n'-1]}$ to set as the off-diagonals in the target row.
    Then, for each $x\in[n']$ where $V(y|x\neq y) = 0$, we set $V(x|x)=1$.
    This procedure obtains the upper bound in Eq. \eqref{Eq:ML-appendix} and constructs an affinely independent vertex for each of the $(n'-1)$ binary vectors  in $\{\vec{b}_k\}_{k\in[n'-1]}$.
    For example, targeting row $y=3$ of $\mbf{G}_{\text{ML}}^5$ when $d=3$ yields four vertices,
    
    \begin{equation}
        \mbf{V} \in \left\{\begin{bmatrix}
            1 & 0 & 0 & 0 & 0 \\
            0 & 1 & 0 & 0 & 0 \\
            0 & 0 & 1 & 1 & 1 \\
            0 & 0 & 0 & 0 & 0 \\
            0 & 0 & 0 & 0 & 0 \\
        \end{bmatrix},\quad\begin{bmatrix}
            1 & 0 & 0 & 0 & 0 \\
            0 & 0 & 0 & 0 & 0 \\
            0 & 1 & 1 & 0 & 1 \\
            0 & 0 & 0 & 1 & 0 \\
            0 & 0 & 0 & 0 & 0 \\
        \end{bmatrix},\quad
        \begin{bmatrix}
            1 & 0 & 0 & 0 & 0 \\
            0 & 0 & 0 & 0 & 0 \\
            0 & 1 & 1 & 1 & 0 \\
            0 & 0 & 0 & 0 & 0 \\
            0 & 0 & 0 & 0 & 1 \\
        \end{bmatrix}, \quad
        \begin{bmatrix}
            0 & 0 & 0 & 0 & 0 \\
            0 & 1 & 0 & 0 & 0 \\
            1 & 0 & 1 & 0 & 1 \\
            0 & 0 & 0 & 1 & 0 \\
            0 & 0 & 0 & 0 & 0 \\
        \end{bmatrix}\right\}.
    \end{equation}
    
    \noindent Repeating the procedure for each $y\in[n']$ results in $n'(n'-1)$ affinely independent vertices.
    The vertices enumerated for each target row $y$ are affinely independent from all other target rows because the free unit elements are only allowed in the target row.
    As a final note, this procedure does not work in the case where $d=1$ because there are only $n'$ vertices in $\mc{V}_d^{n \to n'}$ or the case where $d=n'$ because only one vertex $\mbf{V}=\mbf{G}_{\text{ML}}^{n'}$ maximizes Eq. \eqref{Eq:ML-appendix}.
    Since $n'(n'-1) = \dim(\mc{C}_d^{n'\to n'})$ affinely independent vertices are constructed, $(\mbf{G}_{\text{ML}}^{n'},d)$ is proven to be a tight Bell inequality of all signaling polytopes with $n'> d > 1$.
\end{proof}
\end{proposition}

\subsection{Ambiguous Guessing Facets} \label{Appendix-ambiguous-guessing-facets}

In this section we discuss the conditions for which ambiguous guessing games (see main text) are tight Bell inequalities.
Consider the \textit{ambiguous guessing} Bell inequality $(\mbf{G}_?^{n',d}, \gamma_?^{n',d})$ where $\mbf{G}_{?}^{n',d}\in\mbb{R}^{n'\times (n'-1)}$,

\begin{equation}\label{Eq:ambiguous-game}
    \mbf{G}_?^{n',d} = \sum_{x\in[n'-1]} (n'-d)\op{x}{x} + \op{n'}{x},\quad \text{and} \quad \gamma^{n',d}_? = d(n'-d).
\end{equation}

\noindent This Bell inequality is  best considered as a combination between a 1-guessing game for which  a correct answer provides $(n'-d)$ points extended with an ambiguous row for  which 1 point is scored for choosing the ambiguous output. 
For example when $n' = 6$ and $d=2$ we have

\begin{equation}
    \mbf{G}_?^{n',d} = \left[\begin{array}{c}
        (n'-d)\mbf{G}_{\text{ML}}^{(n'-1)} \\
        \vec{1}
    \end{array} \right], \quad \text{e.g.} \quad \mbf{G}_?^{6,2} = \begin{bmatrix}
        4 & 0 & 0 & 0 & 0 \\
        0 & 4 & 0 & 0 & 0 \\
        0 & 0 & 4 & 0 & 0 \\
        0 & 0 & 0 & 4 & 0 \\
        0 & 0 & 0 & 0 & 4 \\
        \hline
        1 & 1 & 1 & 1 & 1 \\
    \end{bmatrix},
\end{equation}

\noindent where we refer to rows of the $\mbf{G}_{\text{ML}}^{(n'-1)}$ block as \textit{guessing rows} and $\vec{1}$ is a row vector of ones which we refer to as the \textit{ambiguous row}.
Note that $\mbf{G}_?^{n',d}$ is a special case of the ambiguous guessing game $\mbf{G}^{(k)}$ (see main text), and without loss of generality, we express $\mbf{G}_?^{n',d}$ in a normal form where all elements $G_{y,x}$ are non-negative integers.
For any vertex $\mbf{V}\in\mc{V}_d^{n \to n'}$, the inequality

\begin{equation}\label{Eq:ambiguous-bell-inequality}
    \langle \mbf{G}_?^{n',d},  \mbf{V}\rangle \leq  d(n'-d)
\end{equation}

\noindent is satisfied.
We now prove the conditions for which inequality $(\mbf{G}_?^{n',d}, \gamma_?^{n',d})$ is a facet of $\mc{C}_d^{n\to n'}$. 

\begin{proposition} \label{prop-ambiguous-facet-tight}
    The inequality $(\mbf{G}_?^{n',d}, \gamma_?^{n',d})$ is a facet of $\mc{C}_d^{n \to n'}$ when  $n = n'-1$ and $n'-2 \geq d \geq 2$.
    
    \begin{proof}
        To prove that $(\mbf{G}_{?}^{n',d},\gamma_{?}^{n',d})$ is a facet of $\mc{C}_d^{n \to n'}$ we construct a set of $\dim(\mc{C}_d^{(n'-1) \to n'})=(n'-1)^2$ affinely independent vertices $\{\mbf{V}\in\mc{V}_d^{(n'-1) \to n'}|\; d(n'-d) = \langle\mbf{G}_?^{n',d}, \mbf{V}\rangle\}$.
        Using Proposition \ref{Appendix-ml-facets} we can easily enumerate $(n'-1)(n'-2)$ affinely independent vertices that optimize the $\mbf{G}_{\text{ML}}^{(n'-1)}$ block.
        The remaining vertices are constructed using the ambiguous row and $(d-1)$ guessing rows.
        In these vertices, the ambiguous row has $(d-1)$ null elements and $(n'-d)$ unit elements, hence, Lemma \ref{lemma-aff-ind-enumeration} can be used to $(n'-1)$ affinely independent arrangements of the ambiguous row.
        For each of the $(n'-1)$ arrangements, a vertex $\mbf{V}_?$ is constructed by setting each $V_?(x|x) = 1$ where $x\in[n'-1]$ and $V_?(n'|x) = 0$.
        Combining the $(n'-1)(n'-2)$ vertices from the $\mbf{G}_{\text{ML}}^{(n'-1)}$ block and the $(n'-1)$ vertices from the ambiguous row, a total of $(n'-1)^2$ affinely independent vertices are found.
        Therefore, $(\mbf{G}_?^{n',d}, \gamma_?^{n',d})$ is a tight Bell inequality of $\mc{C}_d^{(n'-1) \to n'}$.
        The upper bound $n'-2\geq d$ follows from the fact that if $d \geq (n'-1)$, then no optimal vertices would ever use the ambiguous row resulting in an insufficient number of vertices to justify the facet.
    \end{proof}
\end{proposition}

\subsubsection{Rescalings of Ambiguous Guessing facets}

An ambiguous guessing facet $(\mbf{G}_?^{n',d}, \gamma_?^{n',d})$ as defined in Proposition \ref{prop-ambiguous-facet-tight} can be rescaled to $\mbf{G}_?^{\prime 5,2}\in\mbb{R}^{n'\times (n+1)}$ by taking a guessing row $y$ where $G_{y,x}=(n'-d)$ distributing the value between between two columns such that $G'_{y,x}=1$ and $G'_{y,x'} = (n'-d)-1$ where $x'=n'$ is a new column.
This rescaling is a non-trivial input lifting rule.
The bound of the input-lifted facet is the same as the unlifted version.
For example, when $n' = 5$ and $d=2$, the $\mbf{G}_?^{5,2}$ is rescaled along the 4th row as,
    
    \begin{equation}
        \mbf{G}_?^{5,2} = \begin{bmatrix}
            3 & 0 & 0 & 0 \\
            0 & 3 & 0 & 0 \\
            0 & 0 & 3 & 0 \\
            0 & 0 & 0 & 3 \\
            1 & 1 & 1 & 1 \\
        \end{bmatrix} \rightarrow \mbf{G}_?^{\prime 5,2} = \begin{bmatrix}
            3 & 0 & 0 & 0 & 0 \\
            0 & 3 & 0 & 0 & 0 \\
            0 & 0 & 3 & 0 & 0 \\
            0 & 0 & 0 & 1 & 2 \\
            1 & 1 & 1 & 1 & 0 \\
        \end{bmatrix}.
    \end{equation}
    
\noindent This rescaling input lifting is a general trend observed in our computed signaling polytope facets \cite{SignalingDimension.jl}, however, it is not clear how broadly this lifting rule applies or generalizes.

\subsection{Anti-Guessing Facets} \label{appendix-anti-guessing-facets}

Another special case of the $k$-guessing game is the \textit{anti-guessing game} Bell inequality $(\mbf{G}_{\text{A}}^{n'},n')$ where $\mbf{G}_{\text{A}}^{n'} = \mbf{G}_{\text{K}}^{n',(n'-1)}$.
For any channel $\mbf{P}\in\mc{P}^{n \to n'}$ with $n=n'$ The anti-guessing Bell inequality $\langle \mbf{G}_{\text{A}}^{n'}, \mbf{P}\rangle \leq n'$ is satisfied, therefore anti-guessing games are not very useful for witnessing signaling dimension.
That said, the anti-guessing game is significant because it can be combined with a maximum likelihood game in block form to construct a facet of the $d=(n'-2)$ signaling polytope $\mc{C}_{(n'-2)}^{n \to n'}$.
We denote these anti-guessing facets by $\mbf{G}_{\text{A}}^{\varepsilon,m'}$ where the facet is constructed as

\begin{equation}
    \mbf{G}_{\text{A}}^{\varepsilon,m'} = \left[\begin{array}{c|c}
            \mbf{G}_{\text{A}}^{\varepsilon} & \hat{0} \\
            \hline
            \hat{0} & \mbf{G}_{\text{ML}}^{m'}
    \end{array}\right], \quad e.g. \quad     \mbf{G}_{\text{A}}^{4,2} = \left[\begin{array}{cccc|cc}
        1 & 1 & 1 & 0 & 0 & 0 \\
        1 & 1 & 0 & 1 & 0 & 0 \\
        1 & 0 & 1 & 1 & 0 & 0 \\
        0 & 1 & 1 & 1 & 0 & 0 \\
        \hline
        0 & 0 & 0 & 0 & 1 & 0 \\
        0 & 0 & 0 & 0 & 0 & 1 \\
    \end{array}\right]
\end{equation}

\noindent where $\mbf{G}_{\text{A}}^{\varepsilon,m'}\in\mbb{R}^{n'\times n'}$, $n' = \varepsilon + m'$, and $\hat{0}$ is a matrix block of zeros.
For channel $\mbf{P}\in\mc{C}_d^{n' \to n'}$,

\begin{equation} \label{Eq:anti-guessing-inequality}
    \langle \mbf{G}_{\text{A}}^{\varepsilon,m'},\mbf{P}\rangle \leq \varepsilon + d - 2 = \gamma_{\text{A}}^{\varepsilon,d}.
\end{equation}

\noindent This upper bound follows from the fact that no more than two rows are required to score $\varepsilon$ in the $\mbf{G}_{\text{A}}^{\varepsilon}$ block and the remaining $d-2$ rows  score one point each against the $\mbf{G}_{\text{ML}}^{m'}$ block.

\begin{proposition} \label{prop-antidistinguishability-famliy}
    The inequality $(\mbf{G}_{\text{A}}^{\varepsilon,m'}, \gamma_{\text{A}}^{\varepsilon,d})$ is a facet of $\mc{C}_d^{n \to n'}$ where $n=n'$, $n'-2 \geq d \geq 2$, and $n'-d+1 \geq \varepsilon \geq 3$.
    
    \begin{proof}
        To prove the tightness of the anti-guessing Bell inequality we show a row-by-row construction of $\dim(\mc{C}_d^{n\to n'})=n'(n'-1)$ affinely independent vertices $\{\mbf{V}\in\mc{V}_d^{n \to n'}|\langle\mbf{G}_{\text{A}}^{\varepsilon,m'}, \mbf{V}\rangle = \gamma_{\text{A}}^{\varepsilon,d}\}$.
        For convenience, we refer to the first $\varepsilon$ rows of $\mbf{G}_{\text{A}}^{\varepsilon,m'}$ as \textit{anti-guessing} rows and the remaining $m'$ rows as \textit{guessing} rows.
        We treat anti-guessing  and guessing rows individually because each admits its own vertex construction.
        To help illustrate this proof, we draw upon the example where $\varepsilon=m'=d=3$,
        
        \begin{equation}
            \mbf{G}_{\text{A}}^{3,3} = \left[\begin{array}{ccc|ccc}
                    1 & 1 & 0 & 0 & 0 & 0 \\
                    1 & 0 & 1 & 0 & 0 & 0 \\
                    0 & 1 & 1 & 0 & 0 & 0 \\
                    \hline
                    0 & 0 & 0 & 1 & 0 & 0 \\
                    0 & 0 & 0 & 0 & 1 & 0 \\
                    0 & 0 & 0 & 0 & 0 & 1 \\
            \end{array}\right] \text{and} \quad \gamma^{3,3} = 4.
        \end{equation}
        
        For a target anti-guessing row $y\in[1,\varepsilon]$ we construct $(n'-1)$ vertices where $(\varepsilon-1)$ vertices are constructed using the $\mbf{G}_A^{\varepsilon}$ block and $m'$ vertices are constructed using the $\hat{0}$ block in the top right.
        Note that a vertex achieves the upper bound $\gamma_{\text{A}}^{\varepsilon,d}$ only if two or less anti-guessing rows are used.
        A vertex $\mbf{V}_{\text{A}}$ is constructed using the $\mbf{G}_{\text{A}}^{\varepsilon,m'}$ block by setting $V_{\text{A}}(y|x) = 1$ for all $x$ that satisfy $G^{\varepsilon,m'}_{y,x}=1$ and selecting a secondary row $y'\neq y$ with $y'\in[1,\varepsilon]$ and setting $V_{\text{A}}(y'|x') = 1$ where $x'$ is the index of the null element in the target row $G^{\varepsilon,m'}_{y,x'} = 0$.
        All remaining elements of $\mbf{V}_{\text{A}}$ are set so that the first $(d-2)$ diagonal elements of the $\mbf{G}_{\text{ML}}^{m'}$ block are selected and any remaining terms are set as unit elements in the target row.
        An affinely independent vertex is constructed for each of the $(\varepsilon-1)$ choices of secondary row $y'$.
        For example, when targeting row $y=1$  we enumerate two vertices
        
        \begin{equation}
            \mbf{V}_{\text{A}}\in\left\{ \left[\begin{array}{ccc|ccc}
                1 & 1 & 0 & 0 & 1 & 1 \\
                0 & 0 & 1 & 0 & 0 & 0 \\
                0 & 0 & 0 & 0 & 0 & 0 \\
                \hline
                0 & 0 & 0 & 1 & 0 & 0 \\
                0 & 0 & 0 & 0 & 0 & 0 \\
                0 & 0 & 0 & 0 & 0 & 0 \\
            \end{array}\right],  \left[\begin{array}{ccc|ccc}
                1 & 1 & 0 & 0 & 1 & 1 \\
                0 & 0 & 0 & 0 & 0 & 0 \\
                0 & 0 & 1 & 0 & 0 & 0 \\
                \hline
                0 & 0 & 0 & 1 & 0 & 0 \\
                0 & 0 & 0 & 0 & 0 & 0 \\
                0 & 0 & 0 & 0 & 0 & 0 \\
            \end{array}\right]\right\}.
        \end{equation}
        
        For a target anti-guessing row $y$, an additional $m'$ vertices with form $\mbf{V}_{\text{A},0}$ are constructed using the $\hat{0}$ block in the top right.
        If $m'>(d-1)$, we set the target row as $V_{\text{A},0}(y|x) =1$ where $x\in[1,\varepsilon]$.
        The remaining $(d-1)$ rows are then used to maximize the $\mbf{G}_{\text{ML}}^{m'}$ block.
        Using Lemma \ref{lemma-aff-ind-enumeration} a set of $m'$ affinely independent vectors $\{\vec{b}_k\}_{k=1}^{m'}\}$ with $(d-1)$ null elements and $(m'-d+1)$ unit elements can be constructed and used in the $\hat{0}$ block of $\mbf{V}_{\text{A},0}$ by setting $V_{\text{A},0}(y|[\varepsilon+1,n']) = \vec{b}_k$.
        All remaining null elements in the target row of $\mbf{V}_{\text{A},0}$ are then set along the diagonal of the $\mbf{G}_{\text{ML}}^{m'}$ block.
        Since there are $m'$ choices of $\vec{b}_k$, that many affinely independent vertices can be constructed.
        For example, when targeting row $y=1$ we enumerate 3 vertices,
        
        \begin{equation}
            \mbf{V}_{\text{A},0}\in\left\{ \left[\begin{array}{ccc|ccc}
                1 & 1 & 1 & 0 & 0 & 1 \\
                0 & 0 & 0 & 0 & 0 & 0 \\
                0 & 0 & 0 & 0 & 0 & 0 \\
                \hline
                0 & 0 & 0 & 1 & 0 & 0 \\
                0 & 0 & 0 & 0 & 1 & 0 \\
                0 & 0 & 0 & 0 & 0 & 0 \\
            \end{array}\right],  \left[\begin{array}{ccc|ccc}
                1 & 1 & 1 & 0 & 1 & 0 \\
                0 & 0 & 0 & 0 & 0 & 0 \\
                0 & 0 & 0 & 0 & 0 & 0 \\
                \hline
                0 & 0 & 0 & 1 & 0 & 0 \\
                0 & 0 & 0 & 0 & 0 & 0 \\
                0 & 0 & 0 & 0 & 0 & 1 \\
            \end{array}\right], \left[\begin{array}{ccc|ccc}
                1 & 1 & 1 & 1 & 0 & 0 \\
                0 & 0 & 0 & 0 & 0 & 0 \\
                0 & 0 & 0 & 0 & 0 & 0 \\
                \hline
                0 & 0 & 0 & 0 & 0 & 0 \\
                0 & 0 & 0 & 0 & 1 & 0 \\
                0 & 0 & 0 & 0 & 0 & 1 \\
            \end{array}\right]\right\}.
        \end{equation}
        
        \noindent If $m'=(d-1)$, a secondary anti-guessing row $y'$ is selected where the anti-guessing rows are set as $V_{\text{A},0}(y|x) =1$ and $V_{\text{A},0}(y'|x')$ where $x,x'\in[1,\varepsilon]$ and $G_{y,x}^{\varepsilon,m'} =1$ and $G_{y,x'}^{\varepsilon,m'} =0$.
        The remainder of the procedure is the same as the $m'>(d-1)$ case.
        Note that in the $m'=(d-1)$ case one of the $\mbf{V}_{\text{A},0}$ vertices is redundant of a $\mbf{V}_{\text{A}}$ vertex.
        To reconcile this conflict another vertex must be added which maximizes $\mbf{G}_{\text{ML}}^{m'}$ with $V(y|x)=1$ for all $x\in[1,\varepsilon]$ and $V(x'|x') = 1$ for all $x'\in[\varepsilon+1,n']$. 
        By this procedure $(\varepsilon-1)+m'=(n'-1)$ affinely independent vertices are constructed for each target row $y\in[1,\varepsilon]$.
        Thus, $\varepsilon(n'-1)$ affinely independent vertices are constructed for the anti-guessing rows of $\mbf{G}_{\text{A}}^{\varepsilon,m'}$.
        
        For a target guessing row $y\in[\varepsilon+1,n']$ we construct $(n'-1)$ vertices where $\varepsilon$ are constructed using the $\hat{0}$ block in the lower left and $(m'-1)$ vertices using the $\mbf{G}_{\text{ML}}^{m'}$ block.
        Starting with the lower left $\hat{0}$ block we construct a vertex $\mbf{V}_{\text{ML},0}$ for each $x\in[1,\varepsilon]$ by setting $V_{\text{ML},0}(y|x) =1$ and $V_{\text{ML},0}(y|y)=1$.
        Of the remaining $(d-1)$ rows one is used to maximize the $\mbf{G}_{\text{A}}^{\varepsilon}$ block and $(d-2)$ rows maximize the $\mbf{G}_{\text{ML}}^{m'}$ block.
        Any unspecified unit terms of $\mbf{V}_{\text{ML},0}$ are set in the target row $y$.
        Since there are $\varepsilon$ values of $x$ to consider, this procedure produces $\varepsilon$ affinely independent vertices. For example, when targeting row $y=4$ we enumerate 3 vertices,
        
        \begin{equation}
            \mbf{V}_{\text{ML},0}\in\left\{ \left[\begin{array}{ccc|ccc}
                0 & 0 & 0 & 0 & 0 & 0 \\
                0 & 0 & 0 & 0 & 0 & 0 \\
                0 & 1 & 1 & 0 & 0 & 0 \\
                \hline
                1 & 0 & 0 & 1 & 0 & 1 \\
                0 & 0 & 0 & 0 & 1 & 0 \\
                0 & 0 & 0 & 0 & 0 & 0 \\
            \end{array}\right],  \left[\begin{array}{ccc|ccc}
                0 & 0 & 0 & 0 & 0 & 0 \\
                1 & 0 & 1 & 0 & 0 & 0 \\
                0 & 0 & 0 & 0 & 0 & 0 \\
                \hline
                0 & 1 & 0 & 1 & 0 & 1 \\
                0 & 0 & 0 & 0 & 1 & 0 \\
                0 & 0 & 0 & 0 & 0 & 0 \\
            \end{array}\right], \left[\begin{array}{ccc|ccc}
                1 & 1 & 0 & 0 & 0 & 0 \\
                0 & 0 & 0 & 0 & 0 & 0 \\
                0 & 0 & 0 & 0 & 0 & 0 \\
                \hline
                0 & 0 & 1 & 1 & 0 & 1 \\
                0 & 0 & 0 & 0 & 1 & 0 \\
                0 & 0 & 0 & 0 & 0 & 0 \\
            \end{array}\right]\right\}.
        \end{equation}
        
        Next, we use the $\mbf{G}_{\text{ML}}^{m'}$ block to a vertex $\mbf{V}_{\text{ML}}$.
        If $m'>(d-1)$, then we set $V_{\text{ML}}(1|x) = 1$ for all $x\in[1,\varepsilon]$ and use the procedure in Proposition \ref{prop-ml-facet} to enumerate $(m'-1)$ affinely independent vertices that optimize the $\mbf{G}_{\text{ML}}^{m'}$ block in the target row.
        If $m'=(d-1)$, then two anti-guessing rows are selected to maximize the $\mbf{G}_{\text{A}}^{\varepsilon}$ block while the procedure in Proposition \ref{prop-ml-facet} is used for the remaining $(d-2)$ rows are used to construct $(m'-1)$ affinely independent vertices that optimize the $\mbf{G}_{\text{ML}}^{m'}$ block in the target row.
        For example, when targetinng row $y=4$ we enumerate 2 vertices,
        
        \begin{equation}
            \mbf{V}_{\text{ML}}\in\left\{ \left[\begin{array}{ccc|ccc}
                1 & 1 & 1 & 0 & 0 & 0 \\
                0 & 0 & 0 & 0 & 0 & 0 \\
                0 & 0 & 0 & 0 & 0 & 0 \\
                \hline
                0 & 0 & 0 & 1 & 0 & 1 \\
                0 & 0 & 0 & 0 & 1 & 0 \\
                0 & 0 & 0 & 0 & 0 & 0 \\
            \end{array}\right],  \left[\begin{array}{ccc|ccc}
                1 & 1 & 1 & 0 & 0 & 0 \\
                0 & 0 & 0 & 0 & 0 & 0 \\
                0 & 0 & 0 & 0 & 0 & 0 \\
                \hline
                0 & 0 & 0 & 1 & 1 & 0 \\
                0 & 0 & 0 & 0 & 0 & 0 \\
                0 & 0 & 0 & 0 & 0 & 1 \\
            \end{array}\right]\right\}.
        \end{equation}
        
        \noindent Each guessing row produces $\varepsilon + (m'-1) = (n'-1)$ affinely independent vertices, thus in total, we have $m'(n'-1)$ vertices enumerated for the guessing rows. 
        
        Combining the procedures for the guessing and anti-guessing rows, we construct a total of $\varepsilon(n'-1) + m'(n'-1) = n'(n'-1)$ affinely independent vertices.
        Therefore, we prove that  $(\mbf{G}_{\text{A}}^{\varepsilon,m'}, \gamma_{\text{A}}^{\varepsilon,d})$ is a tight Bell inequality.
        We now address the bounds on $d$ and $\varepsilon$.
        The lower bound $\varepsilon\geq 3$ follows from the fact that $\mbf{G}_{\text{A}}^{2} = \mbf{G}_{\text{ML}}^2$ meaning the anti-guessing game is indistinguishable from the maximum likelihood game.
        The upper bound $n'-d+1 \geq \varepsilon$ follows from the fact that $m'\geq (d-1)$ must be satisfied or $n'(n'-1)$ affinely independent vertices cannot be found because the entire diagonal of the $\mbf{G}_{\text{ML}}^{m'}$ block must be used by every vertex to satisfy $\langle\mbf{G}_{\text{A}}^{\varepsilon,m'},\mbf{V}\rangle = \varepsilon + d -2$.
        The upper bound $n'-2\geq d$ results from the lower bound on $\varepsilon$ and the fact that $d$ cannot be so large the $n'-d+1 < 3$.
    \end{proof}
\end{proposition}

\section{Proof of Proposition \ref{prop-ambiguous-guessing-facets}} \label{proof-of-prop-ambiguous-guessing-facets}

In this section we prove the conditions for which the ambiguous guessing game $(\mbf{G}^{n,n'}_{k,d},d)$ is a facet of $\mc{C}_d^{n \to n'}$.

\subsection{Proof of Proposition \ref{prop-ambiguous-guessing-facets}(i)}

\begin{proof}
    To prove Proposition \ref{prop-ambiguous-guessing-facets}(i), we consider the general form of an ambiguous guessing Bell inequality $(\mbf{G}^{n,n'}_{k,d},d)$ where $\mbf{G}^{n,n'}_{k,d}\in\mbb{R}^{n'\times n}$ is row stochastic and contains $k=n'$ \textit{guessing} rows (see main text).
    Note that matrix $\mbf{G}^{n,n'}_{k,d}$ is row stochastic and therefore describes any input/output lifting and permutation of the maximum likelihood game $\mbf{G}_{\text{ML}}^{m'} = \mbb{I}_{m'}$ where $\min\{n,n'\} \geq m' \geq 1$. For example,
    
    \begin{equation}
        \begin{bmatrix}
            1 & 0 & 0 & 0 & 0 \\
            1 & 0 & 0 & 0 & 0 \\
            1 & 0 & 0 & 0 & 0 \\
            1 & 0 & 0 & 0 & 0 \\
            1 & 0 & 0 & 0 & 0 \\
        \end{bmatrix}, \quad 
        \begin{bmatrix}
            1 & 0 & 0 & 0 & 0 \\
            1 & 0 & 0 & 0 & 0 \\
            0 & 1 & 0 & 0 & 0 \\
            0 & 1 & 0 & 0 & 0 \\
            0 & 0 & 1 & 0 & 0 \\
        \end{bmatrix}, \quad \text{and}\quad \begin{bmatrix}
            1 & 0 & 0 & 0 & 0 \\
            0 & 1 & 0 & 0 & 0 \\
            0 & 0 & 1 & 0 & 0 \\
            0 & 0 & 0 & 1 & 0 \\
            0 & 0 & 0 & 0 & 1 \\
        \end{bmatrix},
    \end{equation}
    
    \noindent are all instances of $\mbf{G}^{5,5}_{5,d}$.
    By Proposition \ref{prop-ml-facet}, $(\mbf{G}_{\text{ML}}^{m'},d)$ is a facet of $\mc{C}_d^{m'\to m'}$ \textit{iff} $m' > d > 1$, that is, $\rank(\mbf{G}_{\text{ML}}^{m'}) > d$.
    When the trivial lifting rules (see Fig. \ref{Fig:input-output-lifting}) are applied to $\mbf{G}_{\text{ML}}^{m'}$, the rank of the lifted matrix does not change.
    Therefore, any Bell inequality  $(\mbf{G}^{n,n'}_{n',d},d)$ with $\rank(\mbf{G}^{n,n'}_{n',d}) > d$ is a facet of $\mc{C}_d^{n\to n'}$ that has been lifted from $\mc{C}_d^{m' \to m'}$ where $\min\{n,n'\} \geq m' > d$.
    Conversely, if $\rank(\mbf{G}^{n,n'}_{n',d}) < d$, then $\langle \mbf{G}^{n,n'}_{n',d}, \mbf{V}\rangle < d$ for any $\mbf{V}\in\mc{V}_d^{n \to n'}$.
    Likewise, if $\rank(\mbf{G}^{n,n'}_{n',d}) = d$, then there are an insufficient number of affinely independent vertices $\mbf{V}\in\mc{V}_d^{n \to n'}$ which satisfy $\langle \mbf{G}^{n,n'}_{n',d}, \mbf{V}\rangle = d$ because $d$ columns must have fixed values in $\mbf{V}$.
    Thus we conclude that when $\min\{n,n'\}>d>1$ $(\mbf{G}^{n,n'}_{n',d},d)$ is a tight Bell inequality of $\mc{C}_d^{n \to n'}$ \textit{iff} $\rank(\mbf{G}^{n,n'}_{n',d}) > d$.
\end{proof}

\begin{remark}
    Proposition \ref{prop-ambiguous-guessing-facets}(i) is significant because it allows one to easily find a facet of any signaling polytope $\mc{C}_d^{n \to n'}$.
    This enables the use of adjacency decomposition for any signaling polytope (see Appendix \ref{Appendix-Adjacency}).
\end{remark}

\subsection{Proof of Proposition \ref{prop-ambiguous-guessing-facets}(ii)}

\begin{proof}
    To prove Proposition \ref{prop-ambiguous-guessing-facets}(ii), we consider the ambiguous guessing game Bell inequalities $(\mbf{G}^{n,n'}_{k,d},d)$ with $k$ guessing rows and $(n'-k)$ ambiguous rows (see main text).
    Note that the ambiguous rows of $\mbf{G}^{n,n'}_{k,d}$ span the entire width of the matrix. For example,
    
    \begin{equation}
        \begin{bmatrix}
            1 & 0 & 0 & 0 \\
            1 & 0 & 0 & 0 \\
            1 & 0 & 0 & 0 \\
            1 & 0 & 0 & 0 \\
            \frac{1}{3} & \frac{1}{3} & \frac{1}{3} & \frac{1}{3}\\
            \frac{1}{3} & \frac{1}{3} & \frac{1}{3} & \frac{1}{3}\\
        \end{bmatrix}, \quad 
        \begin{bmatrix}
            1 & 0 & 0 & 0 \\
            1 & 0 & 0 & 0 \\
            0 & 1 & 0 & 0 \\
            0 & 1 & 0 & 0 \\
            \frac{1}{3} & \frac{1}{3} & \frac{1}{3} & \frac{1}{3}\\
            \frac{1}{3} & \frac{1}{3} & \frac{1}{3} & \frac{1}{3}\\
        \end{bmatrix}, \quad \text{and}\quad \begin{bmatrix}
            1 & 0 & 0 & 0 \\
            0 & 1 & 0 & 0 \\
            0 & 0 & 1 & 0 \\
            0 & 0 & 0 & 1 \\
            \frac{1}{3} & \frac{1}{3} & \frac{1}{3} & \frac{1}{3}\\
            \frac{1}{3} & \frac{1}{3} & \frac{1}{3} & \frac{1}{3}\\
        \end{bmatrix},
    \end{equation}
    
    \noindent are all instances of $\mbf{G}^{4,6}_{4,2}$.
    Furthermore, any ambiguous guessing facet described by Proposition \ref{prop-ambiguous-facet-tight} $(\mbf{G}_?^{n',d}, d(n'-d))$ of $\mc{C}_d^{(n'-1) \to n'}$ can be converted into an inequality $(\mbf{G}^{(n'-1),n'}_{(n'-1),d},d)$ simply by dividing the inequality by $(n'-d)$, hence, these two matrices describe the same inequality.
    It follows that any ambiguous guessing facet $\mbf{G}_?^{m',d}$ can be input lifted from $\mc{C}_d^{(m'-1) \to m'}$ to $\mc{C}_d^{(m'-1)\to n'}$ where $n' \geq m'$.
    Since the rank of the $(m'-1)$ guessing rows of $\mbf{G}_?^{m',d}$ is $(m'-1)$  and input liftings do not affect the matrix rank, any $\mbf{G}^{(m'-1),n'}_{k,d}$ with $n' > k\geq (m'-1)$ with a similar rank for its guessing rows must be a facet of $\mc{C}_d^{(m'-1) \to n'}$.
    Finally, if the rank of the guessing rows of $\mbf{G}^{n,n'}_{k,d})$ is less than $n$, then $\mbf{G}^{n,n'}_{k,d}$ cannot be a facet of $\mc{C}_d^{n \to n'}$ because there is an insufficient number of affinely independent vertices in $\{\mbf{V}\in \mc{V}_d^{n \to n'}|\; \langle \mbf{G}^{n,n'}_{k,d}, \mbf{V}\rangle = d\}$.
    This is true because Proposition \ref{prop-ml-facet} implies that we can enumerate $(k-1)n$ affinely independent vertices using only guessing rows of $\mbf{G}^{n,n'}_{k,d}$. This requires that the remaining $(n'-k)n$ affinely independent vertices are enumerated using $(d-1)$ guessing rows and one ambiguous row.
    However, as exemplified in the proof of Proposition \ref{prop-ambiguous-facet-tight}, this cannot be done unless there is a nonzero element in each column of the $k$ guessing rows of $\mbf{G}^{n,n'}_{k,d}$.
    Thus we conclude that for $n' > k \geq n$ and $n > d$ $\mbf{G}^{n,n'}_{k,d}$ is a facet of $\mc{C}_d^{n \to n'}$ \textit{iff} the rank of the guessing rows is $n$.
\end{proof}

\begin{remark}
    In our proof, we do not consider input liftings of $\mbf{G}_?^{m',d}$ because it results in matrices which deviate in form from $\mbf{G}^{n,n'}_{k,d}$.
    Input lifting append an all-zero column to $\mbf{G}_?^{m',d}$ while $\mbf{G}^{n,n'}_{k,d}$ is defined to have a nonzero element in each column of an ambiguous row ambiguous.
    Therefore, input liftings of ambiguous guessing facets $\mbf{G}_?^{n',d}$ are incompatible with the ambiguous guessing games $\mbf{G}^{n,n'}_{k,d}$ described in the main text.
\end{remark}

\section{Proof of Theorem \ref{Thm:main1}} \label{appendix-proof-thm-main1}

Our proofs to parts (i) and (ii) of Theorem \ref{Thm:main1} follow the same approach.  In both cases we want to show that the signaling polytope $\mc{C}^{n\to n'}_d$ is equivalent to some convex polytope defined by certain Bell inequalities.  We establish this by showing that the extreme points of the latter are also extreme points of the former; the converse has already been shown in Eq. \eqref{Eq:Ambiguous}.   Recall that the extreme points of $\mc{C}^{n\to n'}_d$ consist of all extreme points of $\mc{P}^{n\to n'}$ having rank no greater than $d$.  In other words, $\mbf{P}$ is extremal in $\mc{C}^{n\to n'}_d$ \textit{iff} it is column stochastic with $0/1$ elements and at most $d$ nonzero rows.

We rely heavily on the following general characterization of extreme points.

\begin{proposition}
\label{Prop:extremal}
Let $\mc{S}\subset\mbb{R}^{n'\times n}$ be some convex polytope in $\mbb{R}^{n'\times n}$.  Then $\mbf{P}$ is an extreme point of $\mc{S}$ \textit{iff} there does not exist some $\mbf{D}\in\mbb{R}^{n'\times n}$ such that $\mbf{P}\pm \mbf{D}\in \mc{S}$.
\end{proposition}
\noindent In our application of Proposition \ref{Prop:extremal}, we will refer to $\mbf{D}\in \mbb{R}^{n'\times n}$ as a ``valid'' perturbation of $\mbf{P}$ if $\mbf{P}\pm \mbf{D}\in\mc{S}$; hence if $\mbf{D}$ is a valid perturbation then $\mbf{P}$ cannot be extremal.

Some other terminology used in our proofs is the following.  For a channel $\mbf{P}\in\mc{P}^{n\to n'}$, an element $P(y|x)$ is called \textit{non-extremal} if it lies in the open interval $(0,1)$.  We say that $P(y|x)$ is a \textit{row maximizer} if it attains the largest value in row $y$ of $\mbf{P}$.  It is further called a \textit{unique row maximizer} if there are no other elements in row $y$ having this value.  Finally, we define the maximum likelihood estimation (ML) sum
\begin{align}
    \phi(\mbf{P}):=\sum_{x=1}^{n'}\Vert \mbf{r}_y\Vert_\infty,
\end{align}
where $\mbf{r}_y$ denotes row $y$ of $\mbf{P}$ and $\Vert \mbf{r}_y\Vert_\infty$ is its row maximizer.  Then the maximum likelihood estimation (ML) polytope can be expressed as
\begin{align}
    \mc{M}_d^{n\to n'}=\{\mbf{P}\in\mc{P}^{n\to n'}\;|\;\phi(\mbf{P})\leq d\}.
\end{align}
Note that $\phi$ is a convex function so that if $\phi(\mbf{P})=d$ with $\mbf{P}=\sum_\lambda p_\lambda \mbf{V}_\lambda$ for extreme points $\mbf{V}_\lambda\in\mc{M}_d^{n\to n'}$ and non-negative numbers $p_\lambda$, then necessarily $\phi(\mbf{V}_\lambda)=d$ for every $\lambda$.

\subsection{Proof of Theorem \ref{Thm:main1}(i)}

The proof of Theorem \ref{Thm:main1}(i) follows immediately from the following lemma due to the convexity of the ML and signaling polytopes.

\begin{lemma}
\label{lem:thm-1-i}
    For arbitrary $n$ and $n'$, the extreme points of $\mc{M}_{n'-1}^{n\to n'}$ are extreme points of $\mc{C}_{n'-1}^{n \to n'}$.
\end{lemma}
\begin{proof}

We first show the conclusion of Lemma \ref{lem:thm-1-i} is true for any extreme point $\mbf{V}$ of $\mc{M}_{n'-1}^{n\to n'}$ having ML sum $\phi(\mbf{V})<n'-1$.  If $\mbf{V}$ is not extremal in $\mc{P}^{n\to n'}$, then $\mbf{V}$ must have at least one column $x$ with two non-extremal elements $V(y_1|x)$ and $V(y_2|x)$.  However, we could then take two perturbations $V(y_1|x)\to V(y_1|x)\pm \epsilon$ and $V(y_2|x)\to V(y_2|x) \mp\epsilon$ with $\epsilon$ chosen sufficiently small so that the ML sum remains $<n'-1$ and the numbers remain non-negative.  Hence by contradiction, $\mbf{V}$ must be extremal in $\mc{P}^{n\to n'}$ with rank clearly $<n'-1$.
  
Let us then consider an extremal point $\mbf{V}$ of $\mc{M}_{n'-1}^{n\to n'}$ for which $\phi(\mbf{V})=n'-1$.  Since $\phi(\mbf{V})=n'-1$ is an integer and  $\mbf{V}$ has $n'$ rows, then $\mbf{V}$ must have at least two non-extremal row maximizers (possibly in different columns).  
We will again introduce perturbations, but care is needed to ensure that the perturbations are valid; \textit{i.e.} the perturbed channels must remain in $\mc{M}_{n'-1}^{n\to n'}$.  There are two cases to consider.

\medskip

\noindent\textbf{Case (a)}:
Suppose that two non-extremal row maximizers occur in the same column: say $V(y_1|x)$ and $V(y_2|x)$ are both row maximizers in column $x$.  Since these values will account for the contributions of rows $y_1$ and $y_2$ in the ML sum, and since there are only $n'$ total rows in this sum, we must have that all other row maximizers are $+1$.   Hence we introduce perturbations $V(y_1|x)\to V(y_1|x)\pm \epsilon$ and $V(y_2|x)\to V(y_2|x)\mp \epsilon$. If $V(y_1|x)$ and $V(y_2|x)$ are unique row maximizers, then this perturbation is valid.  On the other hand, if there are columns $x',x''$ such that $V(y_1|x)=V(y_1|x')$ and/or $V(y_2|x)=V(y_2|x'')$ (with possibly $x'=x''$), then we must also introduce a corresponding perturbation $V(y_1|x')\to V(y_1|x')\pm \epsilon$ and/or $V(y_2|x'')\to V(y_2|x'')\mp \epsilon$.  To preserve normalization in columns $x'$ and/or $x''$, we will have to introduce an off-setting perturbation to some other row in $x'$ and/or $x''$.  This can always be done since either $x'=x''$, or $x'$ and/or $x''$ have a non-extremal element in some other row which is not a row maximizer (since all other row maximizers are $+1$).

\medskip

\noindent\textbf{Case (b)}:  No column has two non-extremal row maximizers, and $\mbf{V}$ has at least two non-extremal row maximizers that belong to different columns.  For each row $y$ with a non-extremal row maximizer, add perturbations $\pm\epsilon_y$ to all the row maximizers in that row.  Since each column has at most one row maximizer, a normalization-preserving perturbation $\mp\epsilon_y$ can be added to another non-extremal element in any column having a row maximizer in row $y$.  Finally, choose the $\epsilon_y$ so that $\sum_{y=1}^{n'}\epsilon_y=0$.

\end{proof}

\subsection{Proof of Theorem \ref{Thm:main1}(ii)}

We now turn to the ambiguous polytopes $\mc{A}_\cap^{n\to n'}:=\cap_{k=n}^{n'}\mc{A}_{k,n-1}^{n\to n'}$.  Recall that $\mc{A}_{k,n-1}^{n\to n'}$ is the polytope of channels $\mbf{P}\in\mc{P}^{n\to n'}$ satisfying all Bell inequalities of the form
\begin{equation}
\label{Eq:appendix-ambiguous-ineq}
    \langle\mbf{G}^{n,n'}_{k,n-1},\mbf{P}\rangle\leq n-1,
\end{equation}

\noindent with $\mbf{G}^{n,n'}_{k,n-1}$ having $k$ guessing rows and $(n'-k)$ ambiguous rows.  In this case, all the elements in an ambiguous row are equal to $\frac{1}{n-d+1}=\frac{1}{2}$.  

To prove Theorem \ref{Thm:main1}(ii) we apply the following lemma to show that the extreme points of $\mc{A}_\cap^{n \to n'}$ are the same as those of $\mc{C}_{n-1}^{n \to n'}$.
Then by convexity of $\mc{A}_{\cap}^{n \to n'}$ and $\mc{C}_{n-1}^{n \to n'}$ we must have $\mc{A}_{\cap}^{n \to n'} = \mc{C}_{n-1}^{n \to n'}$.

\begin{lemma}
\label{lem:thm-1-ii}
      For arbitrary $n'\geq n$, the extreme points of $\mc{A}_{\cap}^{n\to n'}$ are extreme points of $\mc{C}_{n-1}^{n \to n'}$.  
  \end{lemma}
\begin{proof}

We first argue that the conclusion of Lemma \ref{lem:thm-1-ii} is true for any extreme point $\mbf{V}$ of $\mc{A}_\cap^{n\to n'}$ such that $\langle \mbf{G}_{k,n-1}^{n,n'},\mbf{V}\rangle< n-1$ for all $\mbf{G}_{k,n-1}^{n,n'}$ and all integers $k\in[n,n']$.  Analogous to Lemma \ref{lem:thm-1-i}, if $\mbf{V}$ has at least one column $x$ with two non-extremal elements $V(y_1|x)$ and $V(y_2|x)$, we can take two sufficiently small perturbations $V(y_1|x)\to V(y_1|x)\pm \epsilon$ and $V(y_2|x)\to V(y_2|x)\mp\epsilon$ and still satisfy all the constraints of Eq. \eqref{Eq:appendix-ambiguous-ineq}. Hence, $\mbf{V}$ must be an extreme element of $\mc{P}^{n\to n'}$.  In this case,  $\rank(\mbf{V})< n-1$ since $\phi(\mbf{P})<n-1$, and so $\mbf{V}\in\mc{C}_{n-1}^{n\to n'}$.

It remains to prove the conclusion of Lemma \ref{lem:thm-1-ii} whenever Eq. \eqref{Eq:appendix-ambiguous-ineq} is tight for some $\mc{A}_{k,n-1}^{n\to n'}$.  The lengthiest part of this argument is when $k=n'$ and tightness in Eq. \eqref{Eq:appendix-ambiguous-ineq} corresponds to the ML sum equaling $n-1$.  In this case, Proposition \ref{Prop:full-ML} below shows that $\mbf{V}$ must be an extreme point of $\mc{C}_{n-1}^{n\to n'}$.  However, before proving this result, we apply it to show that Lemma \ref{lem:thm-1-ii} holds whenever Eq. \eqref{Eq:appendix-ambiguous-ineq} is tight for some other $\mbf{G}^{n,n'}_{k,n-1}$ with $k<n'$.  Specifically, we will perform a lifting technique on any vertex $\mbf{V}$ satisfying  $\langle\mbf{G}^{n,n'}_{k,n-1},\mbf{V}\rangle= n-1$ and reduce it to the case of the ML sum equaling $(n-1)$.

Suppose that $\phi(\mbf{V})<n-1$ yet there exists some $\mbf{G}^{n,n'}_{k,n-1}$ such that $\langle \mbf{G}^{n,n'}_{k,n-1},\mbf{V}\rangle=n-1$.  The matrix $\mbf{G}^{n,n'}_{k,n-1}$ identifies $(n'-k)$ ambiguous rows, and suppose that $y$ is an ambiguous row such that $\frac{1}{2}\Vert\mbf{r}_y\Vert_1>\Vert\mbf{r}_y\Vert_\infty$, with $\mbf{r}_y$ being the $y^{th}$ row of $\mbf{V}$.  To be concrete, let us suppose without loss of generality that the components of row $\mbf{r}_y$ are arranged in non-increasing order (\textit{i.e.} $V(y|x_{i})\geq V(y|x_{i+1})$), and let $k$ be the smallest index such that 
\begin{equation}
\label{Eq:lemma-ambiguous-step}
\frac{1}{2}\left(-\sum_{i=1}^{k-1}V(y|x_i)+\sum_{i=k}^{n}V(y|x_i)\right)\leq V(y|x_k).
\end{equation}
By the assumption $\frac{1}{2}\Vert\mbf{r}_y\Vert_1>\Vert\mbf{r}_y\Vert_\infty$, we have $k>1$.  Also, since $k$ is the smallest integer satisfying Eq. \eqref{Eq:lemma-ambiguous-step}, we have
\begin{equation}
    \frac{1}{2}\left(-\sum_{i=1}^{k-2}V(y|x_i)+\sum_{i=k-1}^{n}V(y|x_i)\right)> V(y|x_{k-1}).
\end{equation}
Subtracting $V(y|x_{k-1})$ from both sides of this equation implies that the LHS of Eq. \eqref{Eq:lemma-ambiguous-step} is strictly positive.  Hence, there exists some $\lambda\in (0,1]$ such that
\begin{equation}
    \lambda V(y|x_k)=\frac{1}{2}\left(-\sum_{i=1}^{k-1}V(y|x_i)+\sum_{i=k}^{n}V(y|x_i)\right).
\end{equation}
Consider then the new matrix $\wt{\mbf{V}}$ formed from $\mbf{V}$ by splitting row $y$ into $k$ rows as follows:
\begin{align}
\label{Eq:lemma-ambigious-matrix}
    \mbf{r}_y \to\begin{bmatrix}V(y|x_1)&0&\cdots &0 &(1-\lambda)V(y|x_k)&V(y|x_{k+1})&\cdots&V(y|x_n)\\
     0&V(y|x_2)&\cdots &0 &0&0&\cdots&0\\
     \vdots&&\vdots&&\vdots&&\vdots\\
     0&0&\cdots &V(y|x_{k-1}) &0&0&\cdots&0\\
     0&0&\cdots &0 &\lambda V(y|x_{k})&0&\cdots&0
     \end{bmatrix}.
    \end{align}
Notice that we can obtain $\mbf{V}$ from $\wt{\mbf{V}}$ by coarse-graining over these rows.  Moreover, this decomposition was constructed so that
\begin{equation}
   \sum_{i=1}^k\left\Vert\wt{\mbf{r}}_{y_i}\right\Vert_\infty=\sum_{i=1}^{k-1}V(y|x_i)+\lambda V(y|x_k)=\Vert\mbf{r}_y\Vert_1,
\end{equation}
where the $\wt{\mbf{r}}_{y_i}$ are the rows in Eq. \eqref{Eq:lemma-ambigious-matrix}.  Essentially this transformation allows us to replace an ambiguous row with a collection of guessing rows so that the overall guessing score does not change.

We perform this row splitting process on all ambiguous rows of $\mbf{V}$ thereby obtaining a new matrix $\wt{\mbf{V}}$ such that $\phi(\wt{\mbf{V}})=n-1$.  If $m$ is the total number of rows in $\wt{\mbf{V}}$, then $\wt{\mbf{V}}$ will be an element of $\mc{A}_\cap^{n\to m}$.  We decompose $\wt{\mbf{V}}$ into a convex combination of extremal points of $\mc{A}_\cap^{n\to m}$ as $\wt{\mbf{V}}=\sum_\lambda p_\lambda \wt{\mbf{V}}_\lambda$.  By the convexity of $\phi$, it follows that $\phi(\wt{\mbf{V}}_\lambda)=n-1$, and we can therefore apply Proposition \ref{Prop:full-ML} below on the channels $\wt{\mbf{V}}_\lambda$ to conclude that they are extreme points of $\mc{C}_{n-1}^{n \to m}$.  Consequently, each $\wt{\mbf{V}}_\lambda$ has only one nonzero element per row.  Let $\mbf{R}$ denote the coarse-graining map such that $\mbf{V}=\mbf{R}\wt{\mbf{V}}$, and apply
\begin{equation}
    \mbf{V}=\mbf{R}\wt{\mbf{V}}=\sum_\lambda p_\lambda\mbf{R}\wt{\mbf{V}}_\lambda.
\end{equation}
However, by the assumption that $\mbf{V}$ is extremal, this is only possible if $\mbf{R}\wt{\mbf{V}}_\lambda$ is the same for every $\lambda$.  As a result, any two $\wt{\mbf{V}}_\lambda$ and $\wt{\mbf{V}}_{\lambda'}$ can differ only in rows that coarse-grain into the same rows by $\mbf{R}$.  From this it follows that $\mbf{V}$ can have no more than one nonzero element per column and $\rank(\mbf{V}) \leq n-1$.  Hence we've shown that the extreme points of $\mc{A}_{\cap}^{n\to n'}$ are indeed extreme points  of the signaling polytope $\mc{C}_{n-1}^{n \to n'}$.

\medskip

To complete the proof of Lemma \ref{lem:thm-1-ii}, we establish the case when $\phi(\mbf{V})=n-1$, as referenced above.  We begin by proving the partial result provided by Proposition \ref{Prop:partial-ML} and then, use this result to prove Proposition \ref{Prop:full-ML}.

\begin{proposition}
\label{Prop:partial-ML}
If $\mbf{V}$ is an extreme point of $\mc{A}_\cap^{n\to n'}$ satisfying $\phi(\mbf{V})=n-1$, then each column of $\mbf{V}$ must have at least one unique row maximizer or it has only one nonzero element.  
\end{proposition}

\begin{proof}
Suppose on the contrary that some column $x$ has more than one nonzero element yet no unique row maximizer.  Let $\mc{S}_x\subset[n']$ be the set of rows for which column $x$ contains a row maximizer.  Since only one row maximizer per row contributes to the ML sum, and the elements of column $x$ sum to one, we can satisfy $\phi(\mbf{V})=n-1$ \textit{iff} both conditions hold:
\begin{enumerate}
    \item[(i)] each row $y$ in $\mc{S}_x$ has only two nonzero elements $V(y|x)$ and $V(y|x_y)$ for some column $x_y\not=x$;
    \item[(ii)] every other nonzero element in $\mbf{V}$ outside of column $x$ and the rows in $\mc{S}_x$ are unique row maximizers.
\end{enumerate}   

\noindent With this structure, we introduce three cases of valid perturbations. 

\medskip

\noindent\textbf{Case (a)}: $V(y_1|x)$ and $V(y_2|x)$ are non-extremal elements in column $x$ with $y_1,y_2\not\in \mc{S}_x$.  Then $V(y_1|x)\to V(y_1|x)\pm \epsilon$ and $V(y_2|x)\to V(y_2|x)\mp \epsilon$ is a valid perturbation.  Indeed, even if we consider $y_1$ or $y_2$ as ambiguous rows, there is at most one other element in each of these rows (property (i) above), and so this perturbation would not violate any of the inequalities in \eqref{Eq:appendix-ambiguous-ineq}.

\medskip

\noindent\textbf{Case (b)}: $V(y_1|x)$ and $V(y_2|x)$ are non-extremal elements in column $x$ with $y_1\in \mc{S}_x$ and $y_2\not\in\mc{S}_x$.  Then $V(y_1|x)=V(y_1|x_{y_1})$ for some other column $x_{y_1}\not=x$.  By normalization, there will be another element $V(y_3|x_{y_1})$ in column $x_{y_1}$ that by property (ii) is a unique row maximizer.  Hence, we introduce perturbations
\begin{align}
    V(y_1|x)&\to V(y_1|x)\pm \epsilon& V(y_1|x_{y_1})&\to V(y_1|x_{y_1})\pm\epsilon\notag\\
    V(y_2|x)&\to V(y_2|x)\mp\epsilon &&\notag\\
    &&V(y_3|x_{y_1})&\to V(y_3|x_{y_1})\mp\epsilon.
\end{align}
For clarity, the line spacing is chosen here so that elements on the same vertical line correspond to elements in the same row of $\mbf{V}$.  By properties (i) and (ii), these perturbations do not increase the ML sum, nor are they able to violate any of the other inequalities in \eqref{Eq:appendix-ambiguous-ineq}.

\medskip

\noindent\textbf{Case (c)}: $V(y_1|x)$ and $V(y_2|x)$ are non-extremal elements in column $x$ with $y_1,y_2\in \mc{S}_x$.  Then $V(y_1|x)=V(y_1|x_{y_1})$ and $V(y_2|x)=V(y_2|x_{y_2})$ for some other columns $x_{y_1},x_{y_2}\not=x$ (with possibly $x_{y_1}=x_{y_2})$.  By normalization, there will be elements $V(y_3|x_{y_1})$ and $V(y_4|x_{y_2})$ in columns $x_{y_1}$ and $x_{y_2}$ respectively that are unique row maximizers (again by property (ii)).  Note this requires that $y_1,y_2,y_3,y_4$ are all distinct rows.  Hence, we introduce perturbations
\begin{align}
    V(y_1|x)&\to V(y_1|x)\pm \epsilon& V(y_1|x_{y_1})&\to V(y_1|x_{y_1})\pm\epsilon &&\notag\\
    V(y_2|x)&\to V(y_2|x)\mp\epsilon&&& V(y_2|x_{y_2})&\to V(y_2|x_{y_2})\mp\epsilon\notag\\
    &&V(y_3|x_{y_1})&\to V(y_3|x_{y_1})\mp\epsilon\notag\\
    &&&&V(y_4|x_{y_2})&\to V(y_4|x_{y_2})\pm\epsilon,
\end{align}
Normalization is preserved under these perturbations and all the inequalities in \eqref{Eq:appendix-ambiguous-ineq} are satisfied.  

As we have shown valid perturbations in all three cases under the assumption that some column has non-extremal elements with no unique row maximizer, the proposition follows.

\end{proof}

\begin{proposition}
\label{Prop:full-ML}
If $\mbf{V}$ is an extreme point of $\mc{A}_\cap^{n\to n'}$ satisfying $\phi(\mbf{V})=n-1$, then $\mbf{V}$ is an extreme point of $\mc{C}_d^{n\to n'}$.
\end{proposition}

\begin{proof}
Suppose that $\mbf{V}$ has some column $x_1$ containing more than one nonzero element (if no such column can be found, then the proposition is proven).  Let $V(y_1|x_1)\in (0,1)$ denote a unique row maximizer, which is assured to exist by Proposition \ref{Prop:partial-ML}.  We again proceed by considering two cases.

\medskip

\noindent\textbf{Case (a)}:  Column $x_1$ contains only one row maximizer $V(y_1|x_1)$ and all other elements in the column are not row maximizers.  Then there must exist another column $x_1'$ that also contains at least two nonzero elements.  Indeed, if on the contrary all other columns only had one nonzero element each, then it would be impossible for $\phi(\mbf{V})=n-1$.  If $x_1'$ only contains row maximizers, then proceed to case (b) and replace $x_1$ with $x_1'$.  Otherwise, $x_1'$ does not only contain row maximizers; rather it has a unique row maximizer $V(y_3|x_1')$ in row $y_3$ and a nonzero element $V(y_4|x_1')$ in row $y_4$ that is not a row maximizer.  Thus, we can introduce the valid perturbations
\begin{align}
    V(y_1|x_1)&\to V(y_1|x_1)\pm \epsilon\notag\\
    V(y_2|x_1)&\to V(y_2|x_1)\mp\epsilon\notag\\
    &&V(y_3|x')&\to V(y_3|x_1')\mp\epsilon\notag\\
     &&V(y_4|x')&\to V(y_4|x_1')\mp\epsilon
\end{align}
where $V(y_2|x_1)$ denotes another nonzero element in $x_1$ (with possibly $y_2=y_3,y_4$ and/or $y_3=y_1$).  It can be verified that all inequalities in \eqref{Eq:appendix-ambiguous-ineq} are preserved under these perturbations.

\medskip

\noindent\textbf{Case (b)}:  Column $x_1$ only contains row maximizers, with  $V(y_2|x_1)$ being another one in addition to $V(y_1|x_1)$.  If $V(y_2|x_1)$ is a unique row maximizer, then valid perturbations can be made to both $V(y_1|x_1)$ and $V(y_2|x_1)$.  On the other hand, suppose  that $V(y_2|x_1)$ is a non-unique row maximizer, and let $V(y_2|x_2)=V(y_2|x_1)$ be another row maximizer in column $x_2$.  There can be no other nonzero elements in row $y_2$. Indeed,  if there were another column, say $x_3$, such that $V(y_2|x_3)>0$, then we would have 
\begin{equation}
\label{Eq:ML-to-ambiguous-1}
\frac{1}{2}\Vert\mbf{r}_{y_2}\Vert_1\geq \frac{1}{2}\Big(V(y_2|x_1)+V(y_2|x_2)+V(y_2|x_3)\Big)>V(y_2|x_2)=\Vert\mbf{r}_{y_2}\Vert_\infty,
\end{equation}
and so
\begin{equation}
\label{Eq:ML-to-ambiguous-2}
\langle\mbf{G}_{n'-1,n-1}^{n,n'},\mbf{V}\rangle>\phi(\mbf{V})=n-1,
\end{equation}
where the one ambiguous row in $\mbf{G}_{n'-1,n-1}^{n,n'}$ is $y_2$.  Hence, the only nonzero elements in row $y_2$ are $V(y_2|x_1)$ and $V(y_2|x_2)$.  Let $V(y_3|x_2)$ be a unique row maximizer in column $x_2$.  

We must be able to find another column $x_3$ with more than one nonzero element, one of which is a unique row maximizer and the other which is a non-unique row maximizer.  For if this were not the case, then any other column in $\mbf{V}$ would either have a unique row maximizer equaling one, or it would have at least two elements, one being a unique row maximizer and the others not being row maximizers.  However, the latter possibility was covered in case (a) and was shown to be impossible for an extremal $\mbf{V}$.  For the former, if all then other $n-2$ columns outside of $x_1$ and $x_2$ contain unique row maximizers equaling one, then they would collectively contribute an amount of $n-2$ to the ML sum.  Since every element in column $x_1$ is a row maximizer, and $V(y_3|x_2)$ is a row maximizer in column $x_2$, we would have $\phi(\mbf{V})>(n-2)+1+V(y_3|x_2)>n-1$.  Hence, there must exist another column $x_3$ with a non-unique row maximizer $V(y_5|x_3)$ that is shared with column $x_4$ (which may be equivalent to either $x_1$ or $x_2$).  Letting $V(y_4|x_3)$ and $V(y_6|x_4)$ denote unique row maximizers in columns $x_3$ and $x_4$, respectively, we can perform the valid perturbations
\begin{align}
    V(y_1|x_1)&\to V(y_1|x_1)\pm \epsilon\notag\\
    V(y_2|x_1)&\to V(y_2|x_1)\mp\epsilon&V(y_2|x_2)&\to  V(y_2|x_2)\mp\epsilon\notag\\
    &&V(y_3|x_2)&\to V(y_3|x_2)\pm\epsilon\notag\\
    &&&&V(y_4|x_3)&\to V(y_4|x_3)\pm \epsilon\notag\\
    &&&&V(y_5|x_3)&\to V(y_5|x_3)\mp\epsilon&V(y_5|x_4)&\to  V(y_2|x_2)\mp\epsilon\notag\\
    &&&&&&V(y_6|x_4)&\to V(y_6|x_4)\pm\epsilon.
\end{align}
Note that $y_1,y_2,y_3,y_4,y_5,y_6$ are all distinct rows since each row in $\mbf{V}$ can have at most one pair of non-unique row maximizers while rows $y_1,y_3,y_4,y_6$ contain unique row maximizers.  This assures that the perturbations do not violate the inequalities in \eqref{Eq:appendix-ambiguous-ineq}.

As cases (a) and (b) exhaust all possibilities, we see that $\mbf{V}$ can only have one nonzero element per column.  From this the conclusion of Proposition \ref{Prop:full-ML} follows.
\end{proof}

This completes the proof of Lemma \ref{lem:thm-1-ii}.
\end{proof}

\section{Proof of Theorem \ref{Thm:main2}} \label{appendix-proof-of-thm-2}

\label{Appendix-bit}

In this section we analyze the $\mc{C}_2^{n\to 4}$ signaling polytope to prove the Theorem \ref{Thm:main2}.  To begin we define the polyhedron of channels
\begin{equation}
    \mc{C}(\mbf{G},\gamma):=\{\mbf{P}\in\mc{P}^{n\to n'}\;|\;\langle\mbf{G},\mbf{P}\rangle=\sum_{x=1}^n\sum_{y=1}^{n'}G_{y,x}P(y|x)\leq \gamma\}
\end{equation}
for any Bell inequality $(\mbf{G},\gamma)$ with $\mbf{G}\in \mbb{R}^{n'\times n}$ and $\gamma\in\mbb{R}$.
Since $\mc{C}_d^{n\to n'}$ is a convex polytope, there exists a finite number of polyhedra $\{\mc{C}(\mbf{G}_m,\gamma_m)\}_{m=1}^r$ such that
\begin{align}
\label{Eq:polyheda-bound}
   \mc{C}_d^{n\to n'}=\bigcap_{m=1}^r\mc{C}(\mbf{G}_m,\gamma_m).
\end{align}
\begin{remark}
\label{Remark:non-negative}
Without loss of generality, we can assume that the matrices $\mbf{G}_m$ contain non-negative elements.  Indeed, if $G_{y,x}<0$ is the smallest element in column $x$ of $\mbf{G}_m$, then we replace each element in column $x$ as $G_{y',x} \to G_{y',x}+G_{y,x}$ and shift $\gamma\to \gamma+G_{y,x}$. Hence the smallest element in column $x$ of $\mbf{G}_m$ becomes $G_{y,x} = 0$.
\end{remark}

The proof of Theorem \ref{Thm:main2} is a consequence of Lemmas \ref{Lem:reduce} and \ref{lem:nonzero-column-limit} below and our numerical results for the $\mc{C}_2^{n \to 4}$ signaling polytope \cite{SignalingDimension.jl} (see Fig. \ref{Fig:6-2-4_facets}).
First, by Lemma \ref{Lem:reduce} we can reduce any Bell inequality $(\mbf{G},\gamma)$ bounding $\mc{C}_2^{n \to 4}$ to a new Bell inequality $(\hat{\mbf{G}},\hat{\gamma})$ having at most 2 nonzero elements in each column.
The reduced inequality $(\hat{\mbf{G}}, \hat{\gamma})$ satisfies $\mc{C}_2^{n \to 4}\subset\mc{C}(\hat{\mbf{G}},\hat{\gamma}) \subset (\mbf{G},\gamma)$ and thus bounds $\mc{C}_2^{n \to 4}$ more tightly than $(\mbf{G},\gamma)$.
Next, we use Lemma \ref{lem:nonzero-column-limit} to show that for any integer $n$ a tight Bell inequality of $\mc{C}_2^{n \to 4}$ has at most six nonzero columns.
The presence of all-zero columns implies that this inequality is simply an input lifting of a tight bell inequality of $\mc{C}_2^{6 \to 4}$.
Therefore, the complete set of tight Bell inequalities bounding $\mc{C}_2^{n \to 4}$ is the set of all input  liftings and permutations of the generator facets of $\mc{C}_2^{6 \to 4}$ shown in Fig. \ref{Fig:6-2-4_facets}.

\begin{lemma}
\label{Lem:reduce}
If $\mc{C}_d^{n\to n'}\subset\mc{C}(\mbf{G},\gamma)$, then there exists a polyhedron $\mc{C}(\hat{\mbf{G}},\hat{\gamma})$ with $\hat{\mbf{G}}$ having at most $(n'-d)$ nonzero elements in each column and satisfying 
 \begin{equation}
 \label{Eq:theorem-reduce}
     \mc{C}_d^{n\to n'}\subset \mc{C}(\hat{\mbf{G}},\hat{\gamma})\subset \mc{C}(\mbf{G},\gamma).
 \end{equation}
\end{lemma}

\begin{proof}
Suppose  $\mc{C}_d^{n\to n'}\subset\mc{C}(\mbf{G},\gamma)$ and consider an arbitrary $x\in[n]$.  For convenience, let us relabel the elements of the $x^{th}$ column of $\mbf{G}$ in non-increasing order; \textit{i.e} $G_{y,x}\geq G_{y+1,x}$.  Every vertex $\mbf{V}$ of $\mc{C}_{d}^{n\to n'}$ will satisfy

\begin{align} \label{Eq:Vertex-bound}
     \gamma\geq \sum_{x',y}G_{y,x'}V(y|x')&=\sum_y G_{y,x}V(y|x)+\sum_{x'\not=x,y}G_{y,x'}V(y|x')\notag\\
     &=\sum_y G_{y,x}V(y|x)+f(\mbf{G},\mbf{V},x),
\end{align}

\noindent where $f(\mbf{G},\mbf{V},x):=\sum_{x'\not=x,y}G_{y,x'}V(y|x')$.
A key observation is

\begin{align}
 \label{Eq:Vertex-bound-main}
     \gamma\geq G_{d,x} +f(\mbf{G},\mbf{V},x) \quad\text{for every vertex $\mbf{V}$ of $\mc{C}_d^{n\to n'}$}.
\end{align}

We prove this observation using Eq. \eqref{Eq:Vertex-bound}.  First consider any vertex $\mbf{V}$ such that $V(y|x)=\delta_{d'y}$ with $d'\geq d$.  Then Eq. \eqref{Eq:Vertex-bound} shows that $\gamma\geq G_{d',x}+f(\mbf{G},\mbf{V},x)\geq G_{d,x}+f(\mbf{G},\mbf{V},x)$, since we have labeled the elements in non-increasing order.  On the other hand, consider a vertex $\mbf{V}$ for which $V(y|x)=\delta_{d'y}$ with $d'<d$.  Since vertices can be formed with $d$ nonzero rows, we can choose another vertex $\mbf{V}'$ that is identical to $\mbf{V}$ in all columns $x'\not=x$, and yet for column $x$ it satisfies $V'(y|x)=\delta_{d''y}$ with $d\leq d''$.  Hence applying Eq. \eqref{Eq:Vertex-bound} to vertex $\mbf{V}'$ yields
 \begin{equation}
     \gamma\geq G_{d'',x}+f(\mbf{V}',x)\geq G_{d,x}+f(\mbf{G},\mbf{V}',x)=G_{d,x}+f(\mbf{G},\mbf{V},x),
 \end{equation}
 where the last line follows from the fact that $\mbf{V}$ and $\mbf{V}'$ only differ in column $x$.
 
 Having established Eq. \eqref{Eq:Vertex-bound-main}, we next form a new matrix $\hat{\mbf{G}}$ which is obtained from $\mbf{G}$ by replacing its $x^{th}$ column with
 \begin{equation}
    (\hat{G}_{y,x})^T_y:= (\overbrace{0,0\cdots,0}^d,G_{d+1,x}-G_{d,x},\cdots,G_{n',x}-G_{d,x})^T.
 \end{equation}
 Letting $\hat{\gamma}=\gamma-G_{d,x}$, for any vertex $\mbf{V}$ we have
 \begin{align}
     \sum_{x',y}\hat{G}_{y,x'}V(y|x')&=\sum_y \hat{G}_{y,x}V(y|x)+f(\mbf{G},\mbf{V},x)\notag\\
     &=\begin{cases}f(\mbf{G},\mbf{V},x)\quad\qquad\qquad\qquad\text{if $V(y|x)=\delta_{d'y}$ with $d'\leq d$}\\
     G_{d',x}-G_{d,x}+f(\mbf{G},\mbf{V},x)\quad\text{if $V(y|x)=\delta_{d'y}$ with $d'>d$}
     \end{cases}\notag\\
     &\leq \hat{\gamma},
 \end{align}
 where the last inequality follows from Eq. \eqref{Eq:Vertex-bound-main} (in the first case) and Eq. \eqref{Eq:Vertex-bound} (in the second case).  Hence, we have that $\mc{C}_d^{n\to n'}\subset \mc{C}(\hat{\mbf{G}},\hat{\gamma})$.  Conversely, if $\mbf{P}\in  \mc{C}(\hat{\mbf{G}},\hat{\gamma})$, then
 \begin{align}
     \gamma-G_{d,x}\geq \sum_{x',y}\hat{G}_{y,x'}P(y|x')&=\sum_y\hat{G}_{y,x}P(y|x)+\sum_{x'\not=x,y}\hat{G}_{y,x'}P(y|x')\notag\\
     &=\sum_{y=d+1}^{n'}(G_{y,x}-G_{d,x})P(y|x)+\sum_{x'\not=x,y}\hat{G}_{y,x'}P(y|x')\notag\\
     &=-G_{d,x}(1-\sum_{y=1}^{d}P(y|x))+\sum_{y=d+1}^{n'}G_{y,x}P(y|x)+\sum_{x'\not=x,y}\hat{G}_{y,x'}P(y|x')\notag\\
     &\geq -G_{d,x}+\sum_{y=1}^{n'}G_{y,x}P(y|x)+\sum_{x'\not=x,y}\hat{G}_{y,x'}P(y|x')\notag\\
     &=-G_{d,x}+\sum_{x',y}G_{y,x'}P(y|x').
 \end{align}
 Therefore, $\mbf{P}\in\mc{C}(\mbf{G},\gamma)$ and so $\mc{C}_d^{n\to n'}\subset \mc{C}(\hat{\mbf{G}},\hat{\gamma})\subset\mc{C}(\mbf{G},\gamma)$.  Note that if $\mbf{G}$ has only non-negative elements then so will $\hat{\mbf{G}}$.
 
 \end{proof}

\begin{lemma}
 \label{lem:nonzero-column-limit}
 For any finite number of inputs $n$,
 \begin{equation}
 \mc{C}_2^{n\to 4}=\bigcap_{m=1}^s\mc{C}(\mbf{G}_m,\gamma_m)
 \end{equation}
 with each $\mbf{G}_m$ having at most six nonzero columns.
 \end{lemma}
 \begin{proof}
 As a consequence of Lemma \ref{Lem:reduce}, we can always find a complete set of polyhedra $\{\mc{C}(\hat{\mbf{G}}_m,\hat{\gamma}_m)\}_{m=1}^s$ such that
 \[\mc{C}_2^{n\to 4}=\bigcap_{m=1}^s\mc{C}(\hat{\mbf{G}}_m,\hat{\gamma}_m)\]
 such that each $\hat{\mbf{G}}_m$ has no more than positive elements in each column and the rest being zero.  Our goal is to show that the number of such columns can be reduced to six.  The key steps in our reduction are given by the following two propositions.
 
   \begin{proposition}
 \label{Prop:reduction-1}
     Consider the matrices 
\begin{align}
\hat{\mbf{G}}= \begin{bmatrix} a&b&\cdot&\cdots\\c&d&\cdot&\cdots\\
     0&0&\cdot&\cdots\\
     0&0&\cdot&\cdots\end{bmatrix},\qquad   \hat{\mbf{G}}'= \begin{bmatrix} a-c&b+c&\cdot&\cdots\\0&d+c&\cdot&\cdots\\
     0&0&\cdot&\cdots\\
     0&0&\cdot&\cdots\end{bmatrix},\qquad a\geq c\geq 0,
\end{align}
which differ only in the first two columns.  Then $\mc{C}_{2}^{n\to 4}\in \mc{C}(\hat{\mbf{G}},\hat{\gamma})$ \textit{iff} $\mc{C}_{2}^{n\to 4}\in \mc{C}(\hat{\mbf{G}}',\hat{\gamma})$.
 \end{proposition}
 \begin{proof}
 Every vertex $\mbf{V}$ of $\mc{C}_{2}^{n\to 4}$ will have support in only two rows.  If $\mbf{V}$ has support in the first two rows, then its upper left corner will have one of the forms $\left(\begin{smallmatrix}1&1\\0&0\end{smallmatrix}\right)$, $\left(\begin{smallmatrix}1&0\\0&1\end{smallmatrix}\right)$, $\left(\begin{smallmatrix}0&1\\1&0\end{smallmatrix}\right)$, $\left(\begin{smallmatrix}0&0\\1&1\end{smallmatrix}\right)$.  In each of these cases, $\langle \hat{\mbf{G}},\mbf{V}\rangle\leq\hat{\gamma}$ $\Leftrightarrow$ $\langle \hat{\mbf{G}}',\mbf{V}\rangle\leq\hat{\gamma}$.  
 
 The other possibility is that $\mbf{V}$ has support in only one of the first two rows.  This leads to upper left corners of the form $\left(\begin{smallmatrix}1&1\\0&0\end{smallmatrix}\right)$, $\left(\begin{smallmatrix}0&0\\1&1\end{smallmatrix}\right)$, $\left(\begin{smallmatrix}1&0\\0&0\end{smallmatrix}\right)$, $\left(\begin{smallmatrix}0&1\\0&0\end{smallmatrix}\right)$, $\left(\begin{smallmatrix}0&0\\1&0\end{smallmatrix}\right)$, $\left(\begin{smallmatrix}0&0\\0&1\end{smallmatrix}\right)$.  Suppose now that $\mc{C}_{2}^{n\to 4}\in \mc{C}(\hat{\mbf{G}},\hat{\gamma})$.  If a vertex $\mbf{V}$ of $\mc{C}_{2}^{n\to 4}$ has form $\left(\begin{smallmatrix}1&0\\0&0\end{smallmatrix}\right)$ in the upper left corner, non-negativity of $c$ implies that $\langle \hat{\mbf{G}}',\mbf{V}\rangle\leq\hat{\gamma}$.  A somewhat less trivial case is any vertex $\mbf{V}_1$ having form $\left(\begin{smallmatrix}0&1\\0&0\end{smallmatrix}\right)$ in the upper left corner.  Here we need to use the fact that there exists a vertex $\mbf{V}_2$ with $\left(\begin{smallmatrix}1&1\\0&0\end{smallmatrix}\right)$ in the upper left corner but is identical to $\mbf{V}_1$ in all other columns.  Hence we have
 
\begin{align}
    \hat{\gamma}&\geq \langle \hat{\mbf{G}},\mbf{V}_2\rangle= a+b+\kappa\quad\Rightarrow\quad \langle \hat{\mbf{G}}',\mbf{V}_1\rangle=b+c+\kappa\leq a+b+\kappa\leq\hat{\gamma},
\end{align}

\noindent where $\kappa$ is the contribution of the other columns to the inner product, and we have used the assumption that $a\geq c$.  Similar reasoning shows that $\langle \hat{\mbf{G}}',\mbf{V}\rangle\leq\hat{\gamma}$ for all other vertices $\mbf{V}$.  Conversely, by an analogous case-by-case consideration, we can establish that $\mc{C}_{2}^{n\to 4}\in\mc{C}(\hat{\mbf{G}}',\hat{\gamma})$ implies  $\langle \hat{\mbf{G}},\mbf{V}\rangle\leq\hat{\gamma}$ for all vertices $\mbf{V}$ of $\mc{C}_{2}^{n\to 4}$.

\end{proof}

\begin{proposition}
 \label{Prop:reduction-2}
 Consider the matrices
 \begin{align}
     \hat{\mbf{G}}= \begin{bmatrix} a&b&\cdot&\cdots\\0&0&\cdot&\cdots\\
     0&0&\cdot&\cdots\\
     0&0&\cdot&\cdots\end{bmatrix},\qquad   \hat{\mbf{G}}'= \begin{bmatrix} a+b&0&\cdot&\cdots\\0&0&\cdot&\cdots\\
     0&0&\cdot&\cdots\\
     0&0&\cdot&\cdots\end{bmatrix},\qquad\hat{\mbf{G}}''= \begin{bmatrix} 0&a+b&\cdot&\cdots\\0&0&\cdot&\cdots\\
     0&0&\cdot&\cdots\\
     0&0&\cdot&\cdots\end{bmatrix}.
 \end{align}
 which differ only in the first two columns.  Then $\mc{C}_{2}^{n\to 4}\in \mc{C}(\hat{\mbf{G}},\hat{\gamma})$ \textit{iff} $\mc{C}_{2}^{n\to 4}\in \mc{C}(\hat{\mbf{G}}',\hat{\gamma})\cap\mc{C}(\hat{\mbf{G}}'',\hat{\gamma})$.
 \end{proposition}
 \begin{proof}
 This proof considers the vertices of $\mc{C}_2^{n\to 4}$ and applies the same reasoning as the proof of Proposition \ref{Prop:reduction-1}.
 \end{proof}
 
Continuing with the proof of Lemma \ref{lem:nonzero-column-limit}, suppose that $\mc{C}_2^{n\to 4}\in\mc{C}(\hat{\mbf{G}}_m,\hat{\gamma}_m)$ with each column of $\hat{\mbf{G}}_m$ having no more than two nonzero rows.   We can group the columns into six groups according to which two rows have zero (it may be that a column has more than two zeros, in which case we just select one group to place it in).  By repeatedly applying Proposition \ref{Prop:reduction-1}, we can replace $\hat{\mbf{G}}_m$ with a matrix $\hat{\mbf{G}}'_m$ such that each group has at most one column with two nonzero elements; the rest of the columns in that group have at most just one nonzero element.  We then repeatedly apply Proposition \ref{Prop:reduction-2} to remove multiple columns with the same single nonzero row.  In the end, we arrive at the following:
 \begin{equation}
     \mc{C}_{2}^{n\to 4}\in \mc{C}(\hat{\mbf{G}}_m,\hat{\gamma}_m) \qquad \Leftrightarrow \qquad  \mc{C}_{2}^{n\to 4}\in \bigcap_{j}\mc{C}(\hat{\mbf{G}}_{m,j},\hat{\gamma}_m),
 \end{equation}
 where each $\hat{\mbf{G}}_{m,j}$ has at most ten nonzero columns corresponding to the different ways that no more than two nonzero elements can occupy a column.  That is, up to a permutation of columns, each $\hat{\mbf{G}}_{m,j}$ will have the form
 \begin{equation}
     \hat{\mbf{G}}_{m,j}=\begin{bmatrix}
        a_1 & b_1 & c_1 & 0 & 0 & 0 & g & 0 & 0 & 0 & 0 & \cdots \\
        a_2 & 0 & 0 & d_1 & e_1 & 0 & 0 & h & 0 & 0 & 0 & \cdots \\
        0 & b_2 & 0 & d_2 & 0 & f_1 & 0 & 0 & i & 0 & 0 & \cdots \\
        0 & 0 & c_2 & 0 & e_2 & f_2 & 0 & 0 & 0 & j & 0 & \cdots \\
    \end{bmatrix}.
 \end{equation}
 The final step is to remove the block of diagonal elements $[g,h,i,j]$.  To do this, observe that we absorb any of these diagonal elements into an earlier column, provided that the row contains the largest element in that column.  For example, if $f_2> f_1$, then we can replace $\hat{\mbf{G}}_{m,j}$ with
 \begin{align}
     \hat{\mbf{G}}_{m,j}'=\begin{bmatrix}a_1&b_1&c_1&0&0&0&g&0&0&0&0&\cdots\\a_2&0&0&d_1&e_1&0&0&h&0&0&0&\cdots\\0&b_2&0&d_2&0&f_1&0&0&i&0&0&\cdots\\0&0&c_2&0&e_2&f_2+j&0&0&0&0&0&\cdots\end{bmatrix},
 \end{align}
 and we can easily see that $\mc{C}_{2}^{n\to 4}\in \mc{C}(\hat{\mbf{G}}_{m,j},\hat{\gamma}_m)$ \textit{iff} $\mc{C}_{2}^{n\to 4}\in \mc{C}(\hat{\mbf{G}}'_{m,j},\hat{\gamma}_m)$.  By considering the maximum element in each of the first six columns, we can perform this replacement for at least three of the four elements $[g,h,i,j]$.  If we can do this for all four elements, then the proof is complete.  On the other hand, if we can only remove three of these elements, then we will obtain a matrix $\hat{\mbf{G}}''_{m,j}$ of the form (up to row/column permutations)
  \begin{equation}
     \hat{\mbf{G}}''_{m,j}=\begin{bmatrix}a_1&b_1&c_1&0&0&0&g&0&0&0&0&\cdots\\a_2&0&0&d_1&e_1&0&0&0&0&0&0&\cdots\\0&b_2&0&d_2&0&f_1&0&0&0&0&0&\cdots\\0&0&c_2&0&e_2&f_2&0&0&0&0&0&\cdots\end{bmatrix}
 \end{equation}
 with $a_1,b_1,c_2$ not having the largest values in their respective columns.  In this case, we construct the matrix
   \begin{equation}
     \hat{\mbf{G}}'''_{m,j}=\begin{bmatrix}a_1+g&b_1+g&c_1+g&0&0&0&0&0&0&0&0&\cdots\\a_2+g&0&0&d_1&e_1&0&0&0&0&0&0&\cdots\\0&b_2+g&0&d_2&0&f_1&0&0&0&0&0&\cdots\\0&0&c_2+g&0&e_2&f_2&0&0&0&0&0&\cdots\end{bmatrix},
 \end{equation}
 from which it can be verified that $\mc{C}_2^{n\to 4}\subset \mc{C}(\hat{\mbf{G}}''_{m,j},\hat{\gamma}_m)$ \textit{iff} $\mc{C}_2^{n\to 4}\subset \mc{C}(\hat{\mbf{G}}'''_{m,j},\hat{\gamma}_m+2g)$.
 
 \end{proof}

\section{Proof of Theorem \ref{thm-replacer-bounds}} \label{Appendix-replacer}

In this section we provide two propositions that support the proof of Theorem \ref{thm-replacer-bounds}. Recall that a $d$-dimensional partial replacer channel is a quantum channel having the form

\begin{equation}
\label{Eq:replacer-channel-appendix}
    \mc{R}_\mu(X)=\mu X+(1-\mu)\tr[X]\sigma,
\end{equation}

\noindent where $1 \geq \mu \geq 0$, $\sigma$ is some fixed density matrix, and $X$ is a quantum state on a $d$-dimensional Hilbert space.  Note that the partial erasure channel $\mc{E}_\mu$ corresponds to $\sigma$ being an erasure flag $\op{E}{E}$, where $\ket{E}$ is orthogonal to $\{\ket{1},\cdots,\ket{d}\}$.
We first show that the lower bound of $\kappa(\mc{R}_\mu) \geq \lceil \mu d + (1-\mu)\rceil$ (see Eq. \eqref{Eq:replacer-bounds}) is not improved by any choice of states $\{\rho_x\}_x$, POVM $\{\Pi_y\}_y$, or ambiguous guessing game $\mbf{G}_{n',d}^{n,n'}$ with $k=n'$.

\begin{proposition} \label{prop-replacer-ml-bound}
    The maximum likelihood score for any classical channel $\mbf{P}_{\mc{R}_\mu}$ generated using a partial replacer channel $\mc{R}_\mu$ is bounded as
    
    \begin{equation}
        \langle\mbf{G}_{\ML},\mbf{P}_{\mc{R}_\mu}\rangle \leq \mu d + (1- \mu)
    \end{equation}
    
    \noindent where $\mbf{G}_{\ML}$ is any maximum likelihood facet satisfying Proposition \ref{prop-ambiguous-guessing-facets}(i).
    
    \begin{proof}
        In this proof, we first consider the unlifted maximum likelihood $\mbf{G}_{\text{ML}}^{n'} = \mbb{I}_{n'}$ where $n=n'$ (see Appendix \ref{Appendix-ml-facets}), and then generalize across all input/output liftings taking $\mbf{G}^{n'}_{\text{ML}}\to \mbf{G}_{\text{ML}}\in\mbb{R}^{m'\times m}$ where $m',m \geq n'$.
        To begin we maximize $\langle\mbf{G}^{n'}_{\text{ML}}, \mbf{P}_{\mc{R}_\mu}\rangle$ over the quantum states $\{\rho_x\}_x$ and POVM $\{\Pi_y\}_y$,
        
        \begin{align}
            \max\;\langle\mbf{G}^{n'}_{\text{ML}}, \mbf{P}_{\mc{R}_\mu}\rangle & = \max_{\{\rho_x\}_x,\{\Pi_y\}_y}\;\sum_{x=y}\tr\Big[\Pi_y\mc{R}_\mu\big(\rho_x\big) \Big]\\
            & = \max_{\{\rho_x\}_x,\{\Pi_y\}_y}\;\sum_{x=y} \mu \tr\Big[\Pi_y \rho_x\Big] + (1-\mu) \tr\Big[\Pi_y \sigma\Big] \\
            & \leq \max_{\{\Pi_y\}_y}\;\sum_{y} \mu \tr\Big[\Pi_y\Big] + (1-\mu) \tr\Big[\Pi_y \sigma\Big] \label{Eq:replacer-ml-lower-bound-proof-expanded}\\
            & = \mu d + (1 - \mu),\label{Eq:replacer-ml-lower-bound-proof}
        \end{align}
        
        \noindent where line \eqref{Eq:replacer-ml-lower-bound-proof-expanded} uses the fact that $\tr[\Pi_y \rho_x] \leq \tr[\Pi_y ]$ for any choice of $\Pi_y$ and $\rho_x$ while
        line \eqref{Eq:replacer-ml-lower-bound-proof} results from $\sum_y \tr[\Pi_y] = d$ and $\sum_y \tr[\Pi_y \sigma] = \tr[\sigma] = 1$.
        A simple example that achieves this bound is the scenario where Alice sends orthogonal states $\{\op{x}{x}\}_{x=1}^d$
        and Bob measures with a similar POVM $\{\op{y}{y}\}_{y=1}^d $, then
        
        \begin{align}
            \langle\mbf{G}^{d}_{\text{ML}},\mbf{P}_{\mc{R}_\mu}\rangle & = \sum_{y=x=1}^d \tr\Big[\op{y}{y}\mc{R}_\mu\big(\op{x}{x}\big) \Big]\\
            & = \mu \sum_{y=1}^d \tr\Big[\ket{y}\ip{y}{y}\bra{y} \Big] + (1-\mu)\sum_{y=1}^d \tr\Big[\op{y}{y}\sigma\Big] \\
            & = \mu d+(1-\mu).
        \end{align}
        
        \noindent In general, the upper bound is achieved whenever $\Pi_y \rho_x = \Pi_y$ for all $x\in[n]$ and $y\in[n']$.
        Note that this requires $\rank(\Pi_y)=\rank(\rho_x)=1$ and $\Pi_y || \rho_x$.
        
        To extend the bound $\langle\mbf{G}^{n'}_{\text{ML}}, \mbf{P}_{\mc{R}_\mu}\rangle \leq \mu d + (1-\mu)$ to all liftings of $\mbf{G}_{\text{ML}}^{n'}$, we make two observations.
        First, note that the input lifting taking $\mbf{G}_{\text{ML}}^{n'}\to \mbf{G}_{\text{ML}}^{\prime \prime }\in\mbb{R}^{n'\times m}$ contains $(m -n)$ all-zero columns.
        These all-zero columns of $\mbf{G}_{\text{ML}}^{\prime \prime }$ do not contribute to the inner product $\langle \mbf{G}_{\text{ML}}^{\prime \prime},\mbf{P}_{\mc{R}_\mu}\rangle$, and therefore, cannot increase the inner product beyond $\mu d + (1-\mu)$.
        Second, observe that the output lifting taking $\mbf{G}_{\text{ML}}^{n'}\to \mbf{G}_{\text{ML}}^{\prime}\in\mbb{R}^{(n'+1)\times n}$ requires a new POVM $\{\Pi'_{y}\}_{y=1}^{n'+1}$ which must satisfy $\sum_{y=1}^{n'+1} \Pi'_y = \mbb{I}_d$.
        Furthermore, one column $x$ of $\mbf{G}_{\text{ML}}^{\prime}$ has two nonzero elements in rows $y$ and $y'$ where $G'_{y,x}=G'_{y',x}=1$.
        In this case, two POVM elements $\Pi'_y$ and $\Pi'_{y'}$ are both optimized against the state $\rho_x$.
        However, the constraint $\tr[\Pi'_{y'}\rho_x]+\tr[\Pi'_y\rho_x] \leq 1$ holds for any choice of $\rho_x$ and POVM.
        Therefore, the inner product $\langle \mbf{G}_{\text{ML}}^{\prime},\mbf{P}_{\mc{R}_\mu}\rangle \leq \mu d + (1-\mu)$.
        The argument applied for the output lifting holds in general where one or more columns $x$ contain at least two non-zero elements.
        Thus, the upper bound in Eq. \eqref{Eq:replacer-ml-lower-bound-proof} holds for any input/output lifting taking $\mbf{G}_{\text{ML}}^{n'} \to \mbf{G}_{\text{ML}}\in\mbb{R}^{m'\times m}$ where $\min\{m,m'\}\geq n'$. This concludes the proof.
    \end{proof}
\end{proposition}

The upper bound on the maximum likelihood score from Proposition \ref{prop-replacer-ml-bound} serves as a lower bound on the signaling dimension of the partial replacer channel $\kappa(\mc{R}_\mu)$.
This follows from the fact that if $\mbf{P}_{\mc{R}_\mu}\notin \mc{M}_r^{n \to n'}$, then $\kappa(\mc{R}_\mu) > r$
Furthermore, the integer nature of the signaling dimension implies that $\kappa(\mc{R}_\mu) \geq \lceil \mu d + (1-\mu) \rceil$.
We now turn to certify the signaling dimension of the partial erasure channel.

\begin{proposition} \label{prop-signaling-dimension-of-erasure-channel}
    The signaling dimension of of a $d$-dimensional partial erasure channel is,
    
    \begin{equation} \label{Eq:erasure-signaling}
        \kappa(\mc{E}_\mu) =  \min\{d,\lceil \mu d+1 \rceil\}.
    \end{equation}
    
    \begin{proof}
        Let the classical channel $\mbf{P}_{\mc{E}_\mu}$ be induced by the partial erasure channel $\mc{E}_\mu$ via Eq. \eqref{Eq:channel-induce} for any collection of quantum states $\{\rho_x\}_{x}$ and POVM $\{\Pi_y\}_{y}$.
        The transition probabilities are then expressed
        
        \begin{equation} \label{Eq:erasure-channel-correlations}
            P_{\mc{E}_\mu}(y|x)=\mu P_{\text{id}_d}(y|x)+(1-\mu)P_{\ket{E}}(y),
        \end{equation}
        
        \noindent where $P_{\text{id}_d}(y|x) = \tr[\Pi_y \rho_x]$ and $P_{\ket{E}}(y) = \tr[\Pi_y \op{E}{E}]$.
        Since the simulation protocol for partial replacer channels can faithfully simulate $\mbf{P}_{\mc{E}_\mu}$, the upper bound $\kappa(\mathcal{R}_\mu) \leq \lceil \mu d + 1\rceil$ holds (see the proof of Theorem \ref{thm-replacer-bounds} in the main text).  Therefore, $\min\{d,\lceil \mu d + 1 \rceil \} \geq \kappa(\mc{E}_\mu)$.
        To establish a lower bound on $\kappa(\mc{E}_\mu)$ we consider the channel $\mbf{P}_{\mc{E}_\mu}\in\mc{P}^{d \to (d+1)}$ generated by the scenario where Alice sends the computational basis states $\{\op{x}{x}\}_{i=1}^d$ and Bob measures with the POVM $\{ \op{y}{y}\}_{y=1}^{d+1}$ where $\ket{d+1}=\ket{E}$,
        
        \begin{align}
            \mbf{P}_{\mc{E}_\mu} & = \sum_{x=1}^d\sum_{y=1}^{d+1} \tr\Big[\op{y}{y}\mc{E}_\mu(\op{x}{x})\Big]\op{y}{x} \\
            & = \sum_{x=1}^d\sum_{y=1}^{d+1}\left( \mu \tr\Big[\ket{y}\ip{y}{x}\bra{x}\Big] + (1-\mu)\tr\Big[\ket{y}\ip{y}{E}\bra{E}\Big]\right)\op{y}{x} \\
            & = \mu\sum_{x=1}^d\op{x}{x} + (1-\mu)\sum_{x=1}^d \op{E}{x}\label{Eq:erasure-channel-probabilities}.
        \end{align}
        
        \noindent As demonstrated in Proposition \ref{prop-replacer-ml-bound},  $\mbf{P}_{\mc{E}_\mu}$ achieves the maximum likelihood upper bound for partial replacer channels, $\langle\mbf{G}_{\text{ML}},\mbf{P}_{\mc{E}_\mu} \rangle = \mu d + (1-\mu)$.
        In fact, this bound also holds for non-orthogonal quantum states  $\{\rho_x\}_{x\in[n]}$ where $n> d$.
        
        To improve the lower bound on $\kappa(\mc{E}_\mu)$ beyond Proposition \ref{prop-replacer-ml-bound}, we consider the ambiguous polytope $\mc{A}_{(n'-1),r}^{(n'-1)\to n'}$ with ambiguous guessing facets $\mbf{G}_?^{n',r}$ that are tight Bell inequalities of $\mc{C}_r^{n \to n'}$ (see Appendix \ref{Appendix-ambiguous-guessing-facets}). 
        Our goal is to find the smallest integer $r$ such that $\mbf{P}_{\mc{E}_\mu}\in\mc{A}_{(n'-1),r}^{(n'-1) \to n'}$, that is, $\langle\mbf{G}_?^{n',r},\mbf{P}_{\mc{E}_\mu}\rangle\leq r(n'-r)$ is satisfied.
        Consider the erasure channel $\mbf{P}_{\mc{E}_\mu}\in\mc{P}^{d\to (d+1)}$ described by Eq. \eqref{Eq:erasure-channel-probabilities}. We find that the inequality
        
        \begin{align} \label{Eq:erasure-inequality}
            r(n'-r) \geq \langle \mbf{G}_?^{n',r},\mbf{P}_{\mc{E}_\mu}\rangle  =  (n'-r)\mu d + (1-\mu)(n'-1),
        \end{align}
        
        \noindent is violated if $\mbf{P}_{\mc{E}_\mu}\notin\mc{A}_{(n'-1),r}^{(n'-1) \to n'}$ for $n'-2 \geq r \geq 2$.
        Note that in our example $n'=d+1$, however, this procedure holds for any $n' > d$.
        Rearranging inequality \eqref{Eq:erasure-inequality}  into the form,
        
        \begin{equation} \label{Eq:erasure-quadratic}
            0 \geq r^2 - r(\mu d + n') + \mu d n' + (1 - \mu)(n'-1),
        \end{equation}
        
        \noindent allows us to find the values of $r$ for which inequality \eqref{Eq:erasure-inequality} is satisfied by solving for the zeros $r_\pm$ of the quadratic on the RHS of Eq. \eqref{Eq:erasure-quadratic},
        
        \begin{equation}\label{Eq:erasure-zeros}
            r_{\pm} = \frac{1}{2}(\mu d + n') \pm \frac{1}{2}\sqrt{(n'-\mu d)^2 - 4(1-\mu)(n'-1)}.
        \end{equation}

        \noindent Since the parabola of Eq. \eqref{Eq:erasure-quadratic} is concave up, all integer values of $r\in[r_-,r_+]$ satisfy inequality \eqref{Eq:erasure-inequality}.
        Furthermore, the smallest integer for which the inequality is satisfied is $r = \lceil r_{-} \rceil$.
        Therefore, the signaling dimension is bounded as
        
        \begin{equation} \label{Eq:erasure-signaling-dim-lower-bound}
            \kappa(\mc{E}_\mu) \geq r = \left\lceil \frac{1}{2}(\mu d + n') - \frac{1}{2}\sqrt{(n'-\mu d)^2 - 4(1-\mu)(n'-1)} \right\rceil.
        \end{equation}
        
        \noindent The value of $r$ in Eq. \eqref{Eq:erasure-signaling-dim-lower-bound} satisfies the facet inequality \eqref{Eq:erasure-inequality} for all allowed values of $n'$, $\mu$, and $d$.
        Note that $n'$ is a free parameter which we can choose as any integer $n' \geq d+1$.
        In our example, $\mbf{P}_{\mc{E}_\mu}$ has $n'=(d+1)$ which obtains the lower bound $\kappa(\mc{E}_\mu) \geq \lceil\mu d + 1\rceil$.
        To see this, we substitute $n'= (d+1)$ into Eq. \eqref{Eq:erasure-signaling-dim-lower-bound} and perform some algebra,
        
        \begin{align}
            r =& \:\biggl\lceil\frac{1}{2} \Big(\mu d + d + 1\Big) \notag \quad- \frac{1}{2}\sqrt{\Big(d(1-\mu) +1\Big)^2 - 4d(1-\mu)} \biggr\rceil \\
            =& \left\lceil \frac{1}{2} \Big(\mu d + d + 1\Big) - \frac{1}{2}\sqrt{\Big(d(1-\mu) -1\Big)^2}  \right\rceil \\
            =& \left\lceil \frac{1}{2} \Big(\mu d + d + 1\Big) - \frac{1}{2}\Big(d(1-\mu) -1\Big) \right\rceil \\
            =& \left\lceil \frac{1}{2} (\mu d + 1) + \frac{1}{2}(\mu d +1) \right\rceil \\
            =& \;\lceil\mu d + 1 \rceil.
        \end{align}
        
        \noindent Hence $\kappa(\mc{E}_\mu) \geq r = \lceil \mu d + 1 \rceil$.
        Additionally, substituting $n' > d + 1$ into Eq. \eqref{Eq:erasure-signaling-dim-lower-bound} results in a necessarily smaller value or $r$ therefore $n' = d+1$ is a maximum.
        
        It is important to note that the lower bound $\kappa(\mc{E}_\mu) \geq \lceil \mu d + 1\rceil$ only holds for $r \leq n'-2$ because $\mbf{G}^{n',r}_?$ is not a facet for signaling polytopes $\mc{C}_r^{(n'-1)\to n'}$ with $r > n'-2$.
        Therefore, we must consider the edge case where $r=n'-1=n$, that is, the case where the trivial upper bound of Eq. \eqref{Eq:classical-cardinality-bound} is obtained.
        From Theorem \ref{Thm:main1} Condition (ii) we know that $\mc{C}_{n-1}^{n\to n'}=\cap{k=n}^{n'}\mc{A}_{k,n-1}^{n \to n'}$.
        It follows for the edge case $r=n'-1=n$ that if a channel $\mbf{P}_{\mc{E}_{\mu}}\notin\mc{A}_{(n'-1),r}^{(n'-1)\to n'}$, then $\kappa^{n\to n'}(\mc{E}_\mu)=\min\{n,n'\}=(n'-1)$.
        Hence $\kappa(\mc{E}_\mu)$ is proven to be tight with the upper bound.
        To illustrate this case we consider inequality \eqref{Eq:erasure-quadratic} and substitute $r= n'-2$, 
        
        \begin{align}
            0 & \geq (n'^2 - 4 n' + 4) - (n'-2)(\mu d + n') + \mu d n' + (1-\mu)(n'-1) \\
            & \geq -4 n' + 4 + 2 \mu d  + 2 n' + (1-\mu) n' - (1-\mu)\label{Eq:erasure-edge-case-1}.
        \end{align}
        
        \noindent Next, we substitute $d=n'-1$ into Eq. \eqref{Eq:erasure-edge-case-1} as this is the edge case we wish to consider,
        
        \begin{align}
            0 & \geq -4 n' + 4 + 2 \mu (n'-1)  + 2 n' + (1-\mu) n' - (1-\mu) \\
            & \geq n'(\mu - 1) + 3 - \mu \\
            & \geq 3 - n' + \mu(n' - 1) \label{Eq:erasure-edge-case}.
        \end{align}
        
        \noindent Rearranging inequality \eqref{Eq:erasure-edge-case}, we find that it is satisfied \textit{iff}, $\frac{n'-3}{n'-1} \geq \mu$.
        Therefore, when $\mu > \frac{n'-3}{n'-1}$, inequality \eqref{Eq:erasure-inequality} is violated and, by Theorem \ref{Thm:main1}(ii) we certify that $\kappa(\mc{E}_\mu) = d$.
        Considering this edge case, we arrive at the conclusion that $\kappa(\mc{E}_\mu) \geq \min\{d, \lceil \mu d + 1\rceil\}$ which is exactly the upper bound $\min\{d, \lceil \mu d + 1\rceil \}\geq \kappa(\mc{E}_\mu)$.
        That is, the signaling dimension of the erasure channel is bounded tightly from above and below from which it follows, $\kappa(\mc{E}_\mu) = \min\{d,\lceil \mu d + 1\rceil\}$.
    \end{proof}
\end{proposition}

\end{document}